\documentclass[aps,prx]{revtex4}
\usepackage{graphicx}
\usepackage{amssymb}
\usepackage{bm}
\usepackage{xcolor}
\usepackage{physics}
\usepackage{mathtools}
\usepackage{comment}
\usepackage{revsymb}
\usepackage{amsthm}
\usepackage{amsmath,latexsym}
\usepackage{hyperref}

\def\im{\mathop{\rm Im}}

\newcommand{\id}{\operatorname{id}}
\renewcommand{\Tr}{\operatorname{Tr}}
\renewcommand{\tr}{\operatorname{Tr}}
\newcommand{\OL}[1]{\overline{#1}}

\newtheorem{theorem}{Theorem}
\newtheorem{proposition}[theorem]{Proposition} 
\newtheorem{corollary}[theorem]{Corollary}
\newtheorem{remark}[theorem]{Remark}
\newtheorem{lemma}[theorem]{Lemma}

\newtheorem{definition}[theorem]{Definition}
\newtheorem{example}[theorem]{Example}

\newcommand{\fs}[1]{#1}

\begin{document}
\title{Usefulness of adaptive strategies in asymptotic quantum channel discrimination}
\author{Farzin~Salek}
\email{farzin.salek@gmail.com}
\affiliation{Zentrum Mathematik, 
Technische Universit\"{a}t M\"{u}nchen,
%Technical University of Munich,
85748 M\"{u}nchen, Germany}
%\affiliation{Departamento de Teor\'{i}a de la Se\~{n}al y Comunicaciones (TSC), Universitat Polit\`{e}cnica de Catalunya, 08034 Barcelona, Spain}
%\affiliation{Grup d'Informaci\'{o} Qu\`{a}ntica, Departament de F\'{\i}sica, Universitat Aut\`{o}noma de Barcelona, 08193 Bellaterra (Barcelona), Spain}
\author{Masahito Hayashi}
\email{hayashi@sustech.edu.cn}
\affiliation{Shenzhen Institute for Quantum Science and Engineering, Southern University of Science and Technology, Nanshan District, Shenzhen, 518055, China}
\affiliation{International Quantum Academy (SIQA), Futian District, Shenzhen 518048, China}
\affiliation{Graduate School of Mathematics, Nagoya University, Nagoya, 464-8602, Japan}
\author{Andreas~Winter}
\email{andreas.winter@uab.cat}
\affiliation{Instituci\'o Catalana de Recerca i Estudis Avan\c{c}ats (ICREA), Pg.~Lluis Companys, 23, 08010 Barcelona, Spain}
\affiliation{Grup d'Informaci\'{o} Qu\`{a}ntica, Departament de F\'{\i}sica, Universitat Aut\`{o}noma de Barcelona, 08193 Bellaterra (Barcelona), Spain}

\begin{abstract}
Adaptiveness is a key principle in information processing including
statistics and machine learning.
We investigate the usefulness adaptive methods in the framework of asymptotic binary hypothesis testing, 
when each hypothesis represents asymptotically many independent instances 
of a quantum channel, and the tests are based on using the unknown channel and observing outputs.
Unlike the familiar setting of quantum states as hypotheses, there is a 
fundamental distinction between adaptive and non-adaptive strategies with respect 
to the channel uses, and we introduce a number of further variants of 
the discrimination tasks by imposing different restrictions on the 
test strategies. 

%\aw{We have so many results, they are so beautiful, you will love them. 
%The fake science is going to get mad at us, but we have them, those 
%results, they're so big, they are the best. And we are going to deploy them. 
%Let us go forward for the purity of essence. God bless America
%and the integrity of our bodily fluids.}
%
The following results are obtained: 
(1) We prove that for classical-quantum channels, adaptive and non-adaptive 
strategies lead to the same error exponents both in the symmetric (Chernoff)
and asymmetric (Hoeffding, Stein) settings.  
(2) The first separation between adaptive and non-adaptive 
symmetric hypothesis testing exponents for quantum channels, which we 
derive from a general lower bound on the error probability for non-adaptive 
strategies; the concrete example we analyze is a pair of entanglement-breaking 
channels.
(3) We prove, in some sense generalizing the previous statement, that for 
general channels adaptive strategies restricted to classical feed-forward
and product state channel inputs are not superior in the asymptotic limit
to non-adaptive product state strategies.
%
%(4) As an application of our findings, we address the discrimination power of 
%quantum channels \aw{and show that adaptive strategies with classical feedback and
%input quantum memory may increase the discrimination power of entanglement-breaking channel.
%[No, we show no such separation; the theorems in Section VII are all equalities.]}
(4) As an application of our findings, we address the discrimination power of 
an arbitrary quantum channel and show that adaptive strategies with classical 
feedback and no quantum memory at the input do not increase the discrimination 
power of the channel beyond non-adaptive tensor product input strategies.
\end{abstract}

\maketitle

\section{Introduction}
\label{sec:intro}
Adaptiveness is a key principle in information processing including
statistics and machine learning \cite{PhysRevLett.126.190505}, 
which can entail great advantage over non-adaptive methods. 
Because of the higher complexity of adaptive methods, we thus are
motivated to clarify in which situations they offer a significant improvement.
Here, we address this question in the setting of binary hypothesis testing for
quantum channels.
Hypothesis testing is one of the most fundamental primitives both in classical and quantum information processing
because a variety of other information processing problems can be 
cast in the framework of hypothesis testing; both direct coding theorems and 
converses can be reduced to it. 
%Adaptive method is a key method in quantum information processing including quantum machine learning \cite{PhysRevLett.126.190505}. 
%It is desired to clarify what situation brings us significant improvement by this method. 
%For this reason, we address this problem by focusing on binary hypothesis testing for quantum channels, which is one of the most fundamental primitives both in classical and quantum information processing. 
%It is such a central task 
It is expected that this analysis for adaptiveness
reveals the role of adaptive methods in various types of quantum information processing. 
In binary hypothesis testing, the two hypotheses are usually referred to as null and 
alternative hypotheses and accordingly, two error probabilities are defined: 
type-I error due to a wrong decision in favour of the alternative hypothesis 
(while the truth corresponds to the null hypothesis) and type-II error due to the 
alternative hypothesis being rejected despite being correct. The overall objective of the 
hypothesis testing is to minimize the error probability in identifying the hypotheses. 
Depending on the significance attributed to the two types of errors, several 
settings can be distinguished. A historical distinction is between the \emph{symmetric} 
and the \emph{asymmetric} hypothesis testing: in symmetric hypothesis testing, 
the goal is to minimize both error probabilities simultaneously, while in asymmetric 
hypothesis testing, the goal is to minimize one type of error probability subject 
to a constraint on the other type of error probability.

In classical information theory, discriminating two distributions has been studied by 
many researchers; Stein, Chernoff \cite{chernoff1952}, Hoeffding \cite{hoeffding1965} 
and Han-Kobayashi \cite{42188} formulated asymptotic hypothesis testing of 
two distributions as optimization problems and subsequently found optimum error exponents. 
As generalizations of these settings to quantum realm, 
discrimination of two quantum states has been studied extensively in quantum 
information theory, albeit the complications stemming from the noncommutativity of 
quantum mechanics appear in the most visible way among these problems. 
The first study in this direction was done by Hiai and Petz \cite{Hiai-Petz},
which showed the possibility part of the quantum extension of Stein's lemma.
That is, it showed that the error exponent of type II error probability
attains the relative entropy registered between the states
under the constant constraint for type I error probability.
Also, it shows the impossibility to exceed the above error exponent
when type I error probability goes to zero.
Subsequently, Ogawa and Nagaoka \cite{887855} strengthened the above 
impossibility, i.e., it showed the same fact
under the constant constraint for type I error probability.
%The reference \cite{Hayashi:book} reported significant progress in the symmetric setting 
%and finally the...
As the quantum extension of the Chernoff bound, 
Audenaert \emph{et al.} \cite{Audenaert_2007}
derived a lower bound for
the exponent of the sum of type I and type II error probabilities,
and Nussbaum and Szko\l{}a \cite{nussbaum2009} showed its tightness.
(see \cite{Hayashi:book} for earlier significant progress). 
Concerning the quantum extension of the 
Hoeffding bound, the paper \cite{Ogawa-Hayashi} derived a lower bound 
of the exponent of the type II error probability 
under the exponential constraint for type I error probability, 
but it suggested the existence of a tighter lower bound. 
Later, \cite{Hayashi_2007} proved the suggested 
tighter lower bound and subsequently, Nagaoka \cite{nagaoka2006converse} showed its
optimality.

To study the effect of adaptiveness in the viewpoint of the binary hypothesis testing,  
we focus on the
discrimination of (quantum) channels, which is a natural extension of the state discrimination problem. 
Channel discrimination is a fundamental question not only in quantum information but also in other disciplines including theoretical 
computer science where, under the name of oracle identification, discrimination of unitary operations as oracles
in quantum algorithms becomes relevant \cite{Chefles_2007}.
Despite inherent mathematical links between the channel and state discrimination problems, 
due to the additional degrees of freedom introduced by the adaptive strategies, 
discrimination of channels is more complicated. 
Many papers have been dedicated to study the potential advantages of adaptive strategies 
over non-adaptive strategies in channel discrimination, such as \cite{PhysRevA.81.032339,KPP:POVM}. 

The seminal classical work \cite{5165184} showed that in the asymptotic regime, the 
exponential error rate for classical channel discrimination cannot be improved by 
adaptive strategies for any of the symmetric or asymmetric settings, i.e.
the channel versions of Stein's lemma, Chernoff bound, and Hoeffing bound.

Since the publication of \cite{5165184}, significant amount of research has focused on 
showing the potential advantages of adaptive strategies in discrimination of quantum channels.
Significant progress was reported in \cite{berta2018amortized} concerning cq-channels, i.e.,
the case when the channel has a classical input and a quantum output.
There are other pairs of channels, for which it could be shown that adaptive 
strategies do not outperform non-adaptive ones for any finite number of 
copies, such as pairs of von Neumann measurements \cite{Puchala:measurement,Lewandowska-et-al:measurement}
and teleportation-covariant channels (which are programmed by their Choi states) \cite{PL:teleport}.
Wilde \emph{et al.} \cite{berta2018amortized} 
showed that the classical-quantum (cq-)channel extension of Stein's lemma
has no improvement by use of adaptive strategy.
However, for Chernoff and Hoeffding bounds,
they derived upper and lower bounds.
These bounds do not coincide when the cq-channel has a certain non-commutativity.
Therefore, it remained an open problem whether an adaptive method improves
Chernoff and Hoeffding bounds in the classical-quantum channel discrimination.

Concerning quantum-quantum (qq-)channels, i.e. channels having quantum input and quantum outputs,
it is known that adaptive strategies 
offer an advantage in the non-asymptotic regime for discrimination in the symmetric Chernoff 
setting \cite{PhysRevA.81.032339,PhysRevLett.103.210501,7541701,cite-key}. 
In particular, Harrow \emph{et al.} \cite{PhysRevA.81.032339} 
demonstrated the advantage of adaptive strategies in discriminating 
a pair of entanglement-breaking channels that requires just two channel evaluations 
to distinguish them perfectly, but such that
no non-adaptive strategy can give perfect distinguishability using any finite 
number of channel evaluations. 
However, it was open whether the same holds in the asymptotic setting. 
 
This question in the asymmetric regime was recently settled by Wang and Wilde:
In \cite[Thm.~3]{PhysRevResearch.1.033169}, they found an exponent 
in Stein's setting for non-adaptive strategies in terms of channel max-relative entropy,
also in the same paper \cite[Thm.~6]{PhysRevResearch.1.033169},
they found an exponent in Stein's setting for the adaptive strategies in terms of 
amortized channel divergence, a quantity introduced in \cite{berta2018amortized}
to quantify the largest distinguishability between two channels.
However, the fact that adaptive strategies 
do not offer an advantage in the setting of Stein's lemma for quantum channels,
i.e. the equality of the aforementioned exponents of Wang and Wilde, was later 
shown in \cite{fang2019chain} via a chain rule for the quantum relative 
entropy proven therein.
Cooney \emph{et al.} \cite{Cooney2016} proved the quantum Stein's lemma for 
discriminating between an arbitrary quantum channel and a ``replacer channel'' 
that discards its input and replaces it with a fixed state. 
This work led to the conclusion that at least in the asymptotic regime, a non-adaptive 
strategy is optimal in the setting of Stein's lemma. However, in the Hoeffding and 
Chernoff settings, the question of potential advantages of adaptive strategies 
involving replacer channels remains open.

Hirche \emph{et al.} \cite{Hirche} studied the maximum power of a fixed quantum
detector, i.e. a POVM, in discriminating two possible states. 
This problem is dual to the state discrimination scenario considered so far in that, 
while in the state discrimination problem the state pair is fixed and optimization is over all measurements, 
in this problem a measurement POVM is fixed and the question is how powerful this discriminator is, 
and then whatever criterion considered for quantifying the power of the given detector, 
it should be optimized over all input states. 
In particular, if $n\geq2$ uses of the detector are available, 
the optimization takes place over all $n$-partite entangled states and also all adaptive strategies 
that may help improve the performance of the measurement.
The main result of \cite{Hirche} states that when asymptotically many uses 
(i.e. $n\rightarrow\infty$) of a given detector is available, its performance does not improve 
by considering general input states or using an adaptive strategy in any of the 
symmetric or asymmetric settings described before.  
%The main ingredient in the present paper is the classical result from \cite{5165184}.
%Namely, it is shown that adaptive processing of the measurement results with 
%general entangled input states can be cast as discriminating two classical channels, 
%which is known not to be improved by adaptive strategies.

%As an application of our cq-channel discrimination theorem, we extend the result of \cite{Hirche} to general quantum channels by considering that the measurement device happens to receive the states it wants to distinguish only after passing through some quantum channel. For this problem, we find the exact expression of the Hoeffding exponent which is given as an optimization over all possible input states. Another application of our main theorem is the discrimination of the quantum channel with classical feed-forward where we show that no adaptive strategy can improve the error exponent of this problem. 

In this paper, we tackle and solve all the aforementioned open problems as we explain next:
(i) We prove that for cq-channels, adaptive and non-adaptive 
strategies lead to the same error exponents both in the symmetric (Chernoff)
and asymmetric (Hoeffding, Stein) settings.  
(ii) We derive the first separation between adaptive and non-adaptive 
symmetric hypothesis testing exponents for qq-channels, which we 
derive from a general lower bound on the error probability for non-adaptive 
strategies. The two concrete examples we analyze are pairs of entanglement-breaking 
channels.
(iii) 
When two qq-channels are given as entanglement-breaking channel with the same measurement,
we prove that 
adaptive and non-adaptive 
strategies lead to the same error exponents both in the symmetric (Chernoff)
and asymmetric (Hoeffding, Stein) settings.  
(iv) As an application of our findings, we address the discrimination power of 
an arbitrary quantum channel and show that adaptive strategies with classical 
feedback and no quantum memory at the input do not increase the discrimination 
power of the channel beyond non-adaptive tensor product input strategies.

The rest of the paper is organized as follows. 
Section \ref{secadaptive} presents our results of cq-channel discrimination
with discrete feedback variables.
Section \ref{new3} gives a general formulation for adaptive discrimination for qq-channels.
In Section \ref{lowerbound} we show 
two examples of qq-channels, of the entanglement breaking form, that have
the first asymptotic separation between 
adaptive and non-adaptive strategies via proving a lower bound on the Chernoff 
error for non-adaptive strategies and analyzing an example where adaptive 
strategies achieve error zero even with two copies of the channels. 
In Section \ref{S-5}, 
we study the discrimination of quantum channels 
when restricting to a subclass of $\mathbb{A}_n$ allowing only strategies with
classical feed-forward and without quantum memory at the input.
Also, Section \ref{S-5} addresses the discrimination of two qq-channels
under a special class of pairs of two qq-channels.
In Section \ref{power} we apply our results to the discrimination power of an 
arbitrary quantum channel. 
We conclude in Section \ref{conclude}. 
Appendices are denoted to prove the results for cq-channel discrimination, which are stated in Section \ref{secadaptive}.

\section{Discrimination of classical-quantum channels}
\label{secadaptive}
In this section, the hypotheses are described by two 
cq-channels.
To spell out the precise questions, let us introduce a bit of notation. 
Throughout the paper, $A$, $B$, $C$, etc, denote quantum systems, 
but also their corresponding Hilbert space. 
A cq-channel is defined with respect to a set $\mathcal{X}$ of input signals and 
the Hilbert space $B$ of the output states. In this case, the channel from $\mathcal{X}$ 
to $B$ is described by the map from the set $\mathcal{X}$ to the set of density 
operators in $B$; as such, a cq-channel is given as 
$\mathcal{N}:x\rightarrow \rho_{x}$, where $\rho_{x}$ denotes the output state 
when the input is $x\in\mathcal{X}$. Our goal is to distinguish between two 
cq-channels, $\mathcal{N}:x\rightarrow \rho_{x}$ and $\overline{\mathcal{N}}:x\rightarrow \sigma_{x}$. 
Here, we do not assume any condition for the set $\mathcal{X}$, except that it is 
a measurable space and that the channels are measurable maps (with the usual 
Borel sets on the state space $\mathcal{S}^B$). In particular, it might be an 
uncountably infinite set.

The task is to discriminate two hypotheses, the null 
hypothesis $H_{0}:\mathcal{N}$ versus the alternative hypothesis 
$H_{1}:\overline{\mathcal{N}}$ where $n\rightarrow\infty$ (independent) uses of the unknown channel 
are provided. Then, the challenge we face is to make a decision in favor of the true 
channel based on $n$ inputs $\vec{x}_{n}=(x_{1},\ldots,x_{n})$ and corresponding output
states on $B^n = B_1\cdots B_n$; 
note that the input $\vec{x}_{n}=(x_{1},\ldots,x_{n})$ is generated by
a very complicated joint distribution of $n$ random variables,
which -- except for $x_1$ -- depend on the actual channel.
Hence, they are written with the capitals as $X^n=X_1,\ldots, X_n$
when they are treated as random variables.
\if0
\fs{The rest of this Section is organized as follows: 
We first introduce quantum instruments and provide useful lemmas needed for the rest of this Section.
Our adaptive method is proven in Subsection B, and subsequently Subsection C proves several auxiliary lemmas
leading to the main result of this Section, which is presented in Subsection D. 
}
\fi
%
%
%
%
%
\subsection{Quantum measurements}
\label{S4}
To formulate our general adaptive method for the discrimination of cq-channels,   
we prepare a general notation for quantum measurements with state changes.
A general quantum state evolution from $A$ to $B$
is written as a cptp map $\mathcal{M}$ from the space 
$\mathcal{T}^A$ to the space $\mathcal{T}^B$ of trace class 
operators on $A$ and $B$, respectively.
When we make a measurement on the initial system $A$, 
we obtain the measurement outcome $K$ and the resultant state on the output system $B$.
To describe this situation, we use a set $\{\kappa_k\}_{k \in \mathcal{K}}$ 
of cp maps from the space $\mathcal{T}^A$ to the space $\mathcal{T}^B$
such that
$\sum_{k \in \mathcal{K}} \kappa_k$ is trace preserving.
In this paper, since the classical feed-forward  information is assumed to be a discrete variable,
$\mathcal{K}$ is a discrete (finite or countably infinite) set.
Since it is a decomposition of a cptp map, it is often called a \emph{cp-map valued measure}, 
and an \emph{instrument} if their sum is cptp \footnote{For simplicity, here and in the rest of the paper, we assume the set $\mathcal{K}$ to be discrete. In fact, if the Hilbert spaces $A$, $B$, etc, on which the cp maps act are finite dimensional, then every instrument is a convex combination, i.e. a probabilistic mixture, of instruments with only finitely many non-zero elements; this carries over to instruments defined on a general measurable space $\mathcal{K}$. Thus, in the finite-dimensional case the assumption of discrete $\mathcal{K}$ is not really a restriction.}.
In this case, when the initial state on $A$ is $\rho$ and 
the outcome $k$ is observed with probability 
$\Tr \kappa_k(\rho)$, where
the resultant state on $B$ is $\kappa_k(\rho)/\Tr \kappa_k(\rho)$.
A state on the composite system of the classical system $K$ and the quantum $B$ 
is written as
$\sum_{k \in {\cal K}}|k\rangle \langle k| \otimes \rho_{B|k}
$, which belongs to the vector space 
$\mathcal{T}^{KB}:=
\sum_{k \in {\cal K}}|k\rangle \langle k| \otimes \mathcal{T}^B$.
The above measurement process can be written as the following cptp 
$\mathcal{E}$ map from $\mathcal{T}^A$ to $\mathcal{T}^{KB}$.
\begin{align}
\mathcal{E}(\rho):= \sum_{k \in \mathcal{K}}|k\rangle \langle k| \otimes \kappa_{k}(\rho).\label{NAC}
\end{align}
In the following, 
both of the above cptp map $\mathcal{E}$ and 
a cp-map valued measure are called a quantum instrument.
%A general cp-map valued measure has the following form.

\subsection{Formulation of adaptive method}\label{S4-A-1}
%\subsubsection{General protocol for cq-channels}
To study the adaptive discrimination of cq-channels,   
the general strategy for discrimination of qq-channels in Sec. \ref{sec:intro}
should be tailored to the cq-channels. 
We argue that the most general strategy in Sec. \ref{sec:intro}
can w.l.o.g. be replaced by the kind of strategy with the instrument 
and only classical feed-forward when the hypotheses are a pair of cq-channels.
This in particular will turn out to be crucial since we consider general cq-channels with arbitrary (continuous) input 
alphabet.
%
%\mh{Section \ref{SV-B} will see how the most general strategy in Sec. \ref{sec:intro} with qq-channels 
%can be recovered from our cq-channel model.}

%Before this protocol, the receiver prepares his quantum memory $R$
The first input is chosen subject to the distribution $p_{X_1}(x_1)$. 
The receiver receives the output $\rho_{x_1}$ or $\sigma_{x_1}$ on $B_1$.
Dependent on the input $x_1$,
the receiver applies the first quantum instrument $\{\Gamma_{k_1|x_1}^{(1)}\}_{k_1 \in {\cal K}_1}:B_1 \to K_1 R_2$,
where $R_2$ is the quantum memory system and 
$K_1$ is the classical outcome.
The receiver sends the outcome $K_1$ to the sender.
Then, 
the sender choose the second input $x_2$ according to
the conditional distribution $p_{X_2|X_1,K_1}(x_2|x_1,k_1)$. 
The receiver receives the second output $\rho_{x_2}$ or $\sigma_{x_2}$ on $B_2$.
Dependent on the previous outcome $k_1$ and
the previous inputs $x_1, x_2$,
the receiver applies the second quantum instrument
$\{\Gamma_{k_2|x_1,x_2,k_1}^{(2)}\}_{k_2 \in {\cal K}_2}:B_2 R_2 \to K_2 R_3$, and
sends the outcome $K_2$ to the sender.
The third input is chosen as the distribution $
p_{X_3|X_1,X_2,K_1,K_2}(x_2|x_1,x_2,k_1,k_2)$. 

In the same way as the above, the $m$-th step is given as follows.
The sender chooses the $m$-th input $x_m$ according to
the conditional distribution $p_{X_m| \vec{X}_{m-1},\vec{K}_{m-1}}(x_{m}|\vec{x}_{m-1}, \vec{k}_{m-1})$. 
The receiver receives the second output $\rho_{x_m}$ or $\sigma_{x_m}$ on $B_m$.
The remaining processes need the following divided cases.
For $m<n$, 
dependent on the previous outcomes $\vec{k}_{m-1}:=(k_1, \ldots, k_{m-1})$ and
the previous inputs $\vec{x}_m:=(x_1, \ldots, x_{m})$,
the receiver 
applies the $m$-th quantum instrument
%$\mathcal{E}_m: 
$\{\Gamma_{k_m|\vec{x}_m,\vec{k}_{m-1}}^{(m)}\}_{k_m \in {\cal K}_m}:
R_m B_m \to K_m R_{m+1}$, and sends the outcome $k_m$ to the sender.
For $m=n$, 
dependent on the previous outcomes 
$\vec{K}_{n-1}$ and the previous inputs $\vec{X}_{n}$,
the receiver measures the final state on ${R_{n}B_n}$ with 
the binary POVM 
$(T_{n|\vec{k}_{n-1},\vec{x}_{n}},I-T_{n|\vec{k}_{n-1},\vec{x}_{n}})$, where
hypothesis $\mathcal{N}$ (resp. $\overline{\mathcal{N}}$) is accepted if and only if
the first (resp. second) outcome clicks.

In the following,
we denote the class of the above general strategies 
by $\underline{\mathbb{A}}^{c,0}$
because it can be considered that this strategy has no quantum memory in the input side and no quantum feedback. 
As a subclass, we focus on the class
when no feedback is allowed
and the input state deterministically is fixed to a single input $x$, 
which is denoted $\underline{\mathbb{P}}^{0}$.

\if0
\medskip
\begin{remark}[Relation to general setting with qq-channels]
Here, we discuss how to derive the above setting from 
the general setting presented in introduction for cq-channels.
In the case with cq-channel, the input needs to be
a classical element in the discrete set $\mathcal{X}$.
To decide the classical input, we need to apply measurement 
after the application of the $m$-th cptp map $\mathcal{F}_m$.
That is, we need to replace the $m$-th cptp map $\mathcal{F}_m$
by a quantum instrument
$\mathcal{E}_m: R_m B_m \to K_m R_{m+1}$, which feeds the outcome $K_m$ forward to 
the next channel use. 
Hence, the obtained procedure is equivalent to the procedure given above.

Note however that this way of thinking of a cq-channel as a special type of
qq-channel is restricted to discrete input alphabets; for general, in particular 
continuous input alphabet to the channels $\mathcal{N}$ and $\overline{\mathcal{N}}$, 
we directly use the description above.
\hfill $\square$
\end{remark}
\medskip
\fi

%\subsubsection{Two types of error probabilities}

\begin{figure}[ht]
\begin{center}
\label{BCP}
\includegraphics[width=0.5\textwidth]{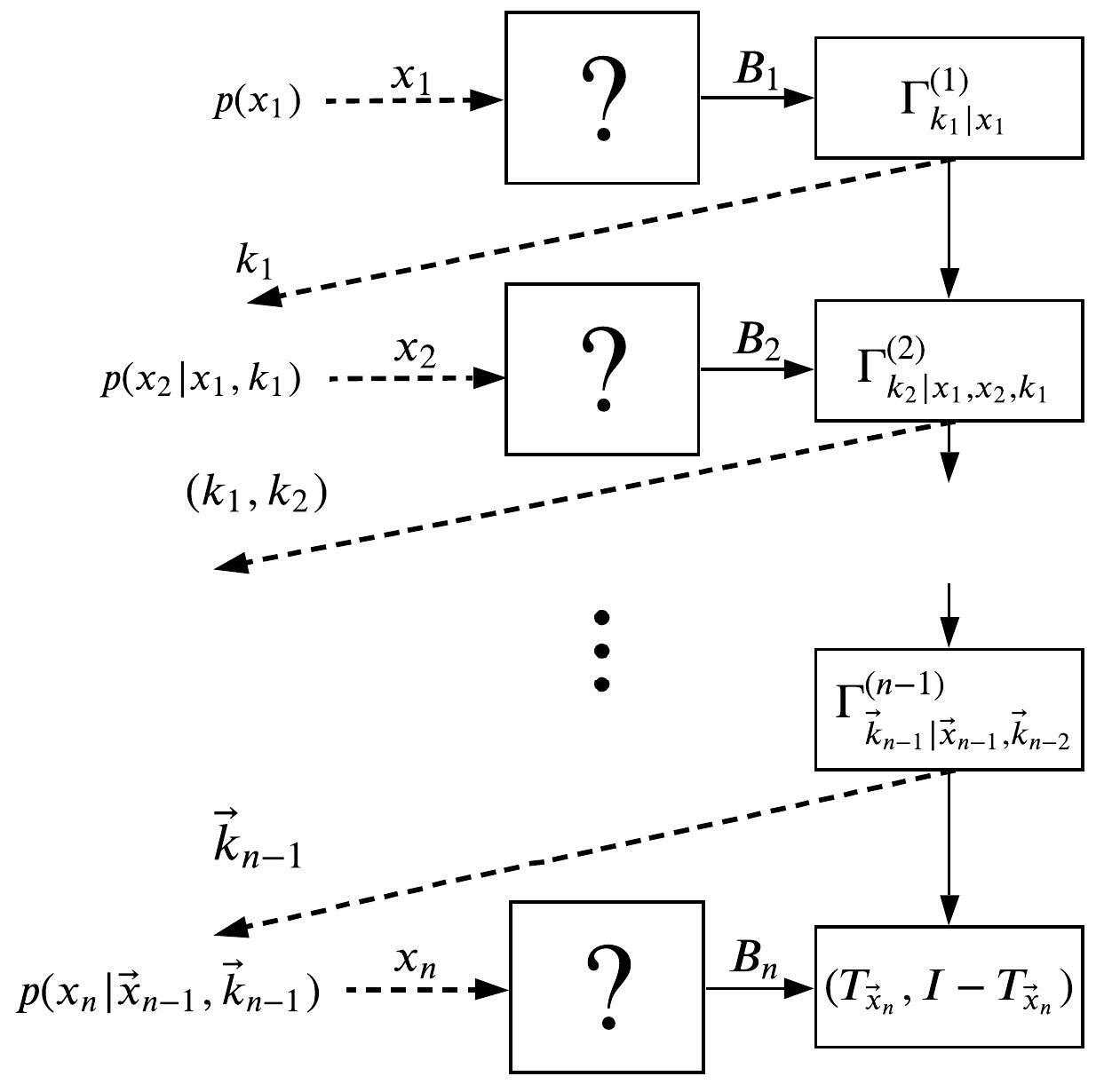}      
\caption{Adaptive strategy for cq-channel discrimination. Solid and dashed lines
         denote flow of classical and quantum information, respectively. 
         The classical outputs of instruments $\{\Gamma_{k_m|\vec{x}_m,\vec{k}_{m-1}}^{(m)}\}_{k_m \in {\cal K}_m}$
         are employed to decide the inputs adaptively, and leave a post-measurement state 
         that can be accessed together with the next channel output.}
\end{center}
\end{figure}

When the true channel is $\mathcal{N}:x\rightarrow \rho_{x}$, the state before the final measurement is 
\begin{align}
\rho^{(n)}\coloneqq 
\sum_{\vec{x}_n,\vec{k}_{n-1}}& p_{X_1}(x_1)\cdots 
p_{{X}_n|\vec{X}_{n-1},\vec{K}_{n-1}}({x}_n|\vec{x}_{n-1},\vec{k}_{n-1})
 \nonumber \\
& \left(
\Gamma_{k_{n-1}|\vec{x}_{n-1},k_{n-2}}
(
\cdots
\Gamma_{k_2|x_1,x_2,k_1}(\Gamma_{k_1|x_1}(\rho_{x_1})\otimes \rho_{x_2})
\otimes \cdots \otimes \rho_{x_{n-1}})\otimes \rho_{x_n}
\otimes \ketbra{\vec{x}_n,\vec{k}_{n-1}}\right),\label{BD1}
\end{align}
where here we need to store the information for inputs $\vec{x}_n$.
Similarly, when the true channel is $\overline{\mathcal{N}}:x\rightarrow \sigma_{x}$
\begin{align}
\sigma^{(n)}\coloneqq
\sum_{\vec{x}_n,\vec{k}_{n-1}} &
p_{X_1}(x_1)\cdots 
p_{{X}_n|\vec{X}_{n-1},\vec{K}_{n-1}}({x}_n|\vec{x}_{n-1},\vec{k}_{n-1})\nonumber \\
& \left(
\Gamma_{k_{n-1}|\vec{x}_{n-1},k_{n-2}}
(
\cdots
\Gamma_{k_2|x_1,x_2,k_1}(\Gamma_{k_1|x_1}(\sigma_{x_1})\otimes \sigma_{x_2})
\otimes \cdots \otimes \sigma_{x_{n-1}})\otimes \sigma_{x_n}
\otimes \ketbra{\vec{x}_n,\vec{k}_{n-1}}\right),\label{BD2}
\end{align}

A test of the hypotheses $\{\mathcal{N},\overline{\mathcal{N}}\}$ 
on the true channel is a two-valued POVM $\{T_{n},I-T_{n}\}$, where $T_{n}$ 
is given as
a Hermitian operator 
$\sum_{\vec{x}_n}T_{\vec{x}_n}\otimes |\vec{x}_n\rangle \langle \vec{x}_n|$
on $B^{\otimes n}\otimes X^{\otimes n}$ satisfying 
$0\leq T_{n}\leq I$.
Overall, our strategy to distinguish the channels 
$\{\mathcal{N},\overline{\mathcal{N}}\}$ when $n$ independent uses of each are available,
is given by the triple
$\mathcal{T}_n:=
(\{\Gamma^{(m)}_{{k}_{m}|\vec{x}_{m},\vec{k}_{m-1}}\}_{m=1}^{n-1},
\{p_{{X}_m|\vec{X}_{m-1},\vec{K}_{m-1}}\}_{m=1}^n,
T_{n})$.
The $n$-copy error probabilities of type I 
and type II are respectively as follows
\begin{align*}
\alpha_{n}({\mathcal{N}}\|\overline{\mathcal{N}}|\mathcal{T}_{n})
   &\coloneqq \Tr \rho^{(n)}(I-T_{n}),\\
\beta_{n}({\mathcal{N}}\|\overline{\mathcal{N}}|\mathcal{T}_{n})
   &\coloneqq \Tr \sigma^{(n)}T_{n}.
\end{align*}
The generalized Chernoff and Hoeffding quantities introduced in the 
introduction read as follows in the present cq-channel case for a given 
class $\underline{\mathbb{S}} = \underline{\mathbb{P}}^{0},\underline{\mathbb{A}}^{c,0}$ :
\begin{align}
\label{geneCher}
C^{\underline{\mathbb{S}}}(a,b|\mathcal{N}\|\overline{\mathcal{N}})&\coloneqq 
\sup_{\{\mathcal{T}_{n}\}}
\left\{\liminf_{n\rightarrow\infty}-\frac{1}{n}\log\left(
2^{an}
\alpha_{n}({\mathcal{N}}\|\overline{\mathcal{N}}|\mathcal{T}_{n})
+2^{bn}
\beta_{n}({\mathcal{N}}\|\overline{\mathcal{N}}|
\mathcal{T}_{n})\right)\right\},\\
\label{hoeffdingD}
B_{e}^{\underline{\mathbb{S}}}(r|\mathcal{N}\|\overline{\mathcal{N}})&\coloneqq
\sup_{\{\mathcal{T}_{n}\}}
\left\{\liminf_{n\rightarrow\infty} -\frac{1}{n}
\log\left(\alpha_{n}({\mathcal{N}}\|\overline{\mathcal{N}}|
\mathcal{T}_{n})\right)\bigg|
\liminf_{n\rightarrow\infty} -\frac{1}{n}\log\left(\beta_{n}({\mathcal{N}}\|\overline{\mathcal{N}}|\mathcal{T}_{n}
)\right)\geq r\right\},
\end{align}
where $a$, $b$, are arbitrary real numbers and $r$ is an arbitrary non-negative number.

\subsection{Main results}
We set $\rho_x:=\mathcal{N}(x)$ and $\sigma_x:=\overline{\mathcal{N}}(x)$,
and define
\begin{align}
C(a,b|\mathcal{N}\|\overline{\mathcal{N}}):
=&\sup_{x}\sup_{0\leq\alpha\leq 1} (1-\alpha)D_{\alpha}(\rho_{x}\|\sigma_{x})
-\alpha a -(1-\alpha)b \nonumber \\
=&\sup_{0\leq\alpha\leq 1} (1-\alpha)D_{\alpha}(\mathcal{N}\|\overline{\mathcal{N}})
-\alpha a -(1-\alpha)b, \label{deff-C}\\
\label{deff}
B(r|\mathcal{N}\|\overline{\mathcal{N}}):
=&\sup_x \sup_{0\le \alpha \le 1}
\frac{\alpha-1}{\alpha}\big(r-D_{\alpha } (\rho_x\|\sigma_x)\big)
= \sup_{0\le \alpha \le 1}
\frac{\alpha-1}{\alpha}\big(r- D_{\alpha } (\mathcal{N}\|\overline{\mathcal{N}})\big),
\end{align}
where $D_{\alpha } (\mathcal{N}\|\overline{\mathcal{N}}) :=
\sup_x D_{\alpha } (\rho_x\|\sigma_x)$
and
$D_{\alpha}(\rho_{x}\|\sigma_{x})\coloneqq\frac{1}{\alpha-1}\log\Tr\rho_{x}^{\alpha}\sigma_{x}^{1-\alpha}$
is a quantum extension of the R\'{e}nyi relative entropy.
In this section, we abbreviate 
$C(a,b|\mathcal{N}\|\overline{\mathcal{N}})$
and $B(r|\mathcal{N}\|\overline{\mathcal{N}})$ to 
$C(a,b)$ and $B(r)$, respectively.

Since $D_{\alpha } (\rho_x\|\sigma_x)$ is monotonically increasing for $\alpha$, 
$D_{\alpha } (\mathcal{N}\|\overline{\mathcal{N}})$
is monotonically increasing for $\alpha$.
Thus,
\begin{align*}
\lim_{\alpha\to 1}D_{\alpha } (\mathcal{N}\|\overline{\mathcal{N}})
&=
\sup_{0\le \alpha \le 1}D_{\alpha } (\mathcal{N}\|\overline{\mathcal{N}})
=\sup_{0\le \alpha \le 1}
\sup_x D_{\alpha } (\rho_x\|\sigma_x)\\
&=\sup_x \sup_{0\le \alpha \le 1}D_{\alpha } (\rho_x\|\sigma_x)
=\sup_x D (\rho_x\|\sigma_x)
=D(\mathcal{N}\|\overline{\mathcal{N}}).
\end{align*}
Before stating the main results of this section we shall study the $B(r)$ function further.
Since the $B(r)$ function is monotonically decreasing in $r$,
 $B(D(\mathcal{N}\|\overline{\mathcal{N}}))=0 $.
 To find $B(0)$, since 
$\frac{1-\alpha}{\alpha} D_{\alpha } (\mathcal{N}\|\overline{\mathcal{N}})
=D_{1-\alpha } (\overline{\mathcal{N}}\|\mathcal{N})$, we infer that
$\frac{1-\alpha}{\alpha} D_{\alpha } (\mathcal{N}\|\overline{\mathcal{N}})$
is monotonically decreasing for $\alpha$, and 
$D(\overline{\mathcal{N}}\|\mathcal{N})=\lim_{\alpha\to 0}
\frac{1-\alpha}{\alpha} D_{\alpha } (\mathcal{N}\|\overline{\mathcal{N}})$.
Hence, 
$B(0)=D(\overline{\mathcal{N}}\|\mathcal{N})$, and 
$B(r)<D(\overline{\mathcal{N}}\|\mathcal{N})$ for $r>0$.

As shown in Appendix \ref{AC}, we have the following lemma.
\begin{lemma}\label{LPX}
When real numbers $a,b$ satisfy
$-D(\mathcal{N}\|\overline{\mathcal{N}})\le a-b \le D(\overline{\mathcal{N}}\|\mathcal{N})$,
there exists $r_{a,b} \in [0, D(\mathcal{N}\|\overline{\mathcal{N}})]$
such that $B(r_{a,b})-r_{a,b}= a-b$.
\end{lemma}
We are now in a position to present and prove our main result, the generalized Chernoff bound and 
Hoeffding bound
as follows:
\begin{theorem}[Generalized Chernoff bound \& Hoeffding bound]
\label{chernoff}
For two cq-channels $\mathcal{N}$ and $\OL{\mathcal{N}}$, 
we have
\begin{align*}
  C^{\underline{\mathbb{A}}^{c,0}}(a,b|\mathcal{N}\|\overline{\mathcal{N}})
    = C^{\underline{\mathbb{P}}^{0}}(a,b|\mathcal{N}\|\overline{\mathcal{N}})
    = C(a,b) = r_{a,b}-b = B(r_{a,b})-a.
\end{align*}
for real numbers $a,b$ satisfying 
$-D(\mathcal{N}\|\overline{\mathcal{N}})\le a-b \le D(\overline{\mathcal{N}}\|\mathcal{N})$,
and
\begin{align*}
  B_{e}^{\underline{\mathbb{A}}^{c,0}}(r|\mathcal{N}\|\overline{\mathcal{N}}) = 
  B_{e}^{\underline{\mathbb{P}}^{0}}(r|\mathcal{N}\|\overline{\mathcal{N}}) = 
  B(r)
\end{align*}
for any $0\leq r\leq D(\mathcal{N}\|\overline{\mathcal{N}})$.
\hfill $\square$ % Followed by proof!
\end{theorem}
This theorem is shown in Appendix \ref{AD} after various preparations.
The key point of the proof of Theorem \ref{chernoff}
is the reduction of our general strategy to the special strategy that restrict general instruments
$\{\Gamma_{k_m|\vec{x}_m,\vec{k}_{m-1}}^{(m)}\}_{k_m \in {\cal K}_m} $ to the application of projective measurement with projection postulate. 
This reduction is stated as Proposition \ref{PNL} in Appendix \ref{A2}. 
In this reduction, as stated in Lemma \ref{NS1},
we convert the cq-channels into classical channels 
be means of the eigenvalue decomposition of the output states, 
using the two distributions introduced by \cite{Hayashi_2002,nussbaum2009}.

\section{Formulation of general adaptive method for qq-channel discrimination}\label{new3}
\if0
This description of the problem presupposes that the two hypotheses 
correspond to objects in a probabilistic framework, in which also the possible 
tests (decision rules) are phrased, so as to give unambiguous meaning to the 
type-I and type-II error probabilities. The traditionally studied framework 
is that each hypothesis represents a probability distribution on a given
set, and more generally a state on a given quantum system. 
\fi
Hereafter, the hypotheses are described by two 
qq-channels, i.e. completely positive and trace preserving (cptp) maps, 
acting on a given quantum system, and more precisely $n \gg 1$ independent 
realizations of the unknown channel. It is not hard to see that both 
the type-I and type-II error probabilities can be made to go to $0$ 
exponentially fast, just as in the case of hypotheses described by quantum 
states, and hence the fundamental question is the characterization of the 
possible pairs of error exponents. 

\medskip
%\textbf{Notation.} 
%
We identify states $\rho$ with their density operators and use superscripts
to denote the systems on which the mathematical objects are defined.
The set of density matrices (positive semidefinite matrices with unit trace) 
on $A$ is written as $\mathcal{S}^{A}$, 
a subset of the trace class operators, denoted $\mathcal{T}^{A}$. \fs{An operator is called projection operator 
if applying it twice has the same effect as applying it once, i.e. $\rho^2=\rho$. The subspace
that the projection operator $\rho$ projects onto is called its image and is denoted by $\im \rho$.}
When talking about tensor products of spaces, we may habitually omit the tensor 
sign, so $A\otimes B = AB$, etc. 
The capital letters $X$, $Y$, etc. denote random variables whose realizations 
and the alphabets will be shown by the corresponding small and calligraphic letters, 
respectively: $X = x \in \mathcal{X}$. 
All Hilbert spaces and ranges of variables may be infinite; the dimension 
of a Hilbert space $A$ is denoted $|A|$, as is the cardinality $|\mathcal{X}|$
of a set $\mathcal{X}$.
For any positive integer $m$, we define $\vec{x}_{m} \coloneqq (x_{1},\cdots,x_{m})$. 
For the state $\rho \in \mathcal{S}^{AB}$ in the composite system $AB$,
the partial trace over system $A$ (resp. $B$) is denoted by $\Tr_{A}$ (resp. $\Tr_{B}$). 
We denote the identity operator by $I$. We use $\log$ and $\ln$
to denote 
base $2$ and 
natural logarithms, respectively. 
Moving on to quantum channels, these are linear, completely positive and 
trace preserving maps $\mathcal{M} : \mathcal{S}^A \rightarrow \mathcal{S}^B$
for two quantum systems $A$ and $B$; $\mathcal{M}$ extends uniquely to a 
linear map from trace class operators on $A$ to those on $B$. We often denote
quantum channels, by slight abuse of notation, as $\mathcal{M} : A \rightarrow B$.
\fs{The input and output systems of quantum channels can include quantum and classical information; 
if both input and output systems are quantum, the channel is referred to as quantum-quantum channel (qq-channel).
Similarly, one can identify classical-quantum channels (cq-channels) and quantum-classical channels (qc-channels).}
The ideal, or identity, channel on $A$ is denoted $\id_A$. Note furthermore 
that a state $\rho^A$ on a system $A$ can be viewed as a quantum channel 
$\rho: 1 \rightarrow A$, where $1$ denotes the canonical one-dimensional 
Hilbert space, isomorphic to the complex numbers $\mathbb{C}$, which 
interprets a state operationally consistently as a state preparation procedure. 

\medskip
The most general operationally justified strategy to distinguish two 
channels $\mathcal{M},\overline{\mathcal{M}}:A\rightarrow B$ 
is to prepare a 
state $\rho^{RA}$, apply the unknown channel to $A$ (and the identity 
channel $\id_R$ to $R$), and then apply a binary measurement POVM
$(T,I-T)$ on $BR$, so that 
\begin{align*}
  %\alpha(\mathcal{M}\|\overline{\mathcal{M}}) &= 
  \alpha = \Tr \bigl( (\id_R\otimes\mathcal{M})\rho \bigr)(I-T)
  \quad\text{and}\quad
  %\beta(\mathcal{M}\|\overline{\mathcal{M}})  &=
  \beta = \Tr \bigl( (\id_R\otimes\OL{\mathcal{M}})\rho \bigr)T,
\end{align*}
are the error probabilities of type I and type II, respectively. 
When we choose the system $R$ to be  sufficiently large, $\rho^{RA}$ is a pure state. 
Since the rank of $\rho^R$ is the same as the rank of $\rho^A$ in this case, 
we can restrict the dimension of $R$ to be $|A|$ without loss of generality.
\if0
It is easy to
see that whatever state $\rho^{AR}$ is considered as
input, it can be purified to $\psi^{ARR'}$, with a suitable Hilbert 
space, and the latter state can be used to get the same error 
probabilities. Then, once there is a pure state, one only needs a 
subspace of $R\otimes R'$ of dimension $|A|$, namely the support of $\psi^{RR'}$, 
which by the Schmidt decomposition is at most $|A|$-dimensional.
Therefore, the state $\rho$ is without loss of generality pure and that hence 
$R$ has dimension at most that of $A$. 
\fi
\fs{For more on the dimension of reference system 
we refer to \cite{cite-key}}.
The strategy is entirely described by the pair $\bigl(\rho,(T,I-T)\bigr)$ consisting 
of the initial state and the final measurement, and we denote it $\mathcal{T}$. 
Consequently, the above error probabilities are more precisely denoted 
$\alpha(\mathcal{M}\|\overline{\mathcal{M}}|\mathcal{T})$ and 
$\beta(\mathcal{M}\|\overline{\mathcal{M}}|\mathcal{T})$,
respectively.

%\medskip
These strategies use the unknown channel exactly once; to use it $n>1$ times, 
one could simply consider that $\mathcal{M}^{\otimes n}$ and $\OL{\mathcal{M}}^{\otimes n}$
are quantum channels themselves and apply the above recipe. While for states 
this indeed leads to the most general possible discrimination strategy, 
for general channels other, more elaborate procedures are possible. 
The most general strategy we shall consider in this paper is the \emph{adaptive} 
strategy, applying the $n$ channel instances sequentially, using 
quantum memory and quantum feed-forward, and a measurement at the end.
This is called, variously, an adaptive strategy, a memory channel or a comb 
in the literature. 
It is defined as 
follows \cite{KretschmannWerner,watrous:comb,chiribella:memory,chiribella:superchannel,chiribella:comb,5165184}.

\medskip
\begin{definition}
\label{Def1}
A general adaptive strategy $\mathcal{T}_n$ is given by an $(n+1)$-tuple 
$\bigl(\rho_1^{R_1A_{1}},\mathcal{F}_{1},\ldots,\mathcal{F}_{n-1},(T,I-T)\bigr)$, 
consisting of an auxiliary system $R_1$ and a state $\rho_1$ on $R_1A_1$,
quantum channels $\mathcal{F}_m : R_{m}B_m \rightarrow R_{m+1}A_{m+1}$ and a 
binary POVM $(T,I-T)$ on $R_nB_n$. It encodes the following procedure
(see Fig.~\ref{general}): 
in the $m$-th round ($1\leq m\leq n$), apply the unknown channel 
$\Xi \in \{\mathcal{M},\overline{\mathcal{M}}\}$ to $\rho_m = \rho_m^{R_m A_m}$, obtaining 
\[
  \omega_m^{R_m B_m} = \omega_m^{R_m B_m}(\Xi) = (\id_{R_t}\otimes\Xi)\rho_m^{R_m A_m}. 
\]
Then, as long as $m<n$, use $\mathcal{F}_{m}$ to prepare the state 
for the next channel use: 
\[
  \rho_{m+1}^{R_{m+1}A_{m+1}} = \mathcal{F}_{m}(\omega_m^{R_m B_m}). 
\]
When $m=n$, measure the state $\omega_n^{R_n B_n}$ with $(T,I-T)$, where the 
first outcome corresponds to declaring the unknown channel to be $\mathcal{M}$,
the second $\overline{\mathcal{M}}$. Thus, the $n$-copy error probabilities of type I and 
type II are given by
\begin{align*}
  \alpha_n(\mathcal{M}\|\OL{\mathcal{M}}|\mathcal{T}_n) &:= \Tr \bigl(\omega_n^{R_nB_n}(\mathcal{M})\bigr)(I-T), \\
  \beta_n(\mathcal{M}\|\OL{\mathcal{M}}|\mathcal{T}_n)  &:= \Tr \bigl(\omega_n^{R_nB_n}(\OL{\mathcal{M}})\bigr)T,
\end{align*}
respectively. 
\hfill$\square$
\end{definition}

\begin{figure}[ht]
\begin{center}
\includegraphics[width=0.9\textwidth]{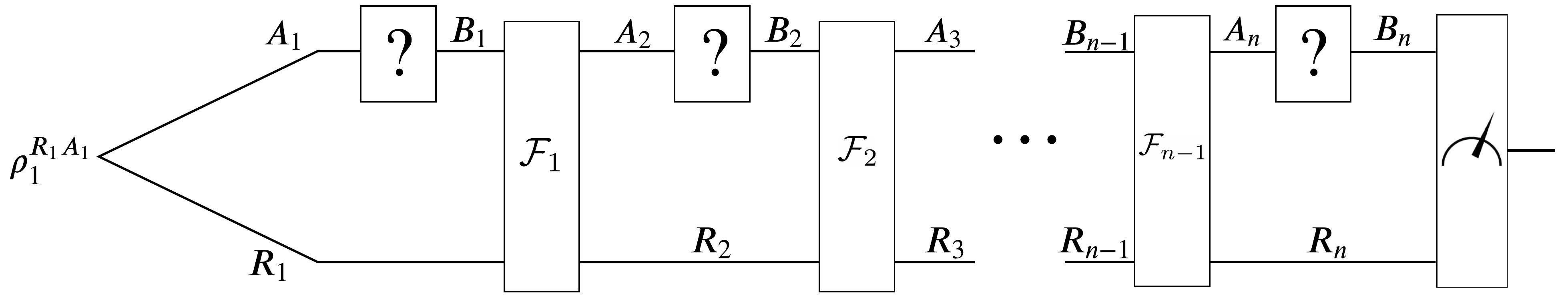}      
\caption{The most general adaptive strategy for discrimination of qq-channels, 
         from the class $\mathbb{A}_{n}$. 
         After the $m$-th use of the unknown channel (denoted `$?$'), the output system 
         $B_{m}$ as well as the state on the memory, i.e. the reference system $R_{m}$, 
         is processed by the cptp map $\mathcal{F}_{m}$, resulting in $\rho_{m+1}^{R_{m+1}A_{m+1}}$; 
         this continues as long as $m<n$. 
         After the $n$-th use of the channel, the state $\omega^{R_{n}B_{n}}_{n}$
         is measured by a two-outcome POVM.
         Two variants of this strategy include restricting feed-forward information 
         to be only classical, and additionally only allowing products state inputs;
         these variants are denoted by $\mathbb{A}_{n}^{c}$ and $\mathbb{A}_{n}^{c,0}$, respectively. }
\label{general}
\end{center}
\end{figure}

\medskip
As in the case of a single use of the channel, one can 
without loss of generality (w.l.o.g.) simplify the strategy, by purifying 
the initial state $\rho_1$, hence $|R_1| \leq |A|$, and for each $m>1$ going
to the Stinespring isometric extension of the cptp map 
$\Tr_{R_{m+1}}\circ\mathcal{F}_{m} : R_{m}B_m \rightarrow A_{m+1}$ that prepares 
the next channel input (and which by the uniqueness of the Stinespring extension is
an extension of the given map $\mathcal{F}_{m}$). This requires a 
system $R_{m+1}$ with dimension no more than $|R_{m+1}| \leq |R_m| |A| |B|$,
cf. \cite{KretschmannWerner}. 
This allows to efficiently parametrize all strategies in the case that 
$A$ and $B$ are finite dimensional. An equivalent description
is in terms of so-called causal channels \cite{KretschmannWerner}, 
which are ruled by a generalization of the Choi isomorphism. This turns 
many optimizations over adaptive strategies into semidefinite programs (SDP)
\cite{KretschmannWerner,chiribella:comb,Watrous:SDP,Gutoski:diamond}, 
which is relevant for practical calculations. See \cite{Pira:survey,KW20} for recent 
comprehensive surveys of the concept of strategy and its history.

\begin{figure}[ht]
\begin{center}
\includegraphics[width=0.45\textwidth]{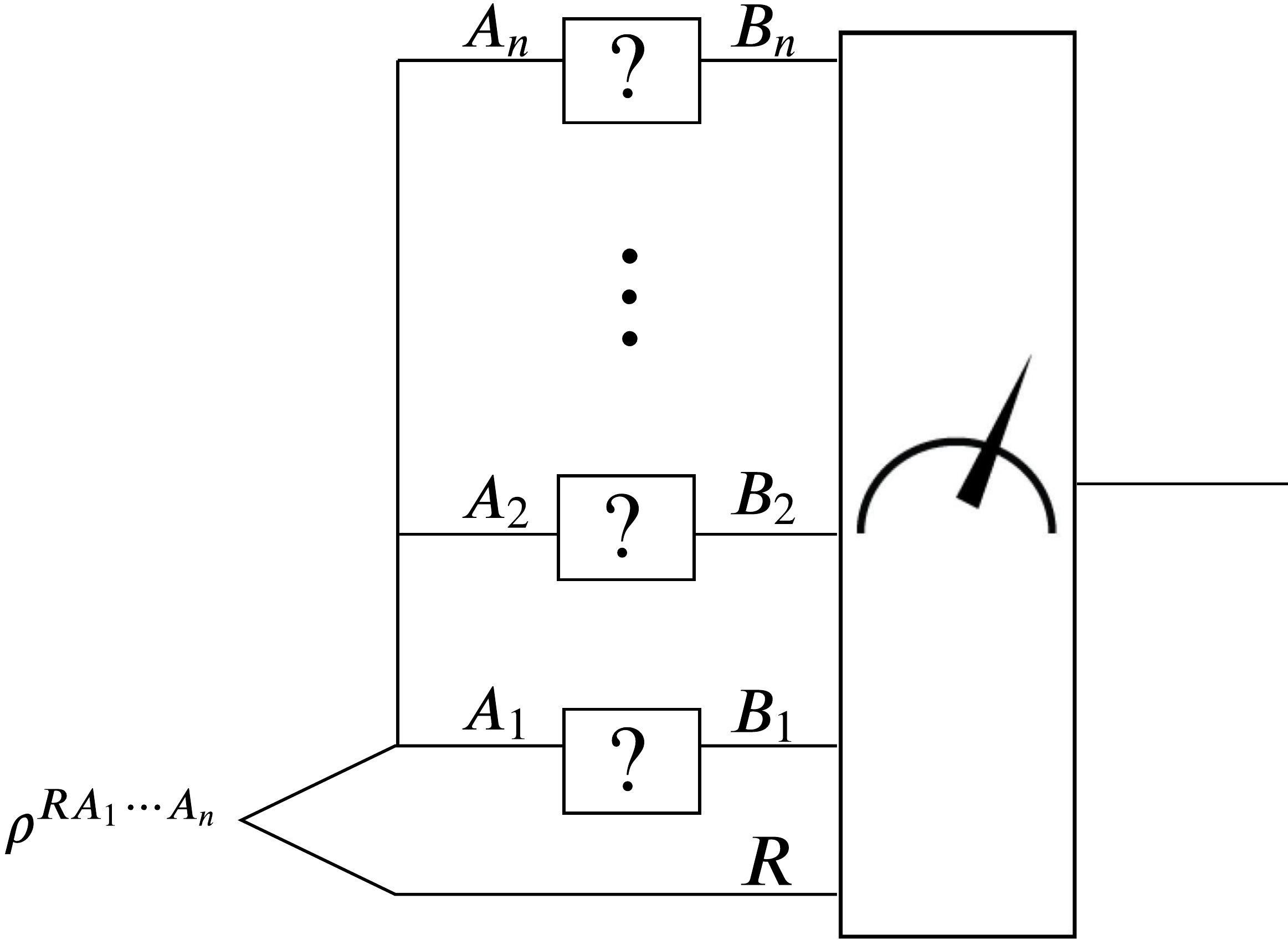}      
\caption{The most general parallel strategy for discrimination of qq-channels, from
         the class $\mathbb{P}_{n}$. 
         An $(n+1)$-partite state $\rho$ on $RA_{1}\ldots A_{n}$ is prepared and each system
         $A_i$ is fed into a separate channel input; the final measurement is performed with
         a two-outcome POVM on $RB_{1}\ldots B_{n}$. If we do not allow input states to be 
         entangled among different $A$-systems or with the reference system $R$, the strategy 
         falls into the class $\mathbb{P}_{n}^{0}$.}
\label{parallel}
\end{center}
\end{figure}

The set of all adaptive strategies of $n$ sequential channel uses is denoted 
$\mathbb{A}_n$. It quite evidently includes the $n$ parallel uses 
described at the beginning, when a single-use strategy 
is applied to the channel \fs{$(?)^{\otimes n}$, i.e. $n$-fold tensor-product of the unknown channel}; 
the set of these non-adaptive or parallel strategies 
is denoted $\mathbb{P}_n$. Among those again, we can distinguish the subclass 
of parallel strategies without quantum memory, meaning that $R=1$ is trivial 
and that the input state $\rho^{A^n}$ at the input system $A^n = A_1\ldots A_n$ 
is a product state, $\rho^{A^n} = \rho_1^{A_1}\otimes\cdots\otimes\rho_n^{A_n}$; 
this set is denoted $\mathbb{P}^{0}_n$. 
Other restricted sets of strategies we are considering in the present 
paper are that of adaptive strategies with classical feed-forward, 
denoted $\mathbb{A}^{c}_n$, and with classical feed-forward and no quantum 
memory at the input, denoted $\mathbb{A}^{c,0}_n$, as well as no quantum memory 
at the input but quantum feed-forward, denoted $\mathbb{A}^{0}_n$.
They are defined formally in Section \ref{S-5}. 

The various classes considered obey the following inclusions 
that are evident from the definitions. Note that all of them are strict:
\begin{alignat}{4}
  \mathbb{A}_n\quad   &\supset &\quad\mathbb{A}_n^{c\phantom{,0}}\quad     &\!\supset &\quad\mathbb{P}_n \nonumber\\
  \label{eq:inclusions}
      \cup\quad\       &\phantom{\supset} &   \cup\qquad            &\!\phantom{\supset} &   \cup\     \\
  \mathbb{A}_n^0\quad &\supset &\quad\mathbb{A}_n^{c,0}\quad &\supset &\!\quad\mathbb{P}_n^0 \nonumber
\end{alignat}

In Table \ref{table-of-classes} we show a summary of the different 
classes and their notation, and where they are discussed.

%\newpage

\begin{table}[ht]
\renewcommand{\arraystretch}{1.5}
\begin{center}
    \begin{tabular}{ | l || p{1.2cm} | p{4.5cm} | p{5.5cm} | p{1.4cm} |}
    \hline
    Name           & Defined            & Mathematical elements & Description & Discussed \\ \hline\hline 
    $\mathbb{A}_n$ & \parbox[t]{1.2cm}%
                    {Def.~\ref{Def1},\\ 
                    Fig.~\ref{general}}  & \parbox[t]{4.5cm}{state $\rho_1^{R_1A_{1}}$, channels\\
                                           $\mathcal{F}_{m}:R_{m}B_m \rightarrow R_{m+1}A_{m+1}$,\\
                                           POVM $(T,I-T)$}\vspace{1mm} & Most general adaptive strategy 
                                                                         of $n$ channel uses
                                                                              & \parbox[t]{1.15cm}{Sec.~\ref{lowerbound}, 
                                                                                \ref{subsec:cq-as-qq}, 
                                                                                \ref{subsec:power:A-n}} \\ \hline
    $\mathbb{P}_n$ & \parbox[t]{1.2cm}%
                    {Def.~\ref{Def1},\\ 
                    Fig.~\ref{parallel}} & \parbox[t]{4.5cm}{state $\rho^{RA^n}$, \\
                                           POVM $(T,I-T)$}    & Most general non-adaptive (parallel)
                                                                strategy allowing quantum memory at 
                                                                the input 
                                                                             &  Sec.~\ref{lowerbound} \\ \hline
    $\mathbb{P}_n^0$ & \parbox[t]{1.2cm}%
                    {Def.~\ref{Def1},\\
                     Fig.~\ref{parallel}} & \parbox[t]{4.5cm}{state $\rho_1^{A_1}\otimes\cdots\otimes\rho_n^{A_n}$, \\
                                            POVM $(T,I-T)$}   & Non-adaptive (parallel) strategy 
                                                                without quantum memory at the input
                                                                              & Sec.~\ref{subsec:cq-as-qq},
                                                                                \ref{subsec:A-c-0}, 
                                                                                %\ref{MO1},
                                                                                \ref{subsec:power:A-n} \\ \hline
    $\mathbb{A}_n^c$ & Def.~\ref{defi:A-n-c} 
                                        & \parbox[t]{4.5cm}{state $\rho_1^{R_1A_{1}}$, instruments \\
                                          $\{\mathcal{F}_{\vec{k}_m} \!:\! B_mC_m \!\to\! C_{m+1}\}_{\vec{k}_m}$,\\ 
                                          cptp $\mathcal{P}_{\vec{k}_m}\!:\!R_m \!\to\! R_{m+1}A_{m+1}$,\\
                                          POVM $(T,I-T)$}\vspace{1mm}    & Adaptive strategy of $n$ channel 
                                                                           uses with classical feed-forward,
                                                                           but otherwise arbitrary quantum memory
                                                                           at the input
                                                                              & Sec.~\ref{subsec:weaker-restriction}
                                                                                \\ \hline
    $\mathbb{A}_n^{c,0}$ & \parbox[t]{1.2cm}%
                         {Def.~\ref{defi:A-n-c},\\ 
                          Fig.~\ref{FA}} & \parbox[t]{4.5cm}{states $\rho_{x_m}^{A_m}$, instruments \\
                                           $\{\mathcal{F}_{\vec{k}_m} \!:\! B_mC_m \!\to\! C_{m+1}\}_{\vec{k}_m}$,\\
                                           $q(x_m|\vec{x}_{m-1},\vec{k}_{m-1})$, \\
                                           POVM $(T,I-T)$}\vspace{1mm}   & Adaptive strategy of $n$ channel 
                                                                           uses with classical feed-forward,
                                                                           and without quantum memory
                                                                           at the input 
                                                                              & Sec.~%\ref{subsec:cq-as-qq},
                                                                                \ref{subsec:A-c-0},
                                                                                \ref{subsec:weaker-restriction},
                                                                                %\ref{MO1}, 
                                                                                %\ref{subsec:power:A-n} 
                                                                                \\ \hline
    $\mathbb{A}_n^0$ & \parbox[t]{1.2cm}%
                     {Def.~\ref{Def1},\\
                      Rem.~\ref{rem:A-n-0}} & \parbox[t]{4.5cm}{state $\rho_1^{R_1A_{1}}$, channels \\
                                          $\mathcal{F}_{m}:B_mC_m \rightarrow A_{m+1}C_{m+1}$, \\
                                          POVM $(T,I-T)$}    & Adaptive strategy of $n$ channel uses 
                                                               without quantum memory at the input,
                                                               but otherwise arbitrary quantum feed-forward 
                                                                              & Sec.~\ref{subsec:weaker-restriction} 
                                                                                \\ \hline
    \end{tabular}
\end{center}
\caption{The different classes of adaptive strategies 
         considered in the present paper, how they are denoted, 
         where they are defined and which mathematical elements 
         have to specified to identify a strategy from each class. 
         In the last column we point to the sections of the paper 
         containing results on the respective classes.}
\label{table-of-classes}
\end{table}

For a given class $\mathbb{S}_n \subset \mathbb{A}_n$ of adaptive strategies 
for any number $n$ of channel uses, the fundamental problem is now to characterize the 
possible pairs of error exponents for two channels $\mathcal{M}$ and $\overline{\mathcal{M}}$:
\begin{align}
\nonumber
  &\mathfrak{E}(\mathcal{M}\|\overline{\mathcal{M}}|\mathbb{S}) \\
   & \coloneqq \left\{ (r,s) : \exists\mathcal{T}_n\in\mathbb{S}_n \ 
                0 \leq r \leq \liminf_{n\rightarrow\infty} 
                          -\frac1n \log\beta_n(\mathcal{M}\|\overline{\mathcal{M}}|\mathcal{T}_n),       \,                          
                              0 \leq s \leq \liminf_{n\rightarrow\infty} 
                          -\frac1n \log\alpha_n(\mathcal{M}\|\overline{\mathcal{M}}|\mathcal{T}_n)\right\}.
\end{align}
In particular, we are interested, for each \fs{$r\geq 0$, in the largest $s$} 
such that $(r,s) \in \mathfrak{E}(\mathcal{M}\|\overline{\mathcal{M}}|\mathbb{S})$.
To this end, we define the error rate tradeoff
\begin{align}
  B_e^{\mathbb{S}}(r|\mathcal{M}\|\overline{\mathcal{M}}) %\nonumber\\
    := \sup \left\{ s \left|%\text{ s.t. }
\begin{array}{l}    
    \exists \mathcal{T}_n\in\mathbb{S}_n \ 
                     r \leq \liminf_{n\rightarrow\infty} 
                          -\frac1n \log\beta_n(\mathcal{M}\|\overline{\mathcal{M}}|\mathcal{T}_n),\\
                            s \leq \liminf_{n\rightarrow\infty} 
                          -\frac1n \log\alpha_n(\mathcal{M}\|\overline{\mathcal{M}}|\mathcal{T}_n),
\end{array}
\right.\right\}
\end{align}
\fs{known as Hoeffding exponent}, as well as the closely related function
\begin{align}
  C^{\mathbb{S}}(a,b|\mathcal{M}\|\overline{\mathcal{M}})
    := \inf_{\mathcal{T}_n\in\mathbb{S}_n} 
       \liminf_{n\rightarrow\infty} 
           -\frac1n \log\left( 2^{na}\alpha_n(\mathcal{M}\|\overline{\mathcal{M}}|\mathcal{T}_n)
                               + 2^{nb}\beta_n(\mathcal{M}\|\overline{\mathcal{M}}|\mathcal{T}_n) \right).
\end{align}
Note that $\mathfrak{E}(\mathcal{M}\|\overline{\mathcal{M}}|\mathbb{S})$ is a closed set 
by definition, and for most `natural' restrictions $\mathbb{S}$, it is also convex. 
In the latter case, the graph of $B_e^{\mathbb{S}}(r|\mathcal{M}\|\overline{\mathcal{M}})$ 
traces the upper boundary of $\mathfrak{E}(\mathcal{M}\|\overline{\mathcal{M}}|\mathbb{S})$, 
and it can be reconstructed from $C^{\mathbb{S}}(a,b|\mathcal{M}\|\overline{\mathcal{M}})$ 
by a Legendre transform.
%\aw{[CITATIONS]}. 

Historically, two extreme regimes are of special interest: the maximally 
asymmetric error exponent, 
\[
  \max r \text{ s.t. } \exists s\ (r,s) \in \mathfrak{E}(\mathcal{M}\|\overline{\mathcal{M}}|\mathbb{S})
     = \max r \text{ s.t. } (r,0) \in \mathfrak{E}(\mathcal{M}\|\overline{\mathcal{M}}|\mathbb{S}),
\]
together with the opposite one of maximization of $s$, 
which are known as Stein's exponents, 
and the symmetric error exponent
\[\begin{split}
  C^{\mathbb{S}}(\mathcal{M},\overline{\mathcal{M}}) 
      &= \max r \text{ s.t. } (r,r) \in \mathfrak{E}(\mathcal{M}\|\overline{\mathcal{M}}|\mathbb{S}) \\
      &= C^{\mathbb{S}}(0,0|\mathcal{M}\|\overline{\mathcal{M}}),
\end{split}\]
which is generally known as Chernoff exponent or Chernoff bound. 

\medskip
In the present paper, we assume that all Hilbert spaces of interest
are separable, i.e. they are spanned by countable bases,
and we are primarily occupied with the performance of adaptive strategies. 
Naturally, the first question in this search
would be to investigate the existence of quantum channels for which some class 
$\mathbb{S}_n \subset \mathbb{A}_n$ outperforms the parallel strategy when $n\rightarrow\infty$;
in other words, if there exists a separation between adaptive and non-adaptive strategies. 
We study this question in general, and in particular when the channels are 
entanglement-breaking of the following form:
\begin{align}
\mathcal{M}(\xi) = \sum_x (\Tr E_x \xi) \rho_x, \quad
\overline{\mathcal{M}}(\xi) = \sum_x (\Tr E_x' \xi) \sigma_x, \label{PVM-ex}
\end{align}
where $\{E_x\}$ and $\{E_x'\}$ are PVMs and $\rho_x,\sigma_x$ are states on the output system.
We show that when these two PVMs are the same, $E_x=E_x'$, then
the largest class $\mathbb{A}_n$ cannot outperform the parallel strategy as $n\rightarrow\infty$.
When the two PVMs are different,
we find two examples such that the largest class $\mathbb{A}_n$ outperforms the parallel strategies
as $n\rightarrow\infty$.
For a general pair of qq-channels we focus on the class $\mathbb{A}^{c,0}_{n}$ 
of strategies without quantum memory at the sender's side and with adaptive strategies
that only allow for classical discrete feed-forward.
We show that the class $\mathbb{A}^{c,0}_n$ cannot outperform the parallel strategies 
when $n\rightarrow\infty$.
These findings are then applied to the discrimination power of a quantum channel, which 
quantifies how well two given states in $A^{\otimes n}$ can be discriminated after 
passing through a quantum channel, and whether adaptive strategies can be beneficial. 
To this end, we focus on a particular class of channels, namely 
cq-channels and investigate if the most general strategy 
offers any benefit over the most weak strategy $\mathbb{P}_{n}^{0}$.
This study takes an essential role in the above problems.

\section{The first asymptotic separation between adaptive and non-adaptive strategies}
\label{lowerbound}
\subsection{Useful proposition for asymptotic separation}
In this section we exhibit an asymptotic separation between the Chernoff error 
exponents of discriminating between two channels by adaptive versus
non-adaptive strategies. Concretely, we will show that two 
channels described in \cite{PhysRevA.81.032339}, and shown to be 
perfectly distinguishable by adaptive strategies of $n\geq 2$ copies, 
hence having infinite Chernoff exponent, nevertheless
have a finite error exponent under non-adaptive strategies. 

The separation is based on a general lower bound on non-adaptive strategies
for an arbitrary pair of channels. Consider two quantum channels, i.e. cptp maps, 
$\mathcal{M},\OL{\mathcal{M}}:A\rightarrow B$. 
To fix notation, we can write their Kraus decompositions as
\begin{align*}
  \mathcal{M}(\rho)      = \sum_i E_i \rho E_i^\dagger, 
  \quad
  \OL{\mathcal{M}}(\rho) = \sum_j F_j \rho F_j^\dagger.
\end{align*}
The most general strategy to distinguish them 
consists in the preparation of a, w.l.o.g. pure, state 
$\varphi$ on $A \otimes R$, where $R \simeq A$, send it through the unknown 
channel, and make a binary measurement $(T,I-T)$ on $B \otimes R$:
\begin{align*}
  p = \tr \left((\id_{R}\otimes\mathcal{M})\varphi\right)T, 
  \quad
  q = \tr \left((\id_{R}\otimes\OL{\mathcal{M}})\varphi\right)T, 
\end{align*}
and likewise $1-p$ and $1-q$ by replacing $T$ in the above formulas with $I-T$. 
Note that for uniform prior probabilities on the two hypotheses, 
the error probability in inferring the true channel from the measurement 
output is $\frac12(1-|p-q|)$.

The maximum of $|p-q|$ over state preparations and measurements gives rise to 
the (normalized) diamond norm distance of the channels \cite{Kitaev,AKN,Paulsen:book,Watrous:SDP}:
\[
  \max_{\varphi,T} |p-q| = \frac12 \| \mathcal{M}-\OL{\mathcal{M}} \|_\diamond,
\] 
which in turn quantifies the minimum discrimination error under the 
most general quantum strategy:
\[
  P_{e} = \frac12 \left( 1 - \frac12 \| \mathcal{M}-\OL{\mathcal{M}} \|_\diamond \right).
\]

We are interested in the asymptotics of this error probability when 
the discrimination strategy has access to $n \gg 1$ many instances of
the unknown channel in parallel, or in other words, in a non-adaptive
way. This means effectively that the two hypotheses are the simple 
channels $\mathcal{M}^{\otimes n}$ and $\OL{\mathcal{M}}^{\otimes n}$, so that 
the error probability is
\[
  P_{e,\mathbb{P}}^{(n)} = \frac12 \left( 1-\frac12\left\| \mathcal{M}^{\otimes n}
                                                           -\OL{\mathcal{M}}^{\otimes n} \right\|_\diamond \right).
\]
The (non-adaptive) Chernoff exponent is then given as
\[
  C^{\mathbb{P}}(\mathcal{M},\OL{\mathcal{M}}) 
         = \lim_{n\rightarrow\infty} -\frac{1}{n}\log P_{e,\mathbb{P}}^{(n)},
\]
the existence of the limit being guaranteed by general principles. Note that the limit
can be $+\infty$, which happens in all cases where there is an $n$ such that
$P_{e,\mathbb{P}}^{(n)} = 0$. It is currently unknown whether this is the only case; 
cf. the case of the more flexible adaptive strategies, for which there is a
simple criterion to determine whether there exists an $n$ such that the adaptive 
error probability $P_{e,\mathbb{A}}^{(n)} = 0$ \cite{PhysRevLett.103.210501}, and 
then evidently $C^{\mathbb{A}}(\mathcal{M},\OL{\mathcal{M}}) = +\infty$; 
conversely, we know that in all other cases, the adaptive Chernoff exponent 
is $C^{\mathbb{A}}(\mathcal{M},\OL{\mathcal{M}}) < +\infty$ \cite{2017arXiv170501642Y}. 
There exist also other lower bounds on the symmetric discrimination error by 
adaptive strategies, for instance \cite[Thm.~3]{PLLP:lowerbound} geared towards 
finite $n$.

Duan \emph{et al.} \cite{7541701} have attempted a characterization 
of the channel pairs such that there exists an $n$ with $P_{e,\mathbb{P}}^{(n)} = 0$,
and have given a simple sufficient condition for the contrary. Namely, 
the existing result \cite[Cor.~1]{7541701} states that if $\operatorname{span}\{ E_i^\dagger F_j \}$ 
contains a positive definite element, then for all $n$ we have $P_{e,\mathbb{P}}^{(n)} > 0$.
The following proposition, which makes the result of \cite{7541701} quantitative, 
is the main result of this section.

\medskip
\begin{proposition}
  \label{prop:bound}
When complex numbers $\gamma_{ij} \in \mathbb{C}$ satisfy the condition that
 $\sum_{ij} |\gamma_{ij}|^2 = 1$ 
  and $P := \sum_{ij} \gamma_{ij} E_i^\dagger F_j > 0$, i.e. $P$ is positive definite,
  then the inequality
  \[
    P_{e,\mathbb{P}}^{(n)} \geq \frac14 \lambda_{\min}(P)^{4n}
  \]
  holds for all $n$, where $\lambda_{\min}(A)$ denotes the smallest eigenvalue of the Hermitian 
  operator $A$. Consequently,
  \[
    C^{\mathbb{P}}(\mathcal{M},\OL{\mathcal{M}}) \leq 4\log\|P^{-1}\|_{\infty}.
  \]
%\hfill $\square$ % Followed by proof!
\end{proposition}

\begin{proof}
We begin with a test state $\varphi$ as in the above description of the most general 
non-adaptive strategy for the channels $\mathcal{M}$ and $\OL{\mathcal{M}}$, 
so that the two output states are
$\rho = (\id_R\otimes\mathcal{M})\varphi$, $\sigma = (\id_R\otimes\OL{\mathcal{M}})\varphi$.
By well-known inequalities \cite{761271}, it holds
\[
  \frac12 \| \rho-\sigma \|_1 \leq \sqrt{1-F(\rho,\sigma)^2} \leq 1 - \frac12 F(\rho,\sigma)^2, 
\]
where $F(\rho,\sigma) = \|\sqrt{\rho}\sqrt{\sigma}\|_1$ is the fidelity. 
Thus, it will be enough to lower bound the fidelity between the output states 
of the two channels. 
With $\tau = \tr_R \ketbra{\varphi}$, we have:
\[\begin{split}
  F(\rho,\sigma) &=    \|\sqrt{\rho}\sqrt{\sigma}\|_1 \\
                 &\geq \tr \sqrt{\rho}\sqrt{\sigma}   \\
                 &\geq \tr \rho\sigma\\
                 &=    \sum_{ij} |\tr E_i^\dagger F_j \tau |^2 \\
                 &\geq \left| \sum_{ij} \gamma_{ij} \tr E_i^\dagger F_j \tau \right|^2 \\
                 &=    |\tr \tau P|^2.
\end{split}\]
Here, the second line is by standard inequalities for the trace norm, 
the third is because of $\rho \leq \sqrt{\rho}$, the fourth is a formula
from \cite[Sec.~II]{7541701}, in the fifth we used Cauchy-Schwarz 
inequality and in the last line the definition of $P$. 
Since $\tau$, like $\varphi$, ranges over all states, we get 
\[
  F(\rho,\sigma)^2 \geq \lambda_{\min}(P)^4, 
\]
and so
\[
  P_e \geq \frac14 \lambda_{\min}(P)^4.
\]

We can apply the same reasoning to $\mathcal{M}^{\otimes n}$ and $\OL{\mathcal{M}}^{\otimes n}$,
for which the vector $(\gamma_{ij})^{\otimes n}$ is eligible and leads to the
positive definite operator $P^{\otimes n}$. Thus, 
\[
  P_{e,\mathbb{P}}^{(n)} \geq \frac14 \lambda_{\min}\left(P^{\otimes n}\right)^{4}
                         =    \frac14 \lambda_{\min}(P)^{4n}.
\]
Taking the limit and noting $\lambda_{\min}(P)^{-1} = \|P^{-1}\|_{\infty}$
concludes the proof. 
\end{proof}

\subsection{Two examples}
\begin{example}
\label{ex:Harrow-et-al}
Next we show that two channels defined by Harrow \emph{et al.} \cite{PhysRevA.81.032339} 
yield an example of a pair with $C^{\mathbb{P}}(\mathcal{M},\OL{\mathcal{M}}) < +\infty$,
yet $C^{\mathbb{A}}(\mathcal{M},\OL{\mathcal{M}}) = +\infty$ because indeed
$P_{e,\mathbb{A}}^{(2)} = 0$. 
In \cite{PhysRevA.81.032339}, the following two entanglement-breaking channels
from $A\otimes C=\mathbb{C}^2\otimes\mathbb{C}^2$ (two qubits) to 
$B=\mathbb{C}^2$ (one qubit) are considered:
\begin{align*}
  \mathcal{M}(\rho^A\otimes\gamma^C) &= \ketbra{0}{0} \bra{0}\gamma\ket{0}
                            + \ketbra{0}{0} \bra{1}\gamma\ket{1} \bra{0}\rho\ket{0}
                            + \frac12 I \bra{1}\gamma\ket{1} \bra{1}\rho\ket{1}, \\
  \OL{\mathcal{M}}(\rho^A\otimes\gamma^C) &= \ketbra{+}{+} \bra{0}\gamma\ket{0}
                            + \ketbra{1}{1} \bra{1}\gamma\ket{1} \bra{+}\rho\ket{+}
                            + \frac12 I \bra{1}\gamma\ket{1} \bra{-}\rho\ket{-}, 
\end{align*}
extended by linearity to all states.
Here, $\ket{0}, \ket{1}$ are the computational basis ($Z$ eigenbasis) of the qubits,
while $\ket{+}, \ket{-}$ are the Hadamard basis ($X$ eigenbasis). 

In words, both channels measure the qubit $C$ in the computational basis.
If the outcome is `0', they each prepare a pure state on $B$ (ignoring the input
in $A$): $\ketbra{0}{0}$ for $\mathcal{M}$, $\ketbra{+}{+}$ for $\OL{\mathcal{M}}$.
If the outcome is `1', they each make a measurement on $A$ and prepare an output 
state on $B$ depending on its outcome: standard basis measurement for $\mathcal{M}$
with $\ketbra{0}{0}$ on outcome `0' and the maximally mixed state $\frac12 I$ on 
outcome `1'; 
Hadamard basis measurement for $\OL{\mathcal{M}}$ with $\ketbra{1}{1}$ on outcome `+' and the 
maximally mixed state $\frac12 I$ on outcome `-'. 
In \cite{PhysRevA.81.032339}, a simple adaptive strategy for $n=2$ uses of the
channel is given that discriminates $\mathcal{M}$ and $\OL{\mathcal{M}}$ perfectly: 
The first instance of the channel is fed with $\ketbra{0}\otimes\ketbra{0}$,
resulting in an output state $\rho_1$; the second instance of the channel 
is fed with $\ketbra{1}\otimes\rho_1$; the output state $\rho_2$ of the second
instance is $\ketbra{0}$ if the unknown channel is $\mathcal{M}$, and 
$\ketbra{1}$ if the unknown channel is $\OL{\mathcal{M}}$, so a computational 
basis measurement reveals it. Note that no auxiliary system $R$ is needed, 
but the feed-forward nevertheless requires a qubit of quantum memory for the 
strategy to be implemented. In any case, this proves that $P_{e,\mathbb{A}}^{(2)} = 0$. 
In \cite{PhysRevA.81.032339}, it is furthermore proved that for all $n\geq 1$, 
$P_{e,\mathbb{P}}^{(n)} > 0$.

We now show that Proposition \ref{prop:bound} is applicable to 
yield an exponential lower bound on the non-adaptive error probability.
The Kraus operators of the two channels can be chosen as follows:
\begin{align*}
  \mathcal{M}: E_i\in\Bigl\{ \ket{0}^B & \bra{00}^{AC}, 
                                         & \quad \OL{\mathcal{M}}: F_j\in\Bigl\{ \ket{+}^B&\bra{00}^{AC}, \\
                  \ket{0}^B&\bra{10}^{AC},                     &                 \ket{+}^B&\bra{10}^{AC}, \\
                  \ket{0}^B&\bra{01}^{AC},                     &                 \ket{1}^B&\bra{+1}^{AC}, \\
                  \ket{0}^B&\bra{11}^{AC}/\sqrt{2},            &                 \ket{0}^B&\bra{-1}^{AC}/\sqrt{2}, \\
                  \ket{1}^B&\bra{11}^{AC}/\sqrt{2} \Bigr\},    &                 \ket{1}^B&\bra{-1}^{AC}/\sqrt{2} \Bigr\}.
\end{align*}
Thus, the products $E_i^\dagger F_j$ include the matrices 
\begin{align*}
  E_1^\dagger F_1 &= \sqrt{\frac12}\ketbra{00}{00}, \\
  E_2^\dagger F_2 &= \sqrt{\frac12}\ketbra{10}{10}, \\
  E_5^\dagger F_3 &= \sqrt{\frac12}\ketbra{11}{+1}, \\
  E_5^\dagger F_5 &=       \frac12 \ketbra{11}{-1}, \\
  E_3^\dagger F_4 &= \sqrt{\frac12}\ketbra{01}{-1}, 
\end{align*}
from which we can form, by linear combination, the operators
\begin{align*}
  E_1^\dagger F_1                                 &= \sqrt{\frac12}\ketbra{0}{0}\otimes\ketbra{0}{0}, \\
  E_2^\dagger F_2                                 &= \sqrt{\frac12}\ketbra{1}{1}\otimes\ketbra{0}{0}, \\
  \sqrt{\frac12}E_5^\dagger F_3 - E_5^\dagger F_5 &= \sqrt{\frac12}\ketbra{1}{1}\otimes\ketbra{1}{1}, \\
  \sqrt{\frac12}E_3^\dagger F_4 - E_5^\dagger F_5 &= \sqrt{\frac12}\ketbra{-}{-}\otimes\ketbra{1}{1}, 
\end{align*}
whose sum is indeed positive definite, so we get an exponential lower bound on 
$P_{e,\mathbb{P}}^{(n)}$ and hence a finite value of $C^{\mathbb{P}}(\mathcal{M},\OL{\mathcal{M}})$. 
To get a concrete upper bound on $C^{\mathbb{P}}(\mathcal{M},\OL{\mathcal{M}})$ from 
the above method, 
we choose $\gamma_{11}=\gamma_{22}=\alpha$ and $\gamma_{53}=\gamma_{34}=\gamma_{55}=\beta $,
with $\alpha,\beta>0$ and $2\alpha^2+5\beta^2=1$ 
and $\gamma_{i,j}=0$ for other cases in Proposition \ref{prop:bound}. 
Then, $P$ is written as
\[\begin{split}
  P &= \alpha E_1^\dagger F_1 + \alpha E_2^\dagger F_2 
        + \beta\sqrt{\frac12}E_5^\dagger F_3 + \beta\sqrt{\frac12}E_3^\dagger F_4 - 2\beta E_5^\dagger F_5 \\
    &= \alpha\sqrt{\frac12} I \otimes\ketbra{0} + \beta\sqrt{\frac12}(\ketbra{1}+\ketbra{-})\otimes\ketbra{1}, 
\end{split}\]
which implies the condition $P>0$.
Now $P$ is an orthogonal sum of two rank-two operators, i.e. as a $4\times 4$-matrix 
it has block diagonal structure with two $2\times 2$-blocks. Their minimum eigenvalues
are easily calculated:
they are $\alpha\sqrt{\frac12}$ and $\beta\sqrt{2}\sin^2\frac{\pi}{8}$. 
%\mh{(I could not find how to derive this statement. Could you write down its proof?)}
Since
$\lambda_{\min}(P)$ will be the smaller of the two, we optimize it by making the 
two values equal, i.e. we want $\alpha = 2\beta\sin^2\frac{\pi}{8}$. Inserting this in 
the normalization condition and solving for $\beta$ yields
$\beta^2 = \left(8\sin^4\frac{\pi}{8}+5\right)^{-1}$, thus
\[
  \lambda_{\min}(P) = \sqrt{\frac{2}{8\sin^4\frac{\pi}{8}+5}} \sin^2\frac{\pi}{8}
                    = \frac{2-\sqrt{2}}{4\sqrt{4-\sqrt{2}}} 
                    \approx 0.091,
\]
where we have used the identity $\sin^2\frac{\pi}{8} = \frac12(1-\sqrt{\frac12})$.
Hence, Proposition \ref{prop:bound} guarantees that
\[
%  \phantom{==================}
  C^{\mathbb{P}}(\mathcal{M},\OL{\mathcal{M}}) \leq 4\log\frac{4\sqrt{4-\sqrt{2}}}{2-\sqrt{2}} 
                                                 \approx 13.83.
%  \phantom{=================:}\blacksquare
\]

Note that a lower bound is the Chernoff bound of the two pure output
states $\ketbra{0}{0} = \mathcal{M}(\ketbra{00})$ 
and $\ketbra{+}{+} = \OL{\mathcal{M}}(\ketbra{00})$, 
which is $\log 2 = 1$, so
$C^{\mathbb{P}}(\mathcal{M},\OL{\mathcal{M}}) \geq 1$. It seems reasonable to conjecture that this 
is optimal, but we do not have at present a proof of it. 
\hfill$\square$
\end{example}

\medskip
\begin{example}
\label{ex:POVMs}
For later use, we briefly discuss another example due to Krawiec \emph{et al.} \cite{KPP:POVM}, 
which consists of two qc-channels implementing 
two rank-one POVMs on a qutrit $A$, and the output $Y$ is a nine-dimensional
Hilbert space.
They are given by vectors $\ket{x_i}\in A$ and $\ket{y_i}\in A$ ($i=1,\ldots,9$)
such that $\sum_{i=1}^9 \ketbra{x_i} = \sum_{j=1}^9 \ketbra{y_j} = I$:
\begin{equation}
  \mathcal{P}(\rho) = \sum_{i=1}^9 \bra{x_i}\rho\ket{x_i} \ketbra{i},
  \quad
  \OL{\mathcal{P}}(\rho) = \sum_{j=1}^9 \bra{y_j}\rho\ket{y_j} \ketbra{j}. 
\end{equation}
The Kraus operators are $E_i=\ketbra{i}{x_i}$ and $F_j=\ketbra{j}{y_j}$,
which makes it easy to calculate 
$\operatorname{span}\{E_i^\dagger F_j\} = \operatorname{span}\{\ketbra{x_i}{y_i}\}$.

In \cite{KPP:POVM} it is shown how to choose the two POVMs in such a way that 
this subspace does not contain the identity $I$ and indeed satisfies the 
``disjointness'' condition 
of Duan \emph{et al.} \cite{PhysRevLett.103.210501} for perfect finite-copy 
distinguishability of the two channels using adaptive strategies. Thus,
$C^{\mathbb{A}}(\mathcal{P},\OL{\mathcal{P}}) = +\infty$.
On the other hand, it is proven in \cite{KPP:POVM} that the subspace contains 
a positive definite matrix $P>0$. Hence, Proposition \ref{prop:bound} guarantees that
$C^{\mathbb{P}}(\mathcal{P},\OL{\mathcal{P}}) < +\infty$.
\hfill$\square$
\end{example}

\medskip 
So indeed there are channels, entanglement-breaking channels
at that, for which the adaptive and the non-adaptive Chernoff exponents 
are different; in fact, the separation is maximal, in that the former is 
$+\infty$ while the latter is finite: They lend themselves easily to experiments, 
as the channels of Example \ref{ex:Harrow-et-al} are composed of simple qubit measurement and 
state preparations.  It should be noted that this separation 
is a robust phenomenon, and not for example related to the perfect finite-copy
distinguishability. Namely, by simply mixing our example channels with the same 
small fraction $\epsilon>0$ of the completely depolarizing channel $\tau$, 
we get two new channels 
$\mathcal{M}' = (1-\epsilon)\mathcal{M}+\epsilon\tau$ and 
$\OL{\mathcal{M}}' = (1-\epsilon)\OL{\mathcal{M}}+\epsilon\tau$
with only smaller non-adaptive Chernoff bound, 
$C^{\mathbb{P}}(\mathcal{M}',\OL{\mathcal{M}}') \leq C^{\mathbb{P}}(\mathcal{M},\OL{\mathcal{M}}) < +\infty$. 
As shown below, by choosing a suitable $\epsilon>0$,
the fully general adaptive strategy satisfies
\begin{align}
C^{\mathbb{P}}(\mathcal{M},\OL{\mathcal{M}})< C^{\mathbb{A}}(\mathcal{M}',\OL{\mathcal{M}}')<
\infty. \label{XL1}
\end{align}
This case gives an example for 
the asymptotic separation between adaptive and non-adaptive strategies
even with a finite exponent for adaptive strategy.

Now, we show the existence of $\epsilon>0$ to satisfy \eqref{XL1}.
Because $C^{\mathbb{A}}(\mathcal{M}',\OL{\mathcal{M}}')$ goes to infinity as $\epsilon$ goes to zero,
there exists $\epsilon>0$ to satisfy the first inequality in \eqref{XL1}.
On the other hand, 
the relation $C^{\mathbb{A}}(\mathcal{M}',\OL{\mathcal{M}}') < +\infty$ with an arbitrary $\epsilon >0$
can be shown in the following way.
Because the Kraus operators $\{E_i^{\prime}\}$ and $\{F_j'\}$
of the channels satisfy $I \in \operatorname{span}\{E_i^{\prime\dagger}F_j'\}$,
Duan \emph{et al.} \cite{PhysRevLett.103.210501} guarantees that 
$\mathcal{M}'$ and $\OL{\mathcal{M}}'$ are not perfectly distinguishable under
any $\mathbb{A}_n$ for any finite $n$.
Applying this fact to the result by Yu and Zhou \cite{2017arXiv170501642Y},
we find the existence of a finite upper bound on the Chernoff exponent $C^{\mathbb{A}}(\mathcal{M}',\OL{\mathcal{M}}')$. 

Furthermore, since the error rate tradeoff function $B_e^{\mathbb{P}}(r|\mathcal{M}\|\OL{\mathcal{M}})$ 
is continuous near $r=C^{\mathbb{P}}(\mathcal{M},\OL{\mathcal{M}})$, whereas the
adaptive variant $B_e^{\mathbb{A}}(r|\mathcal{M}\|\OL{\mathcal{M}})$ is infinite
everywhere, we automatically get separations in the Hoeffding setting, as well. 
Note that there is no contradiction with the results of \cite{PhysRevResearch.1.033169,berta2018amortized}, 
which showed equality of the adaptive and the non-adaptive Stein's exponents, 
which are indeed both $+\infty$: for the non-adaptive one this follows from 
the fact that the channels on the same input prepare different pure states,
$\ketbra{0}{0}$ for $\mathcal{M}$, $\ketbra{+}{+}$ for $\OL{\mathcal{M}}$.

%\label{sec:classicalfeeback}
\section{Responsible resources for quantum advantage}
\label{S-5}
We showed in Section \ref{lowerbound} that quantum feed-forward can improve the error 
exponent in the symmetric and Hoeffding settings for the discrimination of two qq-channels.
This result followed by investigating a pair of entanglement-breaking channels 
introduced in \cite{PhysRevA.81.032339}, and a pair of qc-channels from \cite{KPP:POVM}.

In contrast, the present section investigates which features of general feed-forward 
strategies is responsible for this advantage, and conversely, which restricted feed-forward 
strategies cannot improve the error exponents for discrimination of two qq-channels.
To address this question, we first import the results on cq-channels from 
Section \ref{secadaptive} to a special class of qq-channels.
\fs{Note that if $\mathcal{X}$ is discrete, i.e. either finite or countably infinite, 
with the atomic (power set) Borel algebra, so that arbitrary mappings 
$\mathcal{N}:x\rightarrow \rho_{x}$ and $\overline{\mathcal{N}}:x\rightarrow \sigma_{x}$
define cq-channels, we can think of them as special, entanglement-breaking, 
qq-channels 
$\mathcal{M},\OL{\mathcal{M}}:\mathcal{T}^{\mathcal{X}} \rightarrow \mathcal{T}^{\mathcal{B}}$:
\begin{equation}
  \mathcal{M}(\xi)      = \sum_{x \in \mathcal{X}} \rho_x \Tr \xi E_x, 
    \quad 
  \OL{\mathcal{M}}(\xi) = \sum_{x\in \mathcal{X}} \sigma_x \Tr \xi E_x, \label{MFC}
\end{equation} 
where $\{E_x\}_{x \in \mathcal{X}}$ is a PVM of rank-one projectors $E_x=\ketbra{x}$, 
and $\mathcal{X}$ labels an orthonormal basis $\{\ket{x}\}_{x\in\mathcal{X}}$ of 
a separable Hilbert space, denoted $\mathcal{X}$, too (cf. Eq. \ref{PVM-ex}).}

In particular, using the results of Section \ref{secadaptive},
we will show the following fact for discrimination of special entanglement-breaking channels given by Eq. (\ref{MFC});
The most general class of adaptive strategies $\mathbb{A}_{n}$
offers no gain over the weakest class of strategies $\mathbb{P}_{n}^{0}$,
i.e. non-adaptive strategies without entangled input,
even though this class uses entangled input and quantum feed-forward.
Since in the analysis of cq-channels it turns out that the most general 
strategy does not use quantum memory at the input and feed-forward that is 
classical, we are motivated to consider this restricted class of adaptive 
strategies for general qq-channels, denoted $\mathbb{A}_{n}^{c,0}$, in 
Subsection \ref{subsec:A-c-0}. 
We will show that this subclass of adaptive strategies offers no gain over 
non-adaptive strategies without quantum memory at the input. 
Finally, in Subsection \ref{subsec:weaker-restriction} we consider whether it 
is really necessary to impose both the restriction of no input quantum memory 
and classical feed-forward to rule out an advantage for non-adaptive strategies. 
Indeed, we shall show that the examples considered in Section \ref{lowerbound} 
demonstrate asymptotic advantages both for adaptive strategies with no input 
quantum memory but quantum feed-forward (``$\mathbb{A}_{n}^{0}$'') 
and for adaptive strategies with quantum memory at the input and classical 
feed-forward (``$\mathbb{A}_{n}^{c}$'').

\begin{figure}[ht]
\begin{center}
\includegraphics[width=0.4\textwidth]{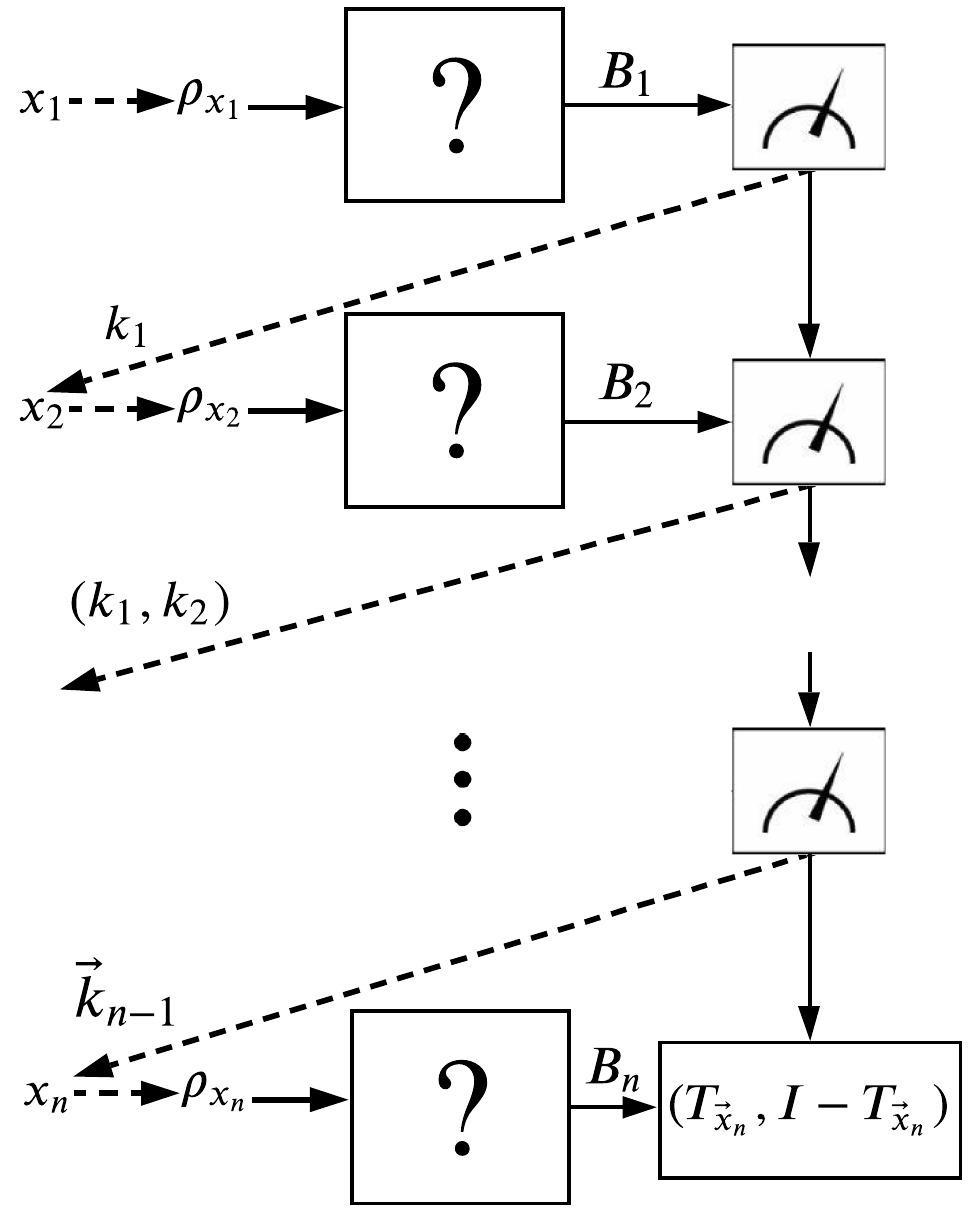}      
\caption{Adaptive quantum channel discrimination with classical feed-forward 
         and without quantum memory at the channel input, from the class $\mathbb{A}^{c,0}_n$. 
         Solid and dashed arrow denote the flow of quantum and classical information, 
         respectively. At step $m$, Alice sends the state $\rho_{{x}_{m}}$
         which she has prepared using Bob's $m-1$ classical feed-forward 
         informations, and sends it via either $\mathcal{M}$ or $\OL{\mathcal{M}}$ to Bob.}
\label{FA}
\end{center}
\end{figure}

\subsection{Discrimination of cq-channels as cptp maps under $\mathbb{A}_{n}$ strategies}
%\label{SV-B}
\label{subsec:cq-as-qq}
The most general class $\mathbb{A}_{n}$ 
of strategies to distinguish two qq-channels $\mathcal{M}$ and $\overline{\mathcal{M}}$
is the set of strategies given in Definition \ref{Def1}.
For this class, recall that we denote 
the generalized Chernoff and Hoeffding quantities as 
$C^{\mathbb{A}}(a,b|\mathcal{M}\|\overline{\mathcal{M}})$ and 
$B_{e}^{\mathbb{A}}(r|\mathcal{M}\|\overline{\mathcal{M}})$, respectively. 
In this subsection,
we discuss the effect of input entanglement for our cq-channel discrimination strategy,
when the input alphabet is discrete. Recall the form \eqref{MFC} of the two 
channels as qq-quantum channels.
\if0
\[
  \mathcal{M}(\xi)            = \sum_{x \in \mathcal{X}} \rho_x \Tr \xi E_x, \quad 
  \overline{\mathcal{M}}(\xi) = \sum_{x \in \mathcal{X}} \sigma_x \Tr \xi E_x, 
\]
where $E_x = \ketbra{x}$ form a PVM of rank-one projectors.
\fi

In this case, the most general strategy stated in Definition \ref{Def1}
for the discrimination of two qq-channels $\mathcal{M}$ and $\overline{\mathcal{M}}$
can be converted to the strategy stated in Subsection \ref{S4-A-1}
for the discrimination of two cq-channels 
$\mathcal{N}: x \mapsto \rho_x$ and 
$\overline{\mathcal{N}}: x \mapsto \sigma_x$
as follows.
In the general strategy for qq-channel, the operation in the $m$-th step is given as 
a quantum channel $\mathcal{F}_m : R_{m}B_m \rightarrow R_{m+1}A_{m+1}$.
To describe the general strategy for cq-channel,
we define the quantum instrument 
$\mathcal{E}_m: R_m B_m \to X_m R_{m+1}$ in the sense of \eqref{NAC} 
as
\begin{align}
\mathcal{E}_m(\xi) := \sum_{x_m \in {\cal X}}
 |x_m\rangle \langle x_m|
\otimes
(\Tr_{A_{m+1}} E_{x_m} \mathcal{F}_m(\xi) ).
\end{align}
Then, 
the general strategy for cq-channel is given 
as applying the above quantum instrument and 
choosing the obtained outcome $x_m$ as the input of the cq-channel to be discriminated.
The final states in the general strategy for qq-channel is the same as the final state in the general strategy for cq-channel.
That is, the performance of the general strategy for these two qq-channels
is the same as 
the performance of the general strategy for the above defined cq-channels.
This fact means that
the adaptive method does not improve the performance of the discrimination 
of the channels \eqref{MFC}.

Furthermore, when the quantum channel $\mathcal{F}_m$ in the strategy is replaced 
by the channel $\mathcal{F}_m'$ defined as 
$\mathcal{F}_m'(\xi) := \sum_{x_m} E_{x_m}\mathcal{F}_m(\xi)E_{x_m}$,
we do not change the statistics of the protocol for either channel. 
Since the output of $\mathcal{F}_m'$ has no entanglement between $X_m$ and $R_{m+1}$,
the presence of input entanglement does not improve the performance in this case.

\medskip
To state the next result, define for two quantum channels $\mathcal{M}$ 
and $\overline{\mathcal{M}}$ mapping $A$ to $B$, 
\begin{align}
  \label{eq:D-channel}
  D(\mathcal{M}\|\OL{\mathcal{M}})
          &\coloneqq \sup_{\rho\in\mathcal{S}^A} D(\mathcal{M}(\rho)\|\OL{\mathcal{M}}(\rho)), \text{ and} \\
  \label{eq:D-alpha-channel}
  D_\alpha(\mathcal{M}\|\overline{\mathcal{M}})
          &\coloneqq \sup_{\rho\in\mathcal{S}^A} D_{\alpha}(\mathcal{M}(\rho)\|\overline{\mathcal{M}}(\rho)).
\end{align}

%\medskip
\begin{theorem}
\label{en-br-dp3}
Assume that two qq-quantum channels $\mathcal{M}$ and 
$\overline{\mathcal{M}}$ are given by Eq. \eqref{MFC}.
For $0\leq r\leq D(\mathcal{M}\|\OL{\mathcal{M}})$ 
%[see Eq. \eqref{eq:D-channel}] 
and real $a$ and $b$ 
with $-D(\mathcal{M}\|\OL{\mathcal{M}})\le a-b \le D(\OL{\mathcal{M}}\|\mathcal{M})$, 
the following holds:
\begin{align*}
\phantom{==========}
C^{\mathbb{A}}(a,b|\mathcal{M}\|\overline{\mathcal{M}})
  %&= C^{\mathbb{A}^{c,0}}(a,b|\mathcal{M}\|\overline{\mathcal{M}})
   &= C^{\mathbb{P}^{0}}(a,b|\mathcal{M}\|\overline{\mathcal{M}})
   = C(a,b|\mathcal{N}\|\overline{\mathcal{N}}), \\
B_{e}^{\mathbb{A}}(r|\mathcal{M}\|\overline{\mathcal{M}})
 %&= B_{e}^{\mathbb{A}^{c,0}}(r|\mathcal{M}\|\overline{\mathcal{M}})
   &= B_{e}^{\mathbb{P}^{0}}(r|\mathcal{M}\|\overline{\mathcal{M}})
   = B_{e}(r|\mathcal{N}\|\overline{\mathcal{N}}).
   \phantom{==========}\square 
\end{align*}
%\hfill $\square$
\end{theorem}

\medskip
Note that it was essential that not only the channels are entanglement-breaking, 
but that the measurement $\{E_x\}$ is a PVM, and in fact the same PVM for both 
channels. The discussion fails already when the channels each have their own PVM,
which are non-commuting. Indeed, such channels were essential to the counterexample in 
Section \ref{lowerbound}, Example \ref{ex:Harrow-et-al}, showing a genuine 
advantage of general adaptive strategies. 
In this case, the construction of  the channel $\mathcal{F}_m'$ depends on the choice 
of the hypothesis. Therefore, the condition \eqref{MFC} is essential for this discussion.

Furthermore, if the channels are entanglement-breaking, but with a general POVM
in Eq. \eqref{MFC}, i.e. the $E_x$ are not orthogonal projectors, 
the above discussion does not hold, either. Indeed, the second counterexample in 
Section \ref{lowerbound}, Example \ref{ex:POVMs}, consists of qc-channels 
implementing overcomplete rank-one measurements, once more showing a genuine 
advantage of general adaptive strategies.
In this case, the output state is separable, but it cannot be necessarily simulated 
by a separable input state.

\medskip
\begin{remark}
The discussion of this section shows that without loss of generality,
we can assume that the measurement outcome equals the next input
when ${\cal X}$ is discrete.
That is, it is sufficient to consider the case when $k_m=x_m$.
This fact can be shown as follows.
Given two cq-channels $x \mapsto \rho_x$ and $x \mapsto \sigma_x$,
we define two entanglement-breaking channels
$\mathcal{M}$ and $\overline{\mathcal{M}}$ by Eq. \eqref{MFC}.
For the case with two qq-channel,
the most general strategy is given in Definition \ref{Def1}.
For two cq-channels $\mathcal{M}:x \mapsto \rho_x$ and $\OL{\mathcal{M}}:x \mapsto \sigma_x$,
the most general strategy can be simulated by an instrument with $k_m=x_m$.

However, when $\mathcal{X}$ is not discrete, neither can we view the 
cq-channels as special qq-channels (as the Definition in Eq. \eqref{MFC} only 
makes sense for discrete $\mathcal{X}$), nor do we allow arbitrary, 
only discrete feed-forward; hence, to cover the case with continuous 
$\mathcal{X}$, we need to address it using general outcomes $k_m$
as in Section \ref{secadaptive}.
\hfill $\square$
\end{remark}

\subsection{Restricting to classical feed-forward and no quantum memory at the input: $\mathbb{A}_{n}^{c,0}$}
\label{subsec:A-c-0}
In this setting, the protocol is similar to the adaptive protocol described in 
Section \ref{secadaptive}, but extended to general quantum channels (see Fig. \ref{FA}):
after each transmission, the input state $\rho_{x_m}$ is chosen adaptively 
from the classical feed-forward. Denoting this adaptive choice of input states as 
$\vec{x}_{m}=(x_{1},\ldots,x_{m})$, the $m$-th input is chosen conditioned on 
the feed-forward information $\vec{k}_{m-1}$ and $\vec{x}_{m-1}$ 
from the conditional distribution
$p_{X_{m}|\vec{X}_{m-1},\vec{K}_{m-1}}(x_{m}|\vec{x}_{m-1},\vec{k}_{m-1})$.

\medskip
\begin{theorem}
\label{fawm}
Let $\mathcal{M}$ and $\overline{\mathcal{M}}$ be qq-channels. 
Then, for real numbers $a,b$ satisfying
$-D(\mathcal{M}\|\overline{\mathcal{M}})\le a-b \le D(\overline{\mathcal{M}}\|\mathcal{M})$ 
and any $0\leq r\leq D(\mathcal{M}\|\overline{\mathcal{M}})$, it holds
\begin{align*}
C^{\mathbb{A}^{c,0}}(a,b|\mathcal{M}\|\overline{\mathcal{M}})
&=C^{\mathbb{P}^{0}}(a,b|\mathcal{M}\|\overline{\mathcal{M}})
%&=\sup_{\rho}\sup_{0\leq\alpha\leq 1} (1-\alpha)
%D_{\alpha}(\mathcal{M}(\rho)\|\overline{\mathcal{M}}(\rho))
%-\alpha a -(1-\alpha)b\\
 =\sup_{0\leq\alpha\leq 1} (1-\alpha)D_{\alpha}(\mathcal{M}\|\overline{\mathcal{M}})
-\alpha a -(1-\alpha)b, \\
B_{e}^{\mathbb{A}^{c,0}}(r|\mathcal{M}\|\overline{\mathcal{M}})
&=B_{e}^{\mathbb{P}^{0}}(r|\mathcal{M}\|\overline{\mathcal{M}}) 
%&=\sup_\rho \sup_{0\le \alpha \le 1}
%\frac{\alpha-1}{\alpha}\big(r-D_{\alpha }(\mathcal{M}(\rho)\|\overline{\mathcal{M}}(\rho))\big)\\
 =\sup_{0\le \alpha \le 1}
\frac{\alpha-1}{\alpha}\big(r- D_{\alpha }(\mathcal{M}\|\overline{\mathcal{M}})\big),
\end{align*}
where $D(\mathcal{M}\|\overline{\mathcal{M}})$ and $D_{\alpha }(\mathcal{M}\|\overline{\mathcal{M}})$ 
are defined in Eqs. \eqref{eq:D-channel} and \eqref{eq:D-alpha-channel}.
%\hfill $\square$ % Followed by proof!
\end{theorem}

\begin{proof}
Since only classical feed-forward is allowed, one can cast this 
discrimination problem  in the framework of the cq-channel discrimination problem
treated in Section \ref{secadaptive}. 
Namely, we apply Theorem \ref{chernoff} %and Corollary \ref{hoeffding}
to the case when the cq-channels have input alphabet ${\cal X}=\mathcal{S}^{A}$, 
i.e. it equals the set of all states on the input systems.
In other words, we choose the classical (continuous) input alphabet as $\mathcal{X}$, 
where each letter $x\in\mathcal{X}$ is a classical description of a state $\xi$
on the input system $A$. 
In this application, $\rho_x$ and $\sigma_x$ are
given as $\mathcal{M}(\xi)$ and $\overline{\mathcal{M}}(\xi)$, respectively, 
for $x\equiv\xi$.
Hence, 
$\sup_{x} D_{\alpha}(\rho_{x}\|\sigma_{x})$ equals
$D_\alpha(\mathcal{M}\|\overline{\mathcal{M}})
=\sup_{\xi}D_{\alpha}\bigl(\mathcal{M}(\xi)\|\overline{\mathcal{M}}(\xi)\bigr)$. 
Hence, the desired relation is obtained.
\end{proof}

\medskip
\begin{remark}
The above theorem concludes that in the absence of entangled inputs, no 
adaptive strategy built upon classical feed-forward can outperform the best 
non-adaptive strategy, which is in fact a tensor power input. 
In other words, the optimal error rate can be achieved by a simple i.i.d. 
input sequence where all $n$ input states are chosen to be the same: $\rho^{\otimes n}$. 
\hfill $\square$
\end{remark}

\subsection{No advantage of adaptive strategies beyond $\mathbb{A}_n^{c,0}$?}
\label{subsec:weaker-restriction}
One has to wonder whether it is really necessary to impose classical feed-forward
\emph{and} to rule out quantum memory at the channel input to arrive at the 
conclusion of Theorem \ref{fawm}, that non-adaptive strategies with 
tensor product inputs, $\mathbb{P}_n^0$, are already optimal. What can we say
when only one of the restrictions holds? We start with defining the 
class $\mathbb{A}_{n}^{c}$ of adaptive strategies using classical feed-forward.

\medskip
\begin{definition}
\label{defi:A-n-c}
The class $\mathbb{A}_{n}^{c}$ of adaptive strategies  
using classical feed-forward is defined as a subset of $\mathbb{A}_{n}$ given 
in Definition \ref{Def1}, where now the maps $\mathcal{F}_m$ are subject 
to an additional structure. 
To describe it, one has to distinguish two operationally different quantum 
memories, the systems $R_m$ of the sender, and systems $C_m$ of the receiver.
The initial state is $\rho_1^{R_1A_1}$, with trivial system $C_1=1$.
Then, $\mathcal{F}_m$ maps $R_mB_mC_m$ to $R_{m+1}A_{m+1}C_{m+1}$, in the following
way: 
\begin{equation}
  \mathcal{F}_m = \sum_{\vec{k}_m} \mathcal{F}_{\vec{k}_m}\otimes\mathcal{P}_{\vec{k}_m}, 
\end{equation}
where $\{\mathcal{F}_{\vec{k}_m}\}_{\vec{k}_m}$ is an instrument of cp maps 
mapping $B_mC_m$ to $C_{m+1}$ (this is the measurement of the channel outputs up 
to the $m$-th channel use generating the classical feed-forward, 
together with the evolution of the receiver's memory), 
and where all the $\mathcal{P}_{\vec{k}_m}$ are quantum channels mapping $R_m$ to $R_{m+1}A_{m+1}$,  
which serve to prepare the next channel input.

The class $\mathbb{A}_n^{c,0}$ is now easily identified as the subclass of 
strategies in $\mathbb{A}_n^{c}$ where $R_m=1$ is trivial throughtout the protocol.
\hfill $\square$
\end{definition}

\medskip
\begin{remark}
\label{rem:A-n-0}
Regarding adaptive strategies with quantum feed-forward, but no quantum memory 
at the input, which class might be denoted $\mathbb{A}_n^0$: 
Note that the adaptive strategy considered in Section \ref{lowerbound},
Example \ref{ex:Harrow-et-al}, 
that is applied to a pair of entanglement-breaking channels and shown to be 
better than any non-adaptive strategies, while actually using quantum feed-forward, 
required however no entangled inputs nor indeed quantum memory at the channel input. 
This shows that quantum feed-forward alone can be responsible for an
advantage over non-adaptive strategies. 
\hfill $\square$
\end{remark}

\medskip
\begin{remark}
Regarding adaptive strategies with classical feed-forward, however 
allowing quantum memory at the input, i.e.~$\mathbb{A}_n^c$: 
It turns out that the channels considered in Section \ref{lowerbound},
Example \ref{ex:POVMs}, show that this class offers an advantage over 
non-adaptive strategies. This is because they are qc-channels, i.e. their 
output is already classical, and so any general quantum feed-forward protocol 
can be reduced to an equivalent one with classical feed-forward. It can be 
seen, however, that the perfect adaptive discrimination protocol described 
in \cite{KPP:POVM} relies indeed on input entanglement. 
\hfill $\square$
\end{remark}

\section{Discrimination power of a quantum channel}
\label{power}
In this section we study how well a pair of quantum states can be distinguished after 
passing through a quantum channel. This quantifies the power of a quantum channel when
it is seen as a measurement device. In some sense this scenario is dual to the state
discrimination problem in which a pair of states are given and the optimization is taken
over all measurements, whilst in the current scenario a quantum channel is given and the
optimization takes place over all pairs of states passing through the channel.
The reference \cite{Hirche} studies the special case of qc-channels, that is
investigation of the power of a quantum detector given by a specific POVM in discriminating 
two quantum states.
It was shown in the paper that when the qc-channel is available asymptotically many times,
neither entangled state inputs nor classical feedback and adaptive choice of inputs 
can improve the performance of the channel.
We extend the model of the latter paper to general quantum channels, 
considering whether adaptive strategies provide an advantage for the 
discrimination power;
see Fig. \ref{FD}, where we consider classical feedback without quantum 
memory at the sender's side.

\begin{figure}[ht]
\begin{center}
\includegraphics[width=0.5\textwidth]{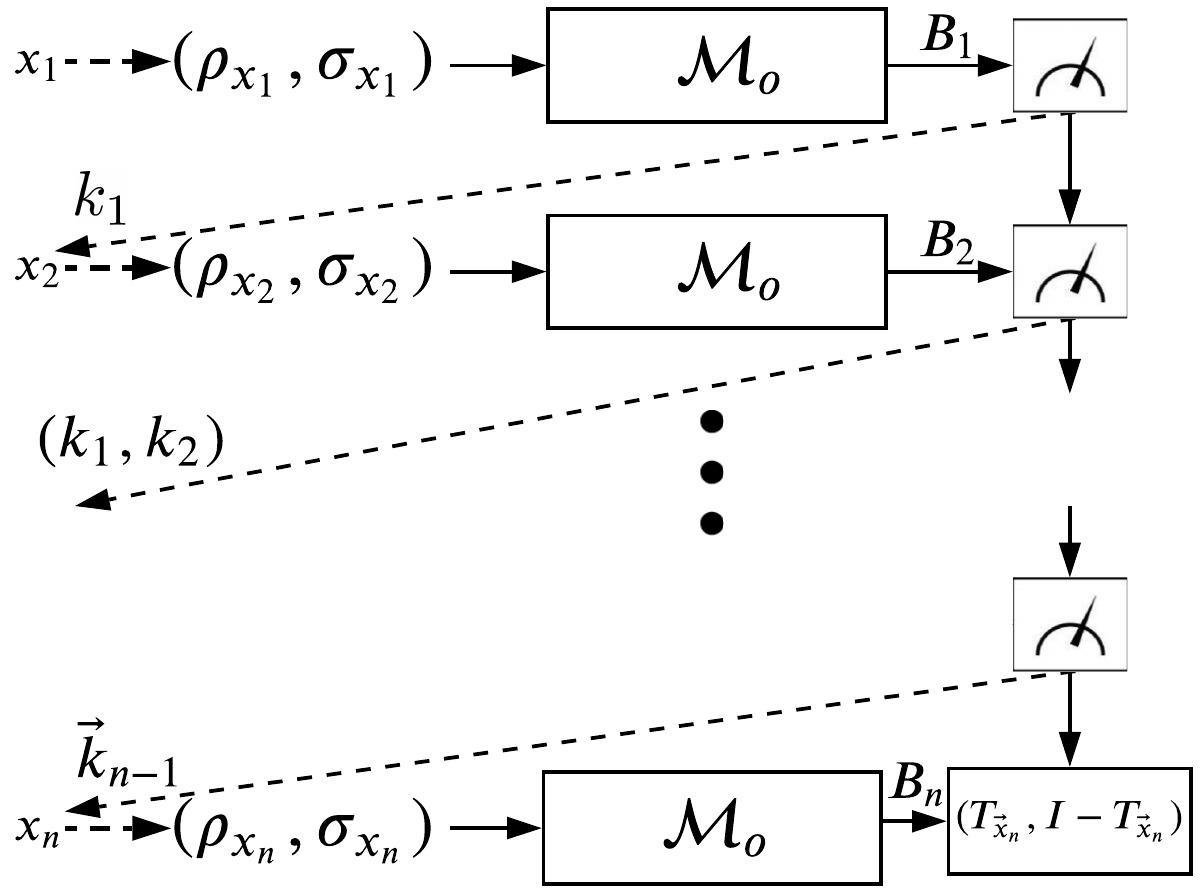}      
\caption{Discrimination with a quantum channel $\mathcal{M}_o$. At step $m$, Alice prepares a state, 
         either $\rho_{{x}_{m}}$ or $\sigma_{{x}_{m}}$, which she has prepared using Bob's 
         $m-1$ feedbacks (dashed arrows), and sends it via the channel $\mathcal{M}_{o}$ to Bob. 
         Bob's measurements resembles the PVM's of Section \ref{secadaptive}; they are used 
         to extract classical information fed back to Alice and to prepare post-measurement
         states that he keeps for the next round of communication.}
\label{FD}
\end{center}
\end{figure}

\subsection{Simple extension of \cite{Hirche} with classical feedback: $\bar{\mathbb{A}}^{c,0}$}
\label{MO1}
It is useful to cast this hypothesis testing setting as a communication problem as follows:
Assume a quantum channel 
$\mathcal{M}_o=\mathcal{M}_o^{A_o\rightarrow B}$ connects Alice and Bob, where 
Alice posses two systems $A_{o,0},A_{o,1}$ and Bob has $B$. 
They did not know which system of $A_{o,0},A_{o,1}$ works as the input system of 
the quantum channel $\mathcal{M}_o=\mathcal{M}_o^{A_o\rightarrow B}$.
Hence, the hypothesis
$H_0$($H_1$) refers to the case when the system $A_{o,0}$ ($A_{o,1}$)
is the input system of the quantum channel $\mathcal{M}_o$ and the remaining system 
is simply traced out, i.e., is discarded.
To identify which hypothesis is true, 
Alice and Bob make a collaboration.
Alice chooses two input states 
$\rho$ and $\sigma$ on $A_{o,0}$ and $A_{o,1}$ for this aim.
Bob receives the output state.
They repeat this procedure $n$ times.
Bob obtains the $n$-fold tensor product system of $B$ whose state
is $\mathcal{M}_o(\rho)^{\otimes n}$ or $\mathcal{M}_o(\sigma)^{\otimes n}$.
Applying two-outcome POVM $\{T_{n}, I-T_{n}\}$ on the $n$-fold tensor product system,
Bob infers the variable $Z$ by estimating which state is the true state.
In this scenario, to optimize the discrimination power, 
Alice chooses the best two input states 
$\rho$ and $\sigma$ on $A_{o,0}$ and $A_{o,1}$ for this aim.
We denote this class of Bob's strategies by $\bar{\mathbb{P}}^{0}$.

Now, in a similar to the reference \cite{Hirche}, we consider an adaptive strategy.
To identify which hypothesis is true, 
Alice and Bob make a collaborating strategy, which 
allows Alice to use the channel $n$ times and also allows Bob, 
who has access to quantum memory, 
to perform any measurement of his desire on its received systems 
and send back classical information to Alice;
then Alice chooses a suitable pair of states $\rho$ and $\sigma$ on two systems 
$A_{o,0}$  and $A_{o,1}$ adaptively based on the feedback that she receives after each transmission. 
We denote this class of adaptive strategies by $\bar{\mathbb{A}}^{c,0}$.

The adaptive strategy in the class $\bar{\mathbb{A}}_{n}^{c,0}$
follows the cq-channel discrimination strategy: 
denoting the input generically as $x_1, \ldots, x_n$,
the sequence of Bob's measurements is given as $\{\Pi^{(m)}_{\vec{k}_{m}|\vec{x}_{m}}\}_{m=1}^{n-1}$
and the classical feedback depends on the previous information $\vec{x}_{m},\vec{k}_{m-1}$, and
Alice's adaptive choice of the input states 
$(\rho_{x_1},\sigma_{x_1}), \ldots,(\rho_{x_n},\sigma_{x_n})$
labeled by
$(x_1, \ldots, x_n)$
 is given as
the sequence of conditional randomized choice 
$\{p_{{X}_m|\vec{X}_{m-1},\vec{K}_{m-1}}\}_{m=1}^n$ 
of the pair of the input states $(\rho,\sigma)$.
In this formulation, after obtaining the measurement outcome ${K}_{m}$, 
using a two-outcome POVM $\{T_{n}, I-T_{n}\}$ on the $n$-fold tensor product system,
Bob decides which hypothesis of $H_0$ and $H_1$ is true
according to the conditional distributions $\{p_{{X}_m|\vec{X}_{m-1},\vec{K}_{m-1}}\}_{m=1}^n$.

In this class, we denote 
the generalized Chernoff and Hoeffding quantities as 
${C}^{\bar{\mathbb{A}}^{c,0}}(a,b|\mathcal{M}_o)$ and 
${B}_{e}^{\bar{\mathbb{A}}^{c,0}}(r|\mathcal{M}_o)$, respectively. 
{When no feedback is allowed
and the input state deterministically is fixed to 
a single form $(\rho,\sigma)$, 
we denote the generalized Chernoff and Hoeffding quantities as 
${C}^{\bar{\mathbb{P}}^{0}}(a,b|\mathcal{M}_o)$ and 
${B}_{e}^{\bar{\mathbb{P}}^{0}}(r|\mathcal{M}_o)$, respectively.}

We set 
\begin{align}
  \label{cq-power}
{D}(\mathcal{M}_o):= \max_{{\rho,\sigma}}D(\mathcal{M}_o(\rho)\|\mathcal{M}_o(\sigma))
                   = \max_{{\rho,\sigma}}D(\mathcal{M}_o(\sigma)\|\mathcal{M}_o(\rho)).
\end{align}

\begin{theorem}\label{MNT}
Let $0\leq r\leq {D}(\mathcal{M}_o)$ and real numbers $a$ and $b$ satisfy 
$-{D}(\mathcal{M}_o)\le a-b \le {D}(\mathcal{M}_o)$, then
we have
\begin{align*}
  {C}^{\bar{\mathbb{A}}^{c,0}}(a,b|\mathcal{M}_o)
    = {C}^{\bar{\mathbb{P}}^{0}}(a,b|\mathcal{M}_o)
   &= \sup_{\rho,\sigma} \sup_{0\le \alpha \le 1} 
        (1-\alpha)D_{\alpha}(\mathcal{M}(\rho)\|\mathcal{M}_o(\sigma)) -\alpha a -(1-\alpha)b, \\
  {B}_{e}^{\bar{\mathbb{A}}^{c,0}}(r|\mathcal{M}_o)
    = {B}_{e}^{\bar{\mathbb{P}}^{0}}(r|\mathcal{M}_o)
   &= \sup_{\rho,\sigma} \sup_{0\le \alpha \le 1}
        \frac{\alpha-1}{\alpha}\big(r-D_{\alpha}(\mathcal{M}_o(\rho)\|\mathcal{M}_o(\sigma))\big).
\end{align*}
%\hfill $\square$ % Followed by proof!
\end{theorem}

\begin{proof}
Here we only need to consider the set 
$\mathcal{S}^{A_o}\times\mathcal{S}^{A_o'}$ of 
pairs of input states as the set $\cal X$. 
In other words, we choose the classical (continuous) input alphabet as $\mathcal{X}=$
$\mathcal{S}^{A_o}\times\mathcal{S}^{A_o'}$, 
where each letter $x \equiv (\rho,\sigma)\in\mathcal{X}$ is a classical description of the 
pair of states $(\rho,\sigma)$. Then the result follows from the adaptive protocol 
in Section \ref{secadaptive}. (Compare also the proof of Theorem \ref{fawm}.)
\end{proof}

\begin{remark}\label{rem}
As another scenario,
the reference \cite{Hirche} also considers the case when the input states 
are given as $\rho_n,\sigma_n$ where
$\rho_n$ and $\sigma_n$ are states on the $n$-fold systems 
$A_{o,0}^{\otimes n} $ and $A_{o,1}^{\otimes n} $, respectively.
It shows that 
this strategy can be reduced to the adaptive strategy presented in this subsection
when the channel $\mathcal{M}_o$ is a qc-channel, i.e.,
it has the form;
\begin{align}
\mathcal{M}_o(\rho)= \sum_{x \in {\cal X}} (\Tr \rho M_x)|x\rangle \langle x| ,
\end{align}
where $\{|x\rangle\}_{x \in {\cal X}}$ forms an orthogonal system.
However, it is not so easy to show the above reduction when 
the channel $\mathcal{M}_o$ is a general qq-channel.
When the channel $\mathcal{M}_o$ is a cq-channel in the sense of 
\eqref{MFC},
the next subsection shows that 
the above type of general strategy cannot improve the performance.
\end{remark}

\subsection{Quantum feedback: $\mathbb{A}_{n}$}
\label{subsec:power:A-n}
%Next, we discuss the most general class of strategies, which is denoted by $\mathbb{A}_{n}$. 
Next, we discuss the most general class, 
in which Bob makes quantum feed back to Alice.
To discuss this case, 
we generalize the above formulated problem by allowing more general inputs
because the formulation in the above section allows 
only a pair of states $(\rho,\sigma) \in 
{\mathcal{S}^{A_{o,0}}\times \mathcal{S}^{A_{o,1}}}$.
In the generalized setting, 
the hypothesis $H_i$ is that the true channel is 
$\rho \mapsto \mathcal{M}_o (\Tr_{i\oplus 1}\rho)$.
Then, Alice and Bob collaborate in order to identify which hypothesis is true.
That is, it is formulated as the discrimination between 
two qq-channels $\mathcal{M}$ and 
$\overline{\mathcal{M}}$, mapping $A=A_{o,0} A_{o,1}$ to $B$ defined as
\begin{align}
  \mathcal{M}(\rho)      
  := \mathcal{M}_o (\Tr_{1}\rho)
= \Tr_{1}(\mathcal{M}_o\otimes \mathcal{M}_o)(\rho), \\
    \OL{\mathcal{M}}(\rho) 
    := \mathcal{M}_o (\Tr_{0}\rho)
= \Tr_{0}(\mathcal{M}_o\otimes \mathcal{M}_o)(\rho), 
\end{align}
for $\rho \in \mathcal{S}^{A_{o,0} A_{o,1}}$.
For this discrimination, 
Alice has the input quantum composite system $A_{o,0} A_{o,1}$ and her own quantum memory.
She makes the input state on this system by using 
the quantum feedback and her own quantum memory.
In each step, Bob makes measurement, and sends back a part of the resultant quantum system to Alice while the remaining part is kept in his local quantum memory.
Therefore, this general strategy contains the strategy presented in 
Remark \ref{rem}.

Then, the problem can be regarded as a special case of discriminating the two qq-channels
given in Section \ref{S-5}.
That is, Bob's operation is given as the strategy of discriminating the two qq-channels.

In the most general class, we denote 
the generalized Chernoff and Hoeffding quantities as 
$C^{\mathbb{A}}(a,b|\mathcal{M}_o)$ and 
$B_{e}^{\mathbb{A}}(r|\mathcal{M}_o)$, respectively. 
When no feedback nor no quantum memory of Alice side is allowed, 
we denote the generalized Chernoff and Hoeffding quantities as 
${C}^{\mathbb{P}^{0}}(a,b|\mathcal{M}_o)$ and 
${B}_{e}^{\mathbb{P}^{0}}(r|\mathcal{M}_o)$, respectively.
As a corollary of Theorems \ref{en-br-dp3} and \ref{MNT},
we have the following corollary.
\medskip
\begin{corollary}
\label{en-br-dp}
Assume that the qq-channel $\mathcal{M}_o$ has the form 
\begin{align}
\mathcal{M}_o(\rho)= \sum_{x}  \rho_x \Tr E_x \rho
\label{MST}
\end{align}
where $\{E_x\}_{x \in \mathcal{X}}$ is a PVM and the rank of $E_x$ is one.
For $0\leq r\leq D(\mathcal{M}_o)$ [see Eq. \eqref{cq-power}] and real $a$ and $b$ with 
$-D(\mathcal{M}_o)\le a-b \le D(\mathcal{M}_o)$, 
the following holds
\begin{align}
  C^{\mathbb{A}}(a,b|\mathcal{M}_o)
     % &= C^{\mathbb{A}^{c,0}}(a,b|\mathcal{M}_o) 
       &= C^{\mathbb{P}^{0}}(a,b|\mathcal{M}_o)= {C}^{\bar{\mathbb{P}}^{0}}(a,b|\mathcal{M}_o), \label{MB1}\\
  B_{e}^{\mathbb{A}}(r|\mathcal{M}_o)
%      &= B_{e}^{\mathbb{A}^{c,0}}(r|\mathcal{M}_o) 
   &   
    = B_{e}^{\mathbb{P}^{0}}(r|\mathcal{M}_o)
 = {B}_{e}^{\bar{\mathbb{P}}^{0}}(r|\mathcal{M}_o).\label{MB2}
\end{align}
\hfill $\square$ % Followed by proof!
\end{corollary}
This corollary shows that 
the above extension of our strategy does not improve
the generalized Chernoff and Hoeffding quantities
under the condition \eqref{MST}.

\begin{proof}
Theorem \ref{en-br-dp3} implies 
the first equations in \eqref{MB1} and \eqref{MB2}
under the condition \eqref{MST}.

When the condition \eqref{MST} holds,
any input state on $A=A_{o,0} A_{o,1}$
can be simulated on a separable input state on $A=A_{o,0} A_{o,1}$.
Such a separable input state can be considered as 
a probabilistic input with the form 
$(\rho,\sigma) \in {\mathcal{S}^{A_{o,0}}\times \mathcal{S}^{A_{o,1}}}$.
The strategy $\bar{\mathbb{A}}^{c,0}$ in the problem setting of Section \ref{MO1} contains such probabilistic input.
Hence, Theorem \ref{MNT} implies the second equations in 
the first equations in \eqref{MB1} and \eqref{MB2}.
\end{proof}

\medskip
\begin{remark}
The above result states that the optimal error rates for discrimination with 
a quantum channel can be achieved by i.i.d. state pairs 
$(\rho^{\otimes n},\sigma^{\otimes n})$, among all strategies
without quantum memory at the sender's side.
On the other hand, when entangled state inputs are allowed, 
we could only show the optimality of
non-adaptive tensor-product strategy $\mathbb{P}_{n}^{0}$
for entanglement-breaking channel of the form \eqref{MST}.
The same conclusion holds for the Chernoff bound and Stein's lemma.
\hfill $\square$
\end{remark}

\subsection{Examples}
\label{subsec:power-examples}
In this subsection we derive the generalized Chernoff and Hoeffding bounds for 
three qubit channels, namely, we study the discrimination power of depolarizing, Pauli and 
amplitude damping channels. In each case, the key is identifying the structure of the
output states of each channel by employing the lessons learned in \cite{904522}. 
Here we briefly summarize the basics. A quantum state $\rho$ in two-level systems 
can be parametrized as $\rho=\frac{1}{2}(I+\vec{r}\cdot\vec{\sigma})$, where 
$\vec{r}=(r_{x},r_{y},r_{z})\in\mathbb{R}^{3}$ is the Bloch vector which satisfies 
$r_{x}^{2}+r_{y}^{2}+r_{z}^{2}\leq 1$ and $\vec{\sigma}$ denotes the vector of 
Pauli matrices $\{\sigma_{x},\sigma_{y},\sigma_{z}\}$ such that 
$\vec{r}\cdot\vec{\sigma} \coloneqq r_{x}\sigma_{x}+r_{y}\sigma_{y}+r_{z}\sigma_{z}$. 
Any cptp map $\mathcal{M}_o$ on qubits can be represented as follows:
\begin{align*}
  \mathcal{M}_o\left( \frac{1}{2}(I+\vec{r}\cdot\vec{\sigma}) \right) 
            = \frac{1}{2}\left(I+(\vec{t}+T\vec{r})\cdot\vec{\sigma}\right),
\end{align*} 
where $\vec{t}$ is a vector and $T$ is a real $3\times3$ matrix. For each channel, 
we first need to identify these parameters. The following lemma comes in handy in 
simplifying the optimization problem.

\medskip
\begin{lemma}[{Cf.~\cite[Thm.~3.10.11]{convexbook}}]
\label{convexextreme}
A continuous convex function $f$ on a compact convex set attains its global maximum at an 
extreme point of its domain.
\hfill $\square$
\end{lemma}

\medskip
\begin{lemma}
\label{convex-pure}
For any quantum channel $\mathcal{M}_o$ we have
\begin{align*}
\sup_{\rho,\sigma}D_{\alpha}(\mathcal{M}_o(\rho)\|\mathcal{M}_o(\sigma))
=\sup_{\ket{\psi},\ket{\phi}}D_{\alpha}(\mathcal{M}_o(\ketbra{\psi})\|\mathcal{M}_o(\ketbra{\phi})),
\end{align*}
that is, pure states are sufficient for the maximization of the R\'{e}nyi 
divergence with channel $\mathcal{M}_o$.
%\hfill $\square$ % Followed by proof!
\end{lemma}

\begin{proof}
This is a consequence of Lemma \ref{convexextreme}. Note that the space of quantum states is a convex set; on the other hand, the R\'{e}nyi divergence is a convex function, and we actually need convexity separately in each argument. Therefore the optimal states are extreme points of the set, i.e. pure states.
\end{proof}

\medskip
\begin{remark}
Since we will focus on $2$-level systems, we should recall that a special property of the convex set of qubits which is not shared by $n$-level systems with $n\geq3$ is that every boundary point of the set is an extreme point. Since the states on the surface of the Bloch sphere are mapped onto the states on the surface of the ellipsoid, the global maximum will be achieved by a pair of states on the surface of the output ellipsoid.
\hfill $\square$
\end{remark}

\medskip
\begin{remark}
\label{choosing}
In the following we will use symmetric properties 
of the states in the Bloch sphere to calculate the R\'{e}nyi divergence. 
Note that the R\'{e}nyi divergence of two qubit states is not just a function of their Bloch sphere distance. 
For instance, for two states $\rho_{1}$ and $\sigma_{1}$ with Bloch vectors 
$\vec{r}_{1}=(0,0,1/4)$ and $\vec{s}_{1}=(0,0,-1/4)$, respectively, 
we can see that $\|\rho_{1}-\sigma_{1}\|_{1}=\|\vec{r}_{1}-\vec{s}_{1}\|_{2}=1/2$ 
and the divergence equals $0.17$ ($\alpha\rightarrow 1$). 
On the other hand, for states $\rho_{2}$ and $\sigma_{2}$ with Bloch vectors 
$\vec{r}_{2}=(0,0,1)$ and $\vec{s}_{2}=(1,0,0)$, respectively, 
we can see that $\|\rho_{2}-\sigma_{2}\|_{1}=\|\vec{r}_{2}-\vec{s}_{2}\|_{2}=\sqrt{2}$ 
and the divergence equals $0$ ($\alpha\rightarrow 1$). 
However, we will see that for states with certain symmetric properties, 
the R\'{e}nyi divergence increases with the distance between two arguments.
\hfill $\square$
\end{remark}

\medskip
\begin{example}[Depolarizing channel] 
\label{exdepol}
For $0\leq q\leq1$, the depolarizing channel is defined as follows:
\begin{align*}
  \mathcal{D}_q : \rho \mapsto (1-q)\rho+q \frac{I}{2},
\end{align*}
that is, the depolarizing channel transmits the state with probability 
$(1-q)$ or replaces it with the maximally mixed state with probability $q$. 
In both generalized Chernoff and Hoeffding exponents, we should be
dealing with two optimizations, one over $(\rho,\sigma)$ 
and the other over $0\leq\alpha\leq1$. We can take the 
supremum over the state pair inside each expression and deal with $\alpha$ next. 
Hence, we start with the supremum of the R\'{e}nyi divergence employing Lemma \ref{convex-pure}.
 
For the depolarizing channel, it can be easily seen that 
\begin{align*}
\vec{t}=\begin{pmatrix}
0\\
0\\
0
\end{pmatrix}
\quad\text{and}\quad
T=
\begin{pmatrix}
1-q & 0 & 0  \\
0 & 1-q & 0  \\
0 & 0 & 1-q  \\
\end{pmatrix}.
\end{align*}
Therefore, the set of output states consists of a sphere of radius $1-q$ centered at 
the origin, i.e. $r_{x}^{2}+r_{y}^{2}+r_{z}^{2}=(1-q)^{2}$. Note that we only consider 
the states on the surface of the output sphere. Because of the symmetry of the problem 
and the fact that divergence is larger on orthogonal states, we can choose any two states 
at the opposite sides of a diameter. Here for simplicity we choose the states corresponding 
to $\vec{r}_{1}=(0,0,1-q)$ and $\vec{r}_{2}=(0,0,-1+q)$ leading to the following states, respectively:
\begin{align}
&\rho'= \left(1-\frac{q}{2}\right)\ketbra{0}{0}+\frac{q}{2} \ketbra{1}{1},\\
&\sigma'= \frac q2\ketbra{0}{0} + \left(1-\frac{q}{2}\right) \ketbra{1}{1}.
\end{align}
Then it can be easily seen that 
\begin{align}
\sup_{\rho,\sigma}D_{\alpha}(\mathcal{M}_o(\rho)\|\mathcal{M}_o(\sigma))
=\frac{1}{\alpha-1}\log Q(q,\alpha),
\end{align}
where $Q(q,\alpha)=(1-\frac q2)^{\alpha}(\frac q2)^{1-\alpha}+(1-\frac q2)^{1-\alpha}(\frac q2)^{\alpha}$. 
By plugging back into the respective equations, we have for $0\leq r\leq-(1-q)\log\frac{q}{2-q}$ and $(1-q)\log\frac{q}{2-q}\leq a-b\leq -(1-q)\log\frac{q}{2-q}$,
\begin{align*}
  C^{\mathbb{A}^{c,0}}(a,b|\mathcal{M}_o)
    &= \sup_{0\le \alpha \le 1} -\log Q(q,\alpha) -\alpha a -(1-\alpha)b,\\
  B_{e}^{\mathbb{A}^{c,0}}(r|\mathcal{M}_o)
    &=\sup_{0\le \alpha \le 1}\frac{\alpha-1}{\alpha}\left(r-\frac{1}{\alpha-1}\log Q(q,\alpha)\right).
\end{align*}
The function $Q(q,\alpha)$ introduced above is important and will also appear in later examples; 
we have
\begin{align*}
 \frac{\partial Q(q,\alpha)}{\partial \alpha}
   &=\left(\ln\frac{q}{2-q}\right)
     \left(\left(\frac{q}{2}\right)^{\alpha}\left(1-\frac{q}{2}\right)^{1-\alpha}
           -\left(\frac{q}{2}\right)^{1-\alpha}\left(1-\frac{q}{2}\right)^{\alpha}\right),\\
 \frac{\partial^{2} Q(q,\alpha)}{\partial \alpha^{2}}
   &= \left(\ln\frac{q}{2-q}\right)^{2}Q(q,\alpha),
\end{align*} 
where $ \frac{\partial}{\partial \alpha}$ and $\frac{\partial^{2}}{\partial \alpha^{2}}$ denote the first and second-order partial derivatives with respect to the variable $\alpha$. It can also be easily checked that $\log\frac{q}{2-q}\leq0$, $0\leq Q(q,\alpha)\leq1$. 

\begin{figure}[ht]
\begin{center}
\includegraphics[width=.8\textwidth]{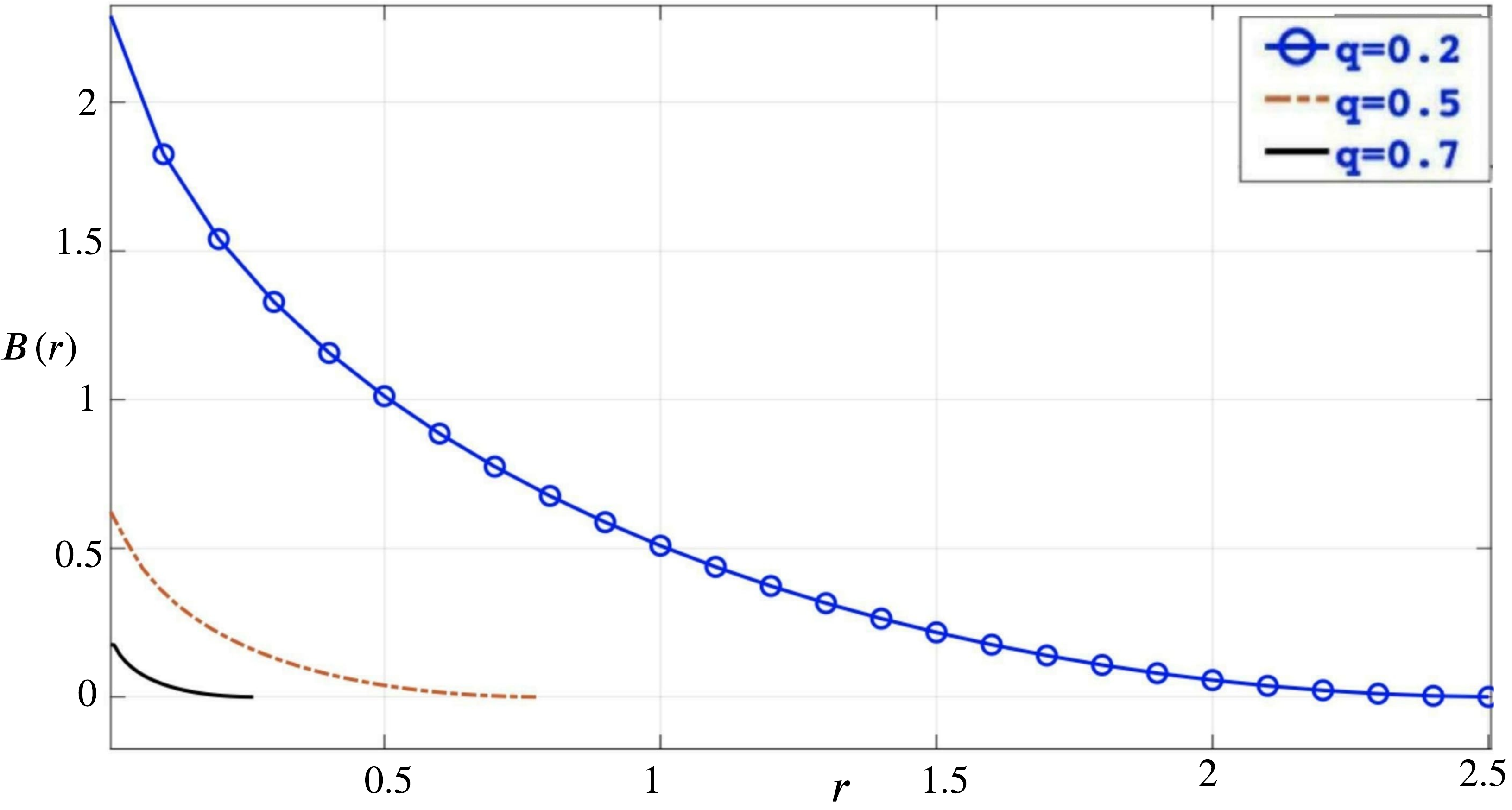}      
\caption{Hoeffding bound for depolarizing channel when entangled inputs are not allowed. (The vertical 
axis shows $B_e^{A_n^{c,0}}(r|\mathcal{M}_o)$, represented as $B(r)$.) The legitimate values of $r$ for each exponent are imposed by Strong Stein's lemma and differ with $q$ as $r=(q-1)\log\frac{q}{2-q}$.}
\label{hff1}
\end{center}
\end{figure}

Let $\overline{C}(\alpha)$ denote the expression inside the supremum in $C^{\mathbb{A}^{c,0}}(a,b|\mathcal{M}_o)$. 
For the generalized Chernoff bound, from the observations above and some algebra, it can be seen that 
\begin{align}
\frac{\partial \overline{C}(\alpha)}{\partial \alpha}=0\implies \alpha=\frac{1}{2}-\frac{\log\frac{\log\frac{q}{2-q}+(a-b)}{\log\frac{q}{2-q}-(a-b)}}{2\log\frac{q}{2-q}}.
\end{align}
On the other hand, it can be checked that $\frac{\partial^{2} \overline{C}(\alpha)}{\partial \alpha^{2}}\geq 0$, making sure that the generalized Chernoff bound is a convex function and also that the above zero is unique. Note that the generalized Chernoff bound is not a monotonic function since $\frac{\partial \overline{C}(\alpha)}{\partial \alpha}$ obviously changes sign, hence the zero is not necessarily at the ends of the interval.

For the Hoeffding exponent $B_{e}^{\mathbb{A}^{c,0}}(r|\mathcal{M}_o)$, 
finding a compact formula for the global maximum is not possible. 
However, numerical simulation guarantees that $B_{e}^{\mathbb{A}^{c,0}}(r|\mathcal{M}_o)$ 
is a convex function that the first derivative has a unique zero. 
We solved the optimization numerically for depolarizing channel with three 
different parameters, see Fig. \ref{hff1}.
\hfill $\square$
\end{example}

\medskip
\begin{example}[Pauli channel]
Let $\vec{p} = (p_{I},p_{x},p_{y},p_{z})$ be a probability vector. 
The Pauli channel is defined as follows:
\begin{align*}
  \mathcal{P}_{\vec{p}} : \rho \mapsto p_{I}\rho+\sum_{i=x,y,z}p_{i}\sigma_{i}\rho\sigma_{i}^\dagger,
\end{align*}
that is, it returns the state with probability $p_{I}$ or applies the Pauli operators 
$\sigma_{x},\sigma_{y},\sigma_{z}$ with probabilities $p_{x},p_{y},p_{z}$, respectively.

For this channel, it can be seen by some algebra that (see e.g. \cite[Sec.~5.3]{Hayashi:book} 
and \cite{PhysRevA.59.3290})
\begin{align}
\vec{t}=\begin{pmatrix}
0\\
0\\
0
\end{pmatrix}
\quad\text{and}\quad
T=
\begin{pmatrix}
p_{I}+p_{x}-p_{y}-p_{z} & 0 & 0  \\
0 & p_{I}-p_{x}+p_{y}-p_{z} & 0  \\
0 & 0 & p_{I}-p_{x}-p_{y}+p_{z}  \\
\end{pmatrix}.
\end{align}
Therefore, the states on the surface of the Bloch sphere are mapped into the surface of the following ellipsoid:
 \begin{align}
 \left(\frac{r_{x}}{p_{I}+p_{X}-p_{Y}-p_{z}}\right)^2 
   + \left(\frac{r_{y}}{p_{I}-p_{X}+p_{Y}-p_{z}}\right)^2 
   + \left(\frac{r_{z}}{p_{I}-p_{X}-p_{Y}+p_{z}}\right)^2 = 1.
 \end{align}
Note that the Pauli channel shrinks the unit sphere with different magnitudes along each axis, and the two states on the surface of the ellipsoid that have the largest distance depends on the lengths of the coordinates on each axis. We need to choose the states along the axis that is shrunk the least. We define the following:
\begin{align}
  p_{\max}=\max\bigl\{ |p_{I}+p_{X}-p_{Y}-p_{Z}|,|p_{I}-p_{X}+p_{Y}-p_{Z}|,|p_{I}-p_{X}-p_{Y}+p_{Z}| \bigr\},
\end{align}
then from the symmetry of the problem and the fact that the eigenvalues of the state 
$\vec{r}=(r_{x},r_{y},r_{z})$ are $\left\{\frac{1-|\vec{r}|}{2},\frac{1+|\vec{r}|}{2}\right\}$, 
the following can be seen after some algebra:
\begin{align}
  \sup_{\rho,\sigma} D_{\alpha}(\mathcal{M}_o(\rho)\|\mathcal{M}_o(\sigma))
       =\frac{1}{\alpha-1}\log Q(1-p_{\max},\alpha).
\end{align}
From this, for $0\leq r\leq -p_{\max}\log\frac{1-p_{\max}}{1+p_{\max}}$ and 
$p_{\max}\log\frac{1-p_{\max}}{1+p_{\max}} \leq a-b \leq -p_{\max}\log\frac{1-p_{\max}}{1+p_{\max}}$, 
we have
\begin{align*}
  C^{\mathbb{A}^{c,0}}(a,b|\mathcal{M})
      &= \sup_{0\le \alpha \le 1} -\log Q(1-p_{\max},\alpha) -\alpha a -(1-\alpha)b,\\
  B_{e}^{\mathbb{A}^{c,0}}(r|\mathcal{M})
      &=\sup_{0\le \alpha \le 1} \frac{\alpha-1}{\alpha}\left(r-\frac{1}{\alpha-1}\log Q(1-p_{\max},\alpha)\right).
\end{align*}
Similar to our findings in Example \ref{exdepol}, we can show that the generalized 
Hoeffding bound is maximized at
\begin{align*}
\alpha=\frac{1}{2}-\frac{\log\frac{\log\frac{1-p_{\text{max}}}{1+p_{\text{max}}}+(a-b)}{\log\frac{1-p_{\text{max}}}{1+p_{\text{max}}}-(a-b)}}{2\log\frac{1-p_{\text{max}}}{1+p_{\text{max}}}},
\end{align*} 
and this point is unique. The same conclusion using numerical optimization indicates that 
the Hoeffding bound of the Pauli channel resembles that of the depolarizing channel. 
Note that a depolarizing channel with parameter $q$ is equivalent to a Pauli channel with 
parameters 
$\left\{p_{I}=1-\frac{3q}{4},p_{x}=\frac{q}{4},p_{y}=\frac{q}{4},p_{z}=\frac{q}{4}\right\}$ \cite[Ex.~5.3]{Hayashi:book}.
%In particular, from the equivalence between a depolarizing channel with parameter $q$ and the Pauli channel with parameters $\{p_{I}=1-3q/4,p_{x}=q/4,p_{y}=q/4,p_{z}=q/4\}$, the three Hoeffding exponents in Fig. \ref{hff1} for $q=0.2,q=0.5,q=0.7$ correspond to the Hoeffding exponent of the Pauli channels with parameters $\{p_{I}=0.85,p_{x}=0.05,p_{y}=0.05,p_{z}=0.05\}$, $\{p_{I}=0.625,p_{x}=0.125,p_{y}=0.125,p_{z}=0.125\}$ and $\{p_{I}=0.475,p_{x}=0.175,p_{y}=0.175,p_{z}=0.175\}$, respectively. 
\hfill $\square$
\end{example}

\medskip
\begin{example}[Amplitude damping channel]
The amplitude damping channel with parameter $0\leq\gamma\leq1$ is defined as follows: 
\begin{align}
  \mathcal{A}_\gamma : \rho \mapsto A_0\rho A_0^\dagger + A_1\rho A_1^\dagger,
\end{align}
where the Kraus operators are given as $A_{0}=\sqrt{\gamma}\ketbra{0}{1}$ and 
$A_{1}=\ketbra{0}{0}+\sqrt{1-\gamma}\ketbra{1}{1}$.

\begin{figure}[ht]
\begin{center}
\includegraphics[width=0.9\textwidth]{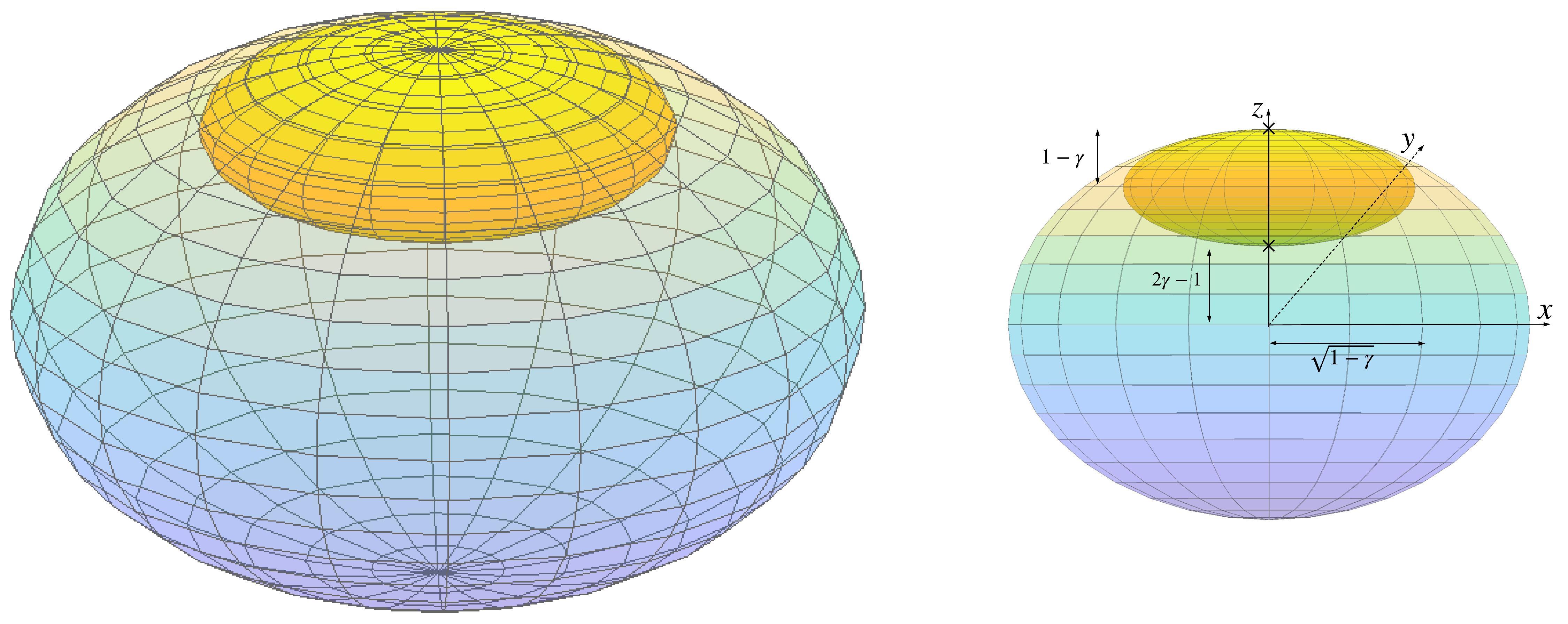}      
\caption{The Bloch sphere and its image under the amplitude damping channel with parameter 
         $\gamma$. There are two large and one small principal axes. As indicated in the right-hand 
         figure, the points $(0,0,1)$ and $(0,0,2\gamma-1)$ are the intersection points of the 
         surface of the displaced ellipsoid with the $z$ axis; the former point also is its 
         intersection point with the Bloch sphere.}
 \label{displaced}
\end{center}
\end{figure}

For this channel, simple algebra shows that
\begin{align}
\vec{t}=\begin{pmatrix}
0\\
0\\
\gamma
\end{pmatrix}
\quad\text{and}\quad
T=
\begin{pmatrix}
\sqrt{1-\gamma} & 0 & 0  \\
0 & \sqrt{1-\gamma} & 0  \\
0 & 0 & 1-\gamma  \\
\end{pmatrix}.
\end{align}

Note that unlike depolarizing and Pauli channels, $\vec{t}$ has a non-zero element for the 
amplitude damping channel, meaning that the amplitude damping channel is not unital. 
The non-zero $\vec{t}$ indicates shifting the center of the ellipsoid. The ellipsoid of 
output states of the amplitude damping channel is depicted in Fig. \ref{displaced}. 
Some algebra reveals the equation of the image to be as follows:
\begin{align}
  \left(\frac{r_{x}}{\sqrt{1-\gamma}}\right)^{2}
    + \left(\frac{r_{y}}{\sqrt{1-\gamma}}\right)^{2}
    + \left(\frac{r_{z}-\gamma}{1-\gamma}\right)^{2}=1.
\end{align} 
To calculate the divergence, from the argument we made in Remark \ref{choosing}, we 
choose the optimal states on $x-z$ plane as $\vec{r}_{1}=(\sqrt{1-\gamma},0,\gamma)$ and 
$\vec{r}_{2}=(-\sqrt{1-\gamma},0,\gamma)$. It can be numerically checked that these points 
lead to maximum divergence. These two points correspond to the following states, respectively:

\begin{align*}
\rho_{1}=\frac{1}{2}
\begin{pmatrix}
 1+\gamma& \sqrt{1-\gamma}   \\
\sqrt{1-\gamma} & 1-\gamma  \\
\end{pmatrix}
\quad\text{and}\quad
\rho_{2}=\frac{1}{2}
\begin{pmatrix}
 1+\gamma& -\sqrt{1-\gamma}   \\
-\sqrt{1-\gamma} & 1-\gamma  \\
\end{pmatrix}.
\end{align*}
Since $|\vec{r}_{1}|=|\vec{r}_{2}|=\sqrt{\gamma^{2}-\gamma+1}$, both states have the following eigenvalues:
\begin{align*}
\lambda_{1},\lambda_{2}=\frac{1\pm\sqrt{\gamma^{2}-\gamma+1}}{2},
\end{align*}
and since $\rho_{1}$ and $\rho_{2}$ obviously do not commute, we find the eigenvectors 
for $\rho_{1}$ and $\rho_{2}$ respectively as follows:
\begin{align*}
 \ket{\nu_{1}}=\frac{1}{\sqrt{1+(\frac{2\lambda_{1}-1-\gamma}{\sqrt{1-\gamma}})^{2}}}
\begin{pmatrix}
1 \\ \frac{2\lambda_{1}-1-\gamma}{\sqrt{1-\gamma}} \\
\end{pmatrix}
,\quad  \ket{\nu_{2}}=\frac{1}{\sqrt{1+(\frac{2\lambda_{2}-1-\gamma}{\sqrt{1-\gamma}})^{2}}}
\begin{pmatrix}
1 \\ \frac{2\lambda_{2}-1-\gamma}{\sqrt{1-\gamma}} \\
\end{pmatrix},
\end{align*}
and
 \begin{align*}
 \ket{\mu_{1}}=\frac{1}{\sqrt{1+(\frac{2\lambda_{1}-1-\gamma}{\sqrt{1-\gamma}})^{2}}}
\begin{pmatrix}
1 \\ -\frac{2\lambda_{1}-1-\gamma}{\sqrt{1-\gamma}} \\
\end{pmatrix}
,\quad  \ket{\mu_{2}}=\frac{1}{\sqrt{1+(\frac{2\lambda_{2}-1-\gamma}{\sqrt{1-\gamma}})^{2}}}
\begin{pmatrix}
1 \\ -\frac{2\lambda_{2}-1-\gamma}{\sqrt{1-\gamma}} \\
\end{pmatrix}.
\end{align*}
The following can be seen after some algebra:
\begin{align*}
\sup_{\rho,\sigma} D_{\alpha}(\mathcal{M}_o(\rho)\|\mathcal{M}_o(\sigma))
=\frac{1}{\alpha-1}\log W(\gamma,\alpha),
\end{align*}
where
\begin{align*}
  W(\gamma,\alpha)
    &= \lambda_{1}\left(\frac{1-\left(\frac{2\lambda_{1}-1-\gamma}{\sqrt{1-\gamma}}\right)^{2}}
                             {1+\left(\frac{2\lambda_{1}-1-\gamma}{\sqrt{1-\gamma}}\right)^{2}}\right)^{2}
       + \lambda_{2}\left(\frac{1-\left(\frac{2\lambda_{2}-1-\gamma}{\sqrt{1-\gamma}}\right)^{2}}
                               {1+\left(\frac{2\lambda_{2}-1-\gamma}{\sqrt{1-\gamma}}\right)^{2}}\right)^{2} \\
    &\phantom{==}
       + \frac{\left(Q(1-\sqrt{\gamma^{2}-\gamma+1},\alpha)\right)
               \left(1-\frac{(2\lambda_{1}-1-\gamma)(2\lambda_{2}-1-\gamma)}{(\sqrt{1-\gamma})^{2}}\right)^{2}}
              {\left(1+\left(\frac{2\lambda_{1}-1-\gamma}{\sqrt{1-\gamma}}\right)^{2}\right)
               \left(1+\left(\frac{2\lambda_{2}-1-\gamma}{\sqrt{1-\gamma}}\right)^{2}\right)}.
\end{align*}
%\lambda_{1}^{\alpha}\lambda_{2}^{1-\alpha}+\lambda_{1}^{1-\alpha}\lambda_{2}^{\alpha}
We also have
\begin{align*}
  D(\mathcal{M}) 
    &= \lambda_{1}\log\lambda_{1}+\lambda_{2}\log\lambda_{2}
      -\lambda_{1}\log\lambda_{1}\left(\frac{1-\left(\frac{2\lambda_{1}-1-\gamma}{\sqrt{1-\gamma}}\right)^{2}}
                                            {1+\left(\frac{2\lambda_{1}-1-\gamma}{\sqrt{1-\gamma}}\right)^{2}}\right)^{2}
      -\lambda_{2}\log\lambda_{2}\left(\frac{1-\left(\frac{2\lambda_{2}-1-\gamma}{\sqrt{1-\gamma}}\right)^{2}}
                                            {1+\left(\frac{2\lambda_{2}-1-\gamma}{\sqrt{1-\gamma}}\right)^{2}}\right)^{2}\\
    &\phantom{=============}
       -\frac{(\lambda_{1}\log\lambda_{2}+\lambda_{2}\log\lambda_{1})
               \left(1-\frac{(2\lambda_{1}-1-\gamma)(2\lambda_{2}-1-\gamma)}{1-\gamma}\right)^{2}}
             {\left(1+\left(\frac{2\lambda_{1}-1-\gamma}{\sqrt{1-\gamma}}\right)^{2}\right)
              \left(1+\left(\frac{2\lambda_{2}-1-\gamma}{\sqrt{1-\gamma}}\right)^{2}\right)}.
\end{align*}

\begin{figure}[ht]
\begin{center}
  \includegraphics[width=0.8\textwidth]{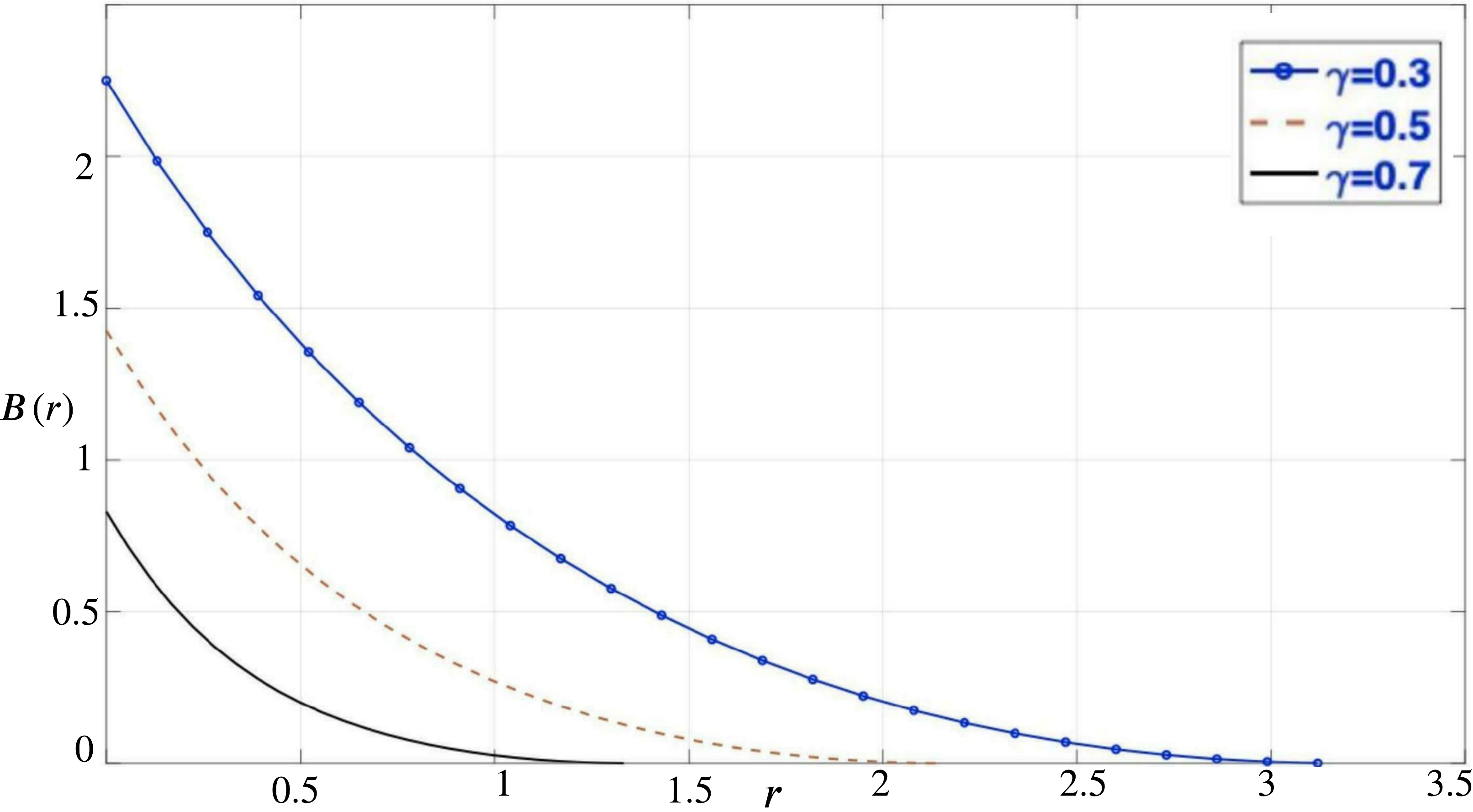}      
\end{center}
\caption{Hoeffding exponent for amplitude damping channel when entangled inputs are 
         not allowed. (The vertical 
axis shows $B_e^{A_n^{c,0}}(r|\mathcal{M}_o)$, represented as $B(r)$.) The legitimate values of $r$ for each exponent are imposed by 
         the strong Stein's lemma and differ as a function of $\gamma$, i.e. $D(\mathcal{M})$.}
\label{hoeffam}
\end{figure}

The cumbersome expressions reflect the complexity of analytically solving the optimizations; 
however, it can be seen numerically that the first derivative of the generalized Chernoff 
bound has a unique zero and its second derivative is positive ensuring the convexity. 
We calculate and plot the Hoeffding exponent for three different parameters of the 
amplitude damping channel in Fig. \ref{hoeffam}.
\hfill $\square$
\end{example}

\section{Conclusion}
\label{conclude}
In an attempt to further extend the classical results of \cite{5165184} to quantum channels, 
we have shown that for the discrimination of a pair of cq-channels, adaptive strategies 
cannot offer any advantage over non-adaptive strategies concerning the asymmetric 
Hoeffding and the symmetric Chernoff problems in the asymptotic limit of error 
exponents, even when the input system is continuous. 
Our approach is to turn the cq-channels into classical channels using eigenvalue
decomposition of the output states 
by using the two distributions introduced by \cite{Hayashi_2002, nussbaum2009}
and subsequently deal with the classical channels.
This latter finding led us to prove the optimality of 
non-adaptive strategies for discriminating qq-channels via a subclass of 
protocols which only use classical feed-forward and product inputs. 

In contrast, the most general strategy for discriminating qq-channels
allows quantum feed-forward and entangled inputs.
In this class, we have obtained two results for a pair of two entanglement-breaking channels. 
When these two entanglement-breaking channels are constructed via the same PVM,
the most general strategy cannot improve the parallel scheme
concerning the asymmetric Hoeffding and the symmetric Chernoff problems.
In contrast, in an example of a pair of entanglement-breaking channels
that are constructed via different PVMs, 
and in another example of a pair of qc-channels implementing general POVMs,
we have shown asymptotic separations between the Chernoff and Hoeffding 
exponents of adaptive and non-adaptive strategies.
These examples show the importance of the above condition for two 
entanglement-breaking channels. For general pairs of qq-channels,
we leave open the question of the condition for the optimality of non-adaptive protocols; 
note that it is open already for entanglement-breaking channels.

We have also studied the hypothesis testing of binary information via a noisy 
quantum channel and have shown that 
when no entangled inputs nor quantum feedback are allowed, 
non-adaptive strategies are optimal.
In addition, when the channel is an entanglement-breaking channel composed of a PVM
followed by a state preparation, 
we have shown the optimality of non-adaptive strategies without the need 
for entangled inputs among all adaptive strategies.

Both our work and quantum machine learning (ML) model in \cite{PhysRevLett.126.190505} aim
at learning about quantum channels.
The reference \cite{PhysRevLett.126.190505} considers a finite number of uses of a quantum channel and establishes quantum advantage
in the sense that classical learning process is exponential while quantum learning process is polynomial
in the number of queries to the quantum channel. These results are conceptually in line with results in \cite{PhysRevA.81.032339,PhysRevLett.103.210501,7541701}
where the number of channel uses leading to perfect identification is relevant. In fact, one particularly
interesting research topic would be to find particular quantum channels such that quantum ML model strictly
outperforms classical ML models in the setting of \cite{PhysRevLett.126.190505}.
We feel that such channels might be found by looking into our Examples \ref{ex:Harrow-et-al} and \ref{ex:POVMs}.
On the other hand, in the context of quantum ML models, 
our results establish 
when and what kind of quantum technology is required to benefit from quantum ML.
In particular, our results in Section \ref{power} are related to the setting of \cite{PhysRevLett.126.190505} in that in both problems there is one single channel inside a black box, and the question is how
well one can predict (or learn about or distinguish between) the outputs of the channel.
While we study the asymptotic advantage,
we expect that our results can be applied to the scenarios where
the scaling of the prediction accuracy with respect to the sample complexity becomes important. 
An example of such scaling relevant to our work is presented in \cite{PhysRevLett.126.190505}. 
More precisely, the so-called scaling property of our results stems from two points:
First, we have shown that, for certain channels, adaptive strategies cannot 
beat non-adaptive ones. We have done so by proving that the R\'{e}nyi
relative entropy in the adaptive setting can be reduced to the R\'{e}nyi relative entropy in the non-adaptive
setting. 
This analysis has been applied to the discrimination of qq-entanglement-breaking
channels where the measurement bases are the same, 
and without considering asymptotic limits.
The relation between the R\'{e}nyi relative entropy and the discrimination task
holds no matter which type of scaling is considered.
Therefore, this reduction is expected to play an important role even in 
the scaling of \cite{PhysRevLett.126.190505}.
Secondly, we have derived a lower bound for discrimination error for 
general non-adaptive strategies.
Since this bound can be applied to any pair of qq-channels,
it led us to the first examples of 
asymptotic separation between adaptive and non-adaptive strategies.
On the other hand, as this lower bound comes from the minimum eigenvalue of a certain operator,
it applies in the non-asymptotic regime, in particular in the scaling of \cite{PhysRevLett.126.190505} as well.

An obvious extension of our work would be to restrict the number of samples (channel uses)
to be a random variable rather than a fixed number (see Li \emph{et al.} \cite{LTT21}).
Besides, the reference \cite{PhysRevLett.126.190505} studies prediction of classical bits even when quantum ML is employed. 
Since our obtained bounds work in the scaling of \cite{PhysRevLett.126.190505},
an extension of our results 
should go beyond the classical data prediction of \cite{PhysRevLett.126.190505}. 

%\vfill

\section*{Acknowledgements}
The authors thank Zbigniew Pucha{\l}a for drawing our attention to 
Example \ref{ex:POVMs}, from Ref. \cite{KPP:POVM},
and Zbigniew Pucha{\l}a and Stefano Pirandola for pointing out various
pertinent previous works.
FS is grateful to Mario Berta for stimulating discussions during QIP 2020. 
He also acknowledges the hospitality and support of the Peng Cheng Laboratory (PCL) 
where the initial part of this work was done. 
AW thanks E. White and D. Colt for spirited conversations on 
issues around discrimination and inequalities. 

FS acknowledges support by the DFG cluster of excellence 2111 (Munich Center for Quantum
Science and Technology).
FS and AW acknowledge partial financial support 
by the Baidu-UAB collaborative project `Learning of Quantum 
Hidden Markov Models', the Spanish MINECO (projects FIS2016-86681-P
and PID2019-107609GB-I00/AEI/10.13039/501100011033) 
with the support of FEDER funds, and the Generalitat de Catalunya
(project 2017-SGR-1127).
FS is also supported by the Catalan Government 001-P-001644 QuantumCAT within the
ERDF Program of Catalunya.
MH is supported in part by Guangdong Provincial Key Laboratory (grant
no. 2019B121203002), a JSPS Grant-in-Aids for Scientific Research (A)
no. 17H01280 and for Scientific Research (B) no. 16KT0017, and 
Kayamori Foundation of Information Science Advancement.

\appendix

\section{Quantum measurements}
\label{A1}
The aim of the appendices is showing Theorem \ref{chernoff}, which discusses the generalized Chernoff bound and 
Hoeffding bound for cq-channel discrimination. 
The appendices are organized as follows: 
We first introduce quantum instruments and provide useful lemmas needed for the rest of this Appendix.
Our adaptive method is proven in Appendix B, and 
subsequently Appendix C proves several auxiliary lemmas
leading to the main result of this Section, which is presented in Appendix D.

Since we need to handle cp-map valued measure,
we prepare the following lemma. % that shows A general cp-map valued measure has the following form.

\begin{lemma}[{Cf.~\cite[Thm.~7.2]{Hayashi:book}}]
\label{LR}
Let $\boldsymbol{\kappa}=\{\kappa_{\omega}:A\rightarrow B\}_{\omega}$ be an instrument 
(i.e. a cp-map valued measure)
with an input system $A$ and an output system $B$. 
Then there exist a POVM $\boldsymbol{M}=\{M_{\omega}\}$ 
on a Hilbert space $A$ and cptp maps ${\kappa}_{\omega}'$
from $A$ to $B$ for each outcome $\omega$, 
such that for any density operator $\rho$,
\begin{align*}
  \phantom{=====================}
  \kappa_{\omega}(\rho) = \kappa_{\omega}'\left(\sqrt{M_{\omega}}\rho\sqrt{M_{\omega}}\right).
  \phantom{====================}
  \square
\end{align*}
%\hfill $\blacksquare$
\end{lemma}

A general POVM can be lifted to a projection-valued measure (PVM), as follows.

\begin{lemma}[{Naimark's theorem \cite{Naimark:POVM}}]
\label{Lnaimark}
Given a positive operated-valued measure (POVM) 
$\boldsymbol{M}=\{M_{\omega}\}_{\omega\in\Omega}$ on $A$ with a discrete 
measure space $\Omega$, there exist a larger Hilbert space $C$ including $A$
and a projection-valued measure (PVM) 
$\boldsymbol{E}=\{E_{\omega}\}_{\omega\in\Omega}$ on $C$ such that 
\begin{align*}
  \phantom{==================:}
  \Tr \rho M_{\omega}=\Tr \rho E_{\omega} \quad\forall\rho\in S^A,\,\omega\in\Omega.
  \phantom{==========}\square
\end{align*}
%\hfill $\blacksquare$
\end{lemma}
%Naimark's theorem states that POVMs can be obtained from PVMs by acting on a larger Hilbert space.

Combining these two lemmas, we have the following corollary.

\begin{corollary}
\label{cor1}
Let $\boldsymbol{\kappa}=\{\kappa_{\omega}:A\rightarrow B\}_{\omega}$ be an instrument 
(i.e. a cp-map valued measure)
with an input system $A$ and an output system $B$.
Then there exist a PVM $\boldsymbol{E}=\{E_{\omega}\}$ 
on a larger Hilbert space $C$ including $A$
and cptp maps ${\kappa}_{\omega}''$
from $C$ to $B$ for each outcome $\omega$, 
such that for any density operator $\rho$,
\begin{align}
  \kappa_{\omega}(\rho) = \kappa_{\omega}''\left(E_{\omega}\rho
  E_{\omega}\right).\label{LNC}
\end{align}
%\hfill $\square$ % Followed by proof!
\end{corollary}

\begin{proof}
First, using Lemma \ref{LR}, we choose 
a POVM $\boldsymbol{M}=\{M_{\omega}\}$ 
on a Hilbert space $A$ and cptp maps $\boldsymbol{\kappa}_{\omega}'$
from $A$ to $B$ for each outcome $\omega$.
Next, using Lemma \ref{Lnaimark}, we choose 
a larger Hilbert space $C$ including $A$
and a projection-valued measure (PVM) 
$\boldsymbol{E}=\{E_{\omega}\}_{\omega\in\Omega}$ on $C$.
We denote the projection from $C$ to $A$ by $P$.
Then, we have 
\begin{align}
(E_\omega P)^\dagger E_\omega P=
PE_\omega P=
M_\omega = \sqrt{M_\omega} \sqrt{M_\omega} 
\end{align}
for any $\omega \in \Omega$.
Thus, there exists a partial isomerty $V_\omega$ from $C$ to $A$ such that
$\sqrt{M_\omega}=V_\omega E_\omega P$.
Hence, we have
\begin{align*}
  \kappa_{\omega}(\rho) 
  = \kappa_{\omega}'\left(\sqrt{M_\omega}\rho \sqrt{M_\omega}\right)
  = \kappa_{\omega}'\left(V_\omega E_{\omega} P \rho P E_{\omega}V_\omega^\dagger \right)
  = \kappa_{\omega}'\left(V_\omega E_{\omega} \rho E_{\omega}V_\omega^\dagger \right).
 \end{align*}
Defining cptp maps ${\kappa}_{\omega}''$
by ${\kappa}_{\omega}''(\rho)=\kappa_{\omega}'(V_\omega  \rho V_\omega^\dagger )$.
This completes the proof.
\end{proof}

\if0
\medskip
\begin{remark}[Relation to general setting with qq-channels]
Here, we discuss how to derive the above setting from 
the general setting presented in introduction for cq-channels.
In the case with cq-channel, the input needs to be
a classical element in the discrete set $\mathcal{X}$.
To decide the classical input, we need to apply measurement 
after the application of the $m$-th cptp map $\mathcal{F}_m$.
That is, we need to replace the $m$-th cptp map $\mathcal{F}_m$
by a quantum instrument
$\mathcal{E}_m: R_m B_m \to K_m R_{m+1}$, which feeds the outcome $K_m$ forward to 
the next channel use. 
Hence, the obtained procedure is equivalent to the procedure given above.

Note however that this way of thinking of a cq-channel as a special type of
qq-channel is restricted to discrete input alphabets; for general, in particular 
continuous input alphabet to the channels $\mathcal{N}$ and $\overline{\mathcal{N}}$, 
we directly use the description above.
\hfill $\square$
\end{remark}
\medskip
\fi

\section{Protocol with PVM for cq-channels}\label{A2}
Now, we rewrite a general adaptive method in a form of a protocol with PVM.
The general procedure for discriminating cq-channels can be rewritten as 
follows using PVMs. 
To start, Fig. \ref{BC} illustrates the general protocol with PVMs, which we shall describe now.
In the following, according to Naimarks dilation theorem, in each $m$-th step, we choose a 
sufficiently large space $B_m$ including the original space $B_m$ such that the measurement is a PVM.

\begin{figure}[ht]
\begin{center}
\label{BC}
\includegraphics[width=0.5\textwidth]{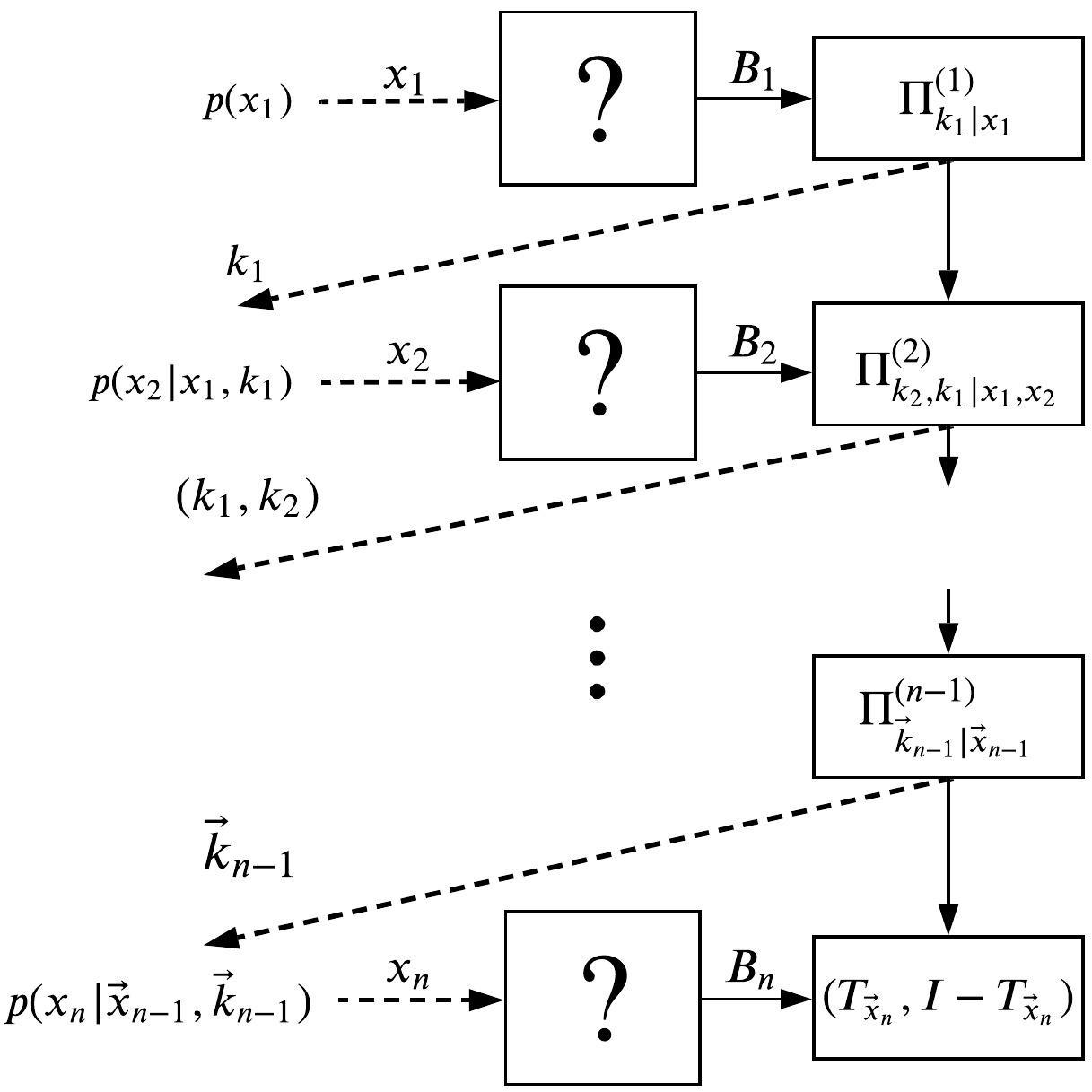}      
\caption{Adaptive strategy for cq-channel discrimination. Solid and dashed lines
         denote flow of classical and quantum information, respectively. The classical outputs of PVMs 
         are employed to decide the inputs adaptively, and leave a post-measurement state 
         that can be accessed together with the next channel output.}
\end{center}
\end{figure}

The first input is chosen subject to the distribution $p_{X_1}(x_1)$. 
Then the output state is measured by a projection-valued measure (PVM)
$\left\{\Pi_{k_1|x_1}^{(1)}\right\}_{k_1}$ on $B_1$.
%where we set the Hilbert space $B$ sufficiently large so that the measurement is projective (see Lemma \ref{Lnaimark}). 
The second input is then chosen according to the distribution 
$p_{X_2|X_1,K_1}(x_2|x_1,k_1)$. 
Then, a PVM
$\left\{\Pi_{k_2,k_1|x_1,x_2}^{(2)}\right\}_{k_{1},k_2}$ is made on $B_1 B_2$,
which satisfies
$\sum_{k_2}\Pi_{k_2,k_1|x_1,x_2}^{(2)} = \Pi_{k_1|x_1}^{(1)}\otimes I$.
The third input is chosen as the distribution 
$p_{X_3|X_1,X_2,K_1,K_2}(x_3|x_1,x_2,k_1,k_2)$, etc. 
Continuing, the $m$-th step is given as follows.
the sender chooses the $m$-th input $x_m$ according to
the conditional distribution $p_{X_m| \vec{X}_{m-1},\vec{K}_{m-1}}(x_{m}|\vec{x}_{m-1}, \vec{k}_{m-1})$. 
The receiver receives the $m$-th output $\rho_{x_m}$ or $\sigma_{x_m}$ on $B_m$.

The description of the remaining processing requires that we distinguish two cases.
\begin{itemize}
  \item For $m<n$, depending on the previous outcomes $\vec{k}_{m-1}=(k_1, \ldots, k_{m-1})$ and
the previous inputs $\vec{x}_m=(x_1, \ldots, x_{m})$,
as the $m$-th projective measurement, the receiver applies a PVM
$\left\{\Pi_{\vec{k}_m|\vec{x}_m}^{(m)}\right\}_{\vec{k}_m}$ on $B_1 B_2\cdots B_m$,
which satisfies the condition
$\sum_{k_m}\Pi_{\vec{k}_m|\vec{x}_m}^{(m)}= 
\Pi_{\vec{k}_{m-1}|\vec{x}_{m-1}}^{(m-1)}\otimes I$.
He sends the outcome $k_m$ to the sender.

  \item For $m=n$, 
dependent on the inputs $\vec{x}_{n}$,
the receiver measures the final state on 
$B_1 B_2\cdots B_n$ with 
the binary POVM 
$(T_{\vec{x}_n},I-T_{\vec{x}_n})$ on $B_1 B_2\cdots B_n$, where
hypothesis $\mathcal{N}$ (resp. $\overline{\mathcal{N}}$) is accepted if and only if
the first (resp. second) outcome clicks.
\end{itemize}

\begin{proposition}\label{PNL}
Any general procedure given in Subsection \ref{S4-A-1}
can be rewritten in the above form.
%\hfill $\square$ % Followed by proof!
\end{proposition}

\begin{proof}
Recall Corollary \ref{cor1} given in Section \ref{S4}.
%Following the general strategy with instrument and classical feed-forward,
Due to Corollary \ref{cor1},
when the Hilbert space $B$ can be chosen sufficiently large,
any state reduction written by a cp-map valued measure
$\{\Gamma_{k_1|x_1}\}_{k_1}$ can also be written as the combination of 
a PVM $\{\Pi_{k_1|x_1}^{(1)}\}_{k_1}$ and a state change by a cptp map 
$\Lambda_{k_1,x_1}$ depending on the measurement outcome $k_1$
such that $\Gamma_{k_1|x_1}(\rho)= 
\Lambda_{k_1,x_1}(\Pi_{k_1|x_1}^{(1)} \rho\Pi_{k_1|x_1}^{(1)})$ for 
$k_1,x_1$. 
Hence, we have 
$\Gamma_{k_1|x_1}(\rho)= \Lambda_{k_1,x_1}(\Pi_{k_1|x_1} \rho\Pi_{k_1|x_1} )$ 
for $k_1,x_1$.

Then, we treat the cptp map $\Lambda_{k_{1},x_{1}}$ as a part of the next measurement.
Let $\{\Gamma_{k_2|x_1,x_2,k_1}\}_{k_2}$ 
be the quantum instrument to describe the second measurement.
We define the quantum instrument 
$\{\overline{\Gamma}_{k_2|x_1,x_2,k_1}\}_{k_2}$ as
$\overline{\Gamma}_{k_2|x_1,x_2,k_1}(\rho):=
{\Gamma}_{k_2|x_1,x_2,k_1}(\Lambda_{k_{1},x_{1}}(\rho))$.
Applying Corollary \ref{cor1} to 
the quantum instrument 
$\{\overline{\Gamma}_{k_2|x_1,x_2,k_1}\}_{k_2}$,
we choose 
the PVM $\{\Pi_{k_2|x_1,x_2,k_1}^{(2)}\}_{k_2}$ on 
$\im \Pi_{k_1|x_1}^{(1)} \otimes B_2$ and the state change by a cptp map 
$\Lambda_{k_1,k_2,x_1,x_2}$ depending on the measurement outcome $k_2$
to satisfy \eqref{LNC}.
Since $\sum_{k_1,k_2}\Pi_{k_2|x_1,x_2,k_1}^{(2)}$ is the identity on $B_1B_2$,
setting $\Pi_{k_1,k_2|x_1,x_2}^{(2)}:=\Pi_{k_2|x_1,x_2,k_1}^{(2)}$,
we define the PVM $\{\Pi_{k_1,k_2|x_1,x_2}^{(2)}\}_{k_1,k_2}$ on $B_1 B_2$.

In the same way, for the $m$-th step, 
using a quantum instrument
$\{\Gamma_{k_m|\vec{x}_m,\vec{k}_{m-1}}\}_{k_m}$,
cptp maps $\Lambda_{\vec{k}_{m-1},\vec{x}_{m-1}}$,
 and
Corollary \ref{cor1}, 
we define 
the PVM $\{\Pi_{k_m|\vec{x}_m,\vec{k}_{m-1}}^{(m)}\}_{k_m}$ 
on $\im \Pi_{\vec{k}_{m-1}|\vec{x}_{m-1}}^{(m-1)} \otimes B_m$
and the state change by cptp maps 
$\Lambda_{\vec{k}_m,\vec{x}_m}$.
Then, setting 
$\Pi_{\vec{k}_m|\vec{x}_m}^{(m)}:= \Pi_{k_m|\vec{x}_m,\vec{k}_{m-1}}^{(m)}$,
we define the PVM $\{\Pi_{\vec{k}_m|\vec{x}_m}^{(m)}\}_{\vec{k}_m}$ on $B_1 B_2\cdots B_m$.

In the $n$-th step, i.e., the final step,
using the binary POVM
$(T_{n|\vec{k}_{n-1},\vec{x}_{n}},I-T_{n|\vec{k}_{t-1},\vec{x}_{n}})$
and
cptp maps $\Lambda_{\vec{k}_{n-1},\vec{x}_{n-1}}$, 
we define
the binary POVM
$(T_{\vec{x}_n},I-T_{\vec{x}_n})$ on $B_1 B_2\cdots B_n$ as follows.
\begin{align}
T_{\vec{x}_n}:=\sum_{\vec{k}_n}
\Lambda_{\vec{k}_{n-1},\vec{x}_{n-1}}^\dagger (T_{n|\vec{k}_{n-1},\vec{x}_{n}}),
%\otimes |\vec{x}_n\rangle \langle \vec{x}_n|,
\end{align}
where $\Lambda_{\vec{k}_{n-1},\vec{x}_{n-1}}^\dagger$ is defined as
$\Tr \Lambda_{\vec{k}_{n-1},\vec{x}_{n-1}}(\rho)X =\Tr \rho \Lambda_{\vec{k}_{n-1},\vec{x}_{n-1}}^\dagger(X)$.
In this way, the general protocol
given in Subsubsection \ref{S4-A-1}
has been converted to a protocol given in this subsection.
\end{proof}

\medskip
It is implicit that the projective measurement $\{\Pi_{\vec{k}_m|\vec{x}_m}^{(m)}\}_{\vec{k}_{m}}$ 
includes first projecting the output from the quantum memory onto a subspace spanned by 
$\{\Pi_{\vec{k}_m-1|\vec{x}_{m-1}}^{(m-1)}\}_{\vec{k}_{m-1}}$,
and then finding $\vec{k}_{m}$ in the entire subspace of $\im \Pi^{(m-1)}_{\vec{k}_{m-1}|\vec{x}_{m-1}}\otimes B_m$.
Hence, $\{\Pi_{\vec{k}_m|\vec{k}_{m-1},\vec{x}_{m-1}}^{(m)}\}_{\vec{k}_m}$ can be regarded as
a PVM on $B_1 B_2\cdots B_m$ and from the construction
\begin{align*}
\sum_{k_{m}}\Pi^{(m)}_{\vec{k}_{m}|\vec{x}_{m}}
=(\Pi^{(1)}_{k_{1}|x_{1}}\otimes I^{\otimes (m-1)})
    \cdots 
 (\Pi_{\vec{k}_{m-1}|\vec{x}_{m-1}}^{(m-1)}\otimes I),
\end{align*}
which shows that the PVMs commute. 

Notice also that 
\begin{align*}
\Pi_{\vec{k}_{n-1}|\vec{x}_{n-1}}^{(n-1)}
\le \Pi_{\vec{k}_{n-2}|\vec{x}_{n-2}}^{(n-2)}\otimes I
\le \cdots \le 
\Pi_{\vec{k}_{1}|\vec{x}_{1}}^{(1)}\otimes I^{\otimes (n-2)}.
\end{align*}

Therefore, 
the states $\rho^{(n)}$ and $\sigma^{(n)}$ before the final measurement, which are defined in 
\eqref{BD1} and \eqref{BD2},
are rewritten as
\begin{align}
\rho^{(n)}=
\sum_{\vec{x}_n,\vec{k}_{n-1}}& p_{X_1}(x_1)\cdots 
p_{{X}_n|\vec{X}_{n-1},\vec{K}_{n-1}}({x}_n|\vec{x}_{n-1},\vec{k}_{n-1})
 \nonumber \\
& \left(\Pi_{\vec{k}_{n-1}|\vec{x}_{n-1}}^{(n-1)}
\big(\rho_{x_1}\otimes \cdots \otimes \rho_{x_n}\big)
\Pi_{\vec{k}_{n-1}|\vec{x}_{n-1}}^{(n-1)}
\otimes \ketbra{\vec{x}_n}\right),\label{BD3}\\
\sigma^{(n)}=
\sum_{\vec{x}_n,\vec{k}_{n-1}} &
p_{X_1}(x_1)\cdots 
p_{{X}_n|\vec{X}_{n-1},\vec{K}_{n-1}}({x}_n|\vec{x}_{n-1},\vec{k}_{n-1})\nonumber \\
& \left(\Pi_{\vec{k}_{n-1}|\vec{x}_{n-1}}^{(n-1)}
\big(\sigma_{x_1}\otimes \cdots \otimes \sigma_{x_n}\big)
\Pi_{\vec{k}_{n-1}|\vec{x}_{n-1}}^{(n-1)}
\otimes \ketbra{\vec{x}_n}\right).
\label{BD4}
\end{align}

\section{Auxiliary results and Proof of Lemma \ref{LPX}}\label{AC}
For our proof of Theorem \ref{chernoff}, we prepare several properties for 
the quantities $C(a,b)$ and 
$B(r)$ defined in \eqref{deff-C} and \eqref{deff}.
The following lemma states the continuity of the $B(r)$ function, of which we give two 
different proofs. The first proof uses the known facts for the case of two states, 
and the cq-channel case is reduced to the former by general statements from convex 
analysis. The second proof is rather more ad-hoc and relies on
peculiarities of the functions at hand.

\begin{lemma}
\label{cont}
The function (Hoeffding expoent) $B(r)$ is continuous in $r$, i.e. for any non-negative real number $r_0$,
\begin{align}
  \lim_{r\to r_0}B(r) = B(r_0).
\label{Eq11}
\end{align}
%\hfill $\square$ % Followed by proof, actually two proofs!
\end{lemma}

\medskip
The combination of Lemma \ref{cont} and the above observation guarantees that 
the map $r \mapsto B(r)-r$ is a continuous and strictly decreasing function from 
$[0, D(\mathcal{N}\|\overline{\mathcal{N}})]$
to 
$[-D(\mathcal{N}\|\overline{\mathcal{N}}), D(\overline{\mathcal{N}}\|\mathcal{N}) ]$.
Hence, we obtain Lemma \ref{LPX}, i.e., 
when real numbers $a,b$ satisfy
$-D(\mathcal{N}\|\overline{\mathcal{N}})\le a-b \le D(\overline{\mathcal{N}}\|\mathcal{N})$,
there exists $r_{a,b} \in [0, D(\mathcal{N}\|\overline{\mathcal{N}})]$
such that $B(r_{a,b})-r_{a,b}= a-b$.
%Hence, we have the following lemma.

\begin{proof}%[Proof of Lemma \ref{cont}]
The crucial difficulty in this lemma is that unlike previous works, 
here we allow that $|\mathcal{X}|$ is infinite. 
Note that in the case of a finite alphabet, we just need to note the role of the channel (as opposed to states): 
it is a supremum over channel inputs $x\in\mathcal{X}$, so a preliminary task is to prove that for a fixed 
                          $x$, i.e. a pair of states $\rho_{x}$ and $\sigma_{x}$, 
the Hoeffding function is continuous. This is already known \cite[Lemma 1]{Ogawa-Hayashi} 
and follows straightforwardly from the convexity and monotonicity of the Hoeffding 
function. After that, the channel's Hoeffding function is the maximum over 
finitely many continuous functions and so continuous. 
However, when the alphabet size is infinite, the supremum of infinitely many continuous 
functions is not necessarily continuous. 
Nevertheless, it inherits the convexity of the functions for each $x$,
cf.~\cite[Cor.~3.2.8]{convexbook}. Since the function is defined on the 
non-negative reals $\mathbb{R}_{\geq 0}$, it is continuous for all $r_0 > 0$,
by the well-known and elementary fact that a convex function on an interval is continuous 
on its interior. It only remains to prove the continuity at $r_0=0$; 
to this end consider swapping null and alternative hypotheses and denote the 
corresponding Hoeffding exponent by $\overline{B}(r)$.
We then find that $\overline{B}(r)$ is the inverse function of $B(r)$. 
Since $\overline{B}(r)$ is continuous even when it is equal to zero, i.e. 
at $r=D(\overline{\mathcal{N}}\|\mathcal{N})$, we conclude $B(r)$ is continuous 
at $r=0$ and $B(0)=D(\overline{\mathcal{N}}\|\mathcal{N})$.
\end{proof}

\begin{lemma}\label{C-H}
When real numbers $a,b$ satisfy  
$-D(\mathcal{N}\|\overline{\mathcal{N}})\le a-b \le D(\overline{\mathcal{N}}\|\mathcal{N})$, 
then we have
\begin{align}
  C(a,b)=r_{a,b}-b=B(r_{a,b})-a.
  \label{Eq12}
\end{align}
%\hfill $\square$ % Followed by proof!
\end{lemma}

\begin{proof}%[Proof of Lemma \ref{C-H}]
Definition of $C(a,b)$, Eq. \eqref{deff-C}, implies that 
$C(a-c,b-c) = C(a,b) + c$.
Hence, it is sufficient to show that for $r \in [0, D(\mathcal{N}\|\overline{\mathcal{N}})]$:
\begin{align}
  C(B(r),r)=0.
  \label{Eq13}
\end{align}
\begin{align*}
C(B(r),r)
&=\sup_{0\leq\alpha\leq 1} (1-\alpha)D_{\alpha}(\mathcal{N}\|\overline{\mathcal{N}})
-\alpha B(r) -(1-\alpha)r\\
&= \sup_{0\le \alpha \le 1}
\alpha
\Big(\frac{\alpha-1}{\alpha}\big(r- D_{\alpha } (\mathcal{N}\|\overline{\mathcal{N}})\big)
- B(r)
\Big)=0,
\end{align*}
where the last equality follows since
$\frac{\alpha-1}{\alpha}\big(r- D_{\alpha } (\mathcal{N}\|\overline{\mathcal{N}})\big)
\le B(r)$ for $0\le \alpha \le 1$.
\end{proof}
\medskip

%If the exponential decreasing rate of $\Tr T_n \sigma^{(n)}$ is $r$, the exponential decreasing rate of $\Tr (I-T_n) \rho^{(n)}$ is not greater than $B(r)$.

Our approach consists of associating suitable classical channels to the
given cq-channels, and noting the lessons learned about adaptive strategy
for discrimination of classical channels in \cite{5165184}. Our proof methodology
however, is also novel for the classical case.
%
%As will be seen, the Nussbaum-Szko\l{}a inequality  \cite{nussbaum2009} cannot be employed in this case and our proof strategy is novel. 
%\mh{Comment: I cannot understand this sentence. In particular, I cannot find what is the Nussbaum-Szko\l{}a inequality.}
The following Lemmas \ref{NS1} and \ref{L1} addresses these matters; 
the former is verified easily and its proof is omitted, and the latter is more involved
and is the key to our developments.

\medskip
\begin{lemma}\label{NS1}
Consider the cq-channels $\mathcal{N}:x\rightarrow \rho_{x}$ and 
$\overline{\mathcal{N}}:x\rightarrow \sigma_{x}$ with input alphabet $\mathcal{X}$ 
and output density operators on Hilbert space $B$. 
Let the eigenvalue decompositions of the output operators be as follows:
\begin{align}
\label{n1}
\rho_x=\sum_{i}\lambda_i^x \ketbra{u_i^x},\\
\label{n2}
\sigma_x=\sum_{j}\mu_j^x \ketbra{v_j^x}.
\end{align}
According to \cite{Hayashi_2002,nussbaum2009}, we define two distributions
\begin{align}
\Gamma_x(i,j)\coloneqq \lambda_i^x \left|\bra{v_j^x}\ket{u_i^x}\right|^2,\\
\overline{\Gamma}_x(i,j)\coloneqq \mu_i^x \left|\bra{v_j^x}\ket{u_i^x}\right|^2.
\end{align}
\fs{First, note that for all pair of indexes $(i,j)$, we have $\Gamma_x(i,j)\geq0,\overline{\Gamma}_x(i,j)\geq0$ and $\sum_{(i,j)}\Gamma_x(i,j)=\sum_{(i,j)}\overline{\Gamma}_x(i,j)=1$, that is, $\Gamma_x(i,j)=p((i,j)|x)$ and $\overline{\Gamma}_x(i,j)=\bar{p}((i,j)|x)$ form conditional probability distributions on the range $\{(i,j)\}$ of the pairs $(i,j)$.} One can think of 
$\Gamma$ and $\overline{\Gamma}$ as classical channels 
form the input system ${\cal X}$ to the output system $\{(i,j)\}_{i,j}$.  
%which produce a pair $(i,j)$ on input $x$ with probabilities $\Gamma_x(i,j)$ and $\overline{\Gamma}_x(i,j)$, respectively.
Second, we have 
\begin{align*}
D_{\alpha } (\rho_x\|\sigma_x)
=
D_{\alpha } (\Gamma_x\|\overline{\Gamma}_{x}),
\end{align*}
which implies [see Eq. (\ref{deff})]
\begin{align}
  \phantom{==}
  B(r) = \sup_x \sup_{0\le \alpha \le 1}
                \frac{\alpha-1}{\alpha}\big(r-D_{\alpha } (\Gamma_x\|\overline{\Gamma}_x)\big) \label{K1}.
  \phantom{==}\square
\end{align}
%\hfill $\square$
\end{lemma}

Note that any extensions of the operators $\{\rho_{x},\sigma_{x}\}$ (not just i.i.d.) 
correspond to the classical extensions by distributions $\Gamma_x(i,j)$ and 
$\overline{\Gamma}_x(i,j)$. Define
\begin{align*}
  \Gamma_{\vec{x}_n}^n(\vec{i}_n,\vec{j}_n)
           \coloneqq \Gamma_{x_1}(i_1,j_1) \cdots \Gamma_{x_n}(i_n,j_n) ,\\
  \overline\Gamma_{\vec{x}_n}^n(\vec{i}_n,\vec{j}_n)
           \coloneqq \overline\Gamma_{x_1}(i_1,j_1)\cdots \overline\Gamma_{x_n}(i_n,j_n) .
\end{align*}
Then, we have the following lemma.

\begin{lemma}\label{L1}
The states $\rho^{(n)}$ and $\sigma^{(n)}$ defined in 
\eqref{BD1} and \eqref{BD2} satisfy
\begin{align*}
&E_{a,b,n}^Q
\coloneqq \min_T 2^{an} \Tr (I-T) \rho^{(n)} +2^{bn} \Tr T \sigma^{(n)}
\ge \frac{1}{2} E_{a,b,n}^{C},
\end{align*}
where
\begin{align*}
E_{a,b,n}^C \coloneqq &
\min_{q_{X_1},\ldots, q_{X_n K_{n-1}|\vec{K}_{n-2}\vec{X}_{n-1}\vec{I}_{n-1}\vec{J}_{n-1}}} \\
& \sum_{\vec{x}_n,\vec{j}_n,\vec{i}_n,\vec{k}_{n-1}}
q_{X_1}(x_1) \ldots
q_{X_n K_{n-1}|\vec{K}_{n-2}\vec{X}_{n-1}\vec{I}_{n-1}\vec{J}_{n-1}}
(x_n,k_{n-1}|\vec{k}_{n-2},\vec{x}_{n-1},\vec{i}_{n-1},\vec{j}_{n-1}) \\
&\hspace{2cm}\cdot \min\left\{2^{an}\Gamma_{\vec{x}_n}^n(\vec{i}_n,\vec{j}_n),
                              2^{bn}\overline{\Gamma}_{\vec{x}_n}^n(\vec{i}_n,\vec{j}_n)\right\}.
\end{align*}
\end{lemma}

\begin{proof}
Let
\begin{align*}
\ket{u_{\vec{i}_n}^{\vec{x}_n}}
\coloneqq \ket{u_{i_1}^{x_1},\ldots, u_{{i}_n}^{{x}_n}},~
\lambda_{\vec{i}_n}^{\vec{x}_n}
\coloneqq \lambda_{i_1}^{x_1},\cdots \lambda_{{i}_n}^{{x}_n},\\
\ket{v_{\vec{j}_n}^{\vec{x}_n}}
\coloneqq \ket{v_{j_1}^{x_1},\ldots, v_{{j}_n}^{{x}_n}},~
\mu_{\vec{j}_n}^{\vec{x}_n}
\coloneqq \mu_{j_1}^{x_1},\cdots \mu_{{j}_n}^{{x}_n}.
\end{align*}
%
%\aw{(Alert: until here the basis of exponential functions was always $2$, now it is $e$! 
%I think it should be $2$ however, or else we need to define all divergences with natural 
%logs rather than binary logs.)}
%
%Consider $\min_{T}e^{an} \Tr (I-T) \rho^{(n)}+e^{bn}\Tr T \sigma^{(n)}$;
Consider $\min_{T} 2^{an} \Tr (I-T) \rho^{(n)} + 2^{bn}\Tr T \sigma^{(n)}$;
it is sufficient to consider $T$ to a projective measurement 
because the minimum can be attained when 
$T$ is a projection onto the subspace that is given as the linear span of 
eigenspaces corresponding to negative eigenvalues of 
%$-e^{an} \rho^{(n)}+e^{bn} \sigma^{(n)}$.
$-2^{an} \rho^{(n)} + 2^{bn} \sigma^{(n)}$.
For a given $\vec{x}_n$, 
the final decision is given as 
the projection $T_{\vec{x}_n}$ on 
the image of the projection $\Pi_{\vec{k}_{n-1}|\vec{x}_{n-1}}^{(n-1)}$ on
$B^{\otimes n}$ depending on $\vec{x}_n$.
Since $\rho^{(n)}$ and $\sigma^{(n)}$ both commute with 
the projection $\Pi_{\vec{k}_{n-1}|\vec{x}_{n-1}}^{(n-1)}$, 
without loss of generality, we can assume that
the projection $T_{\vec{x}_n}$ is also commutative with 
 $\Pi_{\vec{k}_{n-1}|\vec{x}_{n-1}}^{(n-1)}$.
Then, the final decision operator $T_{n}$ is given as the projection
%{\color{red}{do we really need this definition!!}} 
$T_{n}\coloneqq \sum_{\vec{x}_{n}} T_{\vec{x}_n} \otimes |\vec{x}_n\rangle \langle \vec{x}_n| $.

Now, we recall 
the forms \eqref{BD3} and \eqref{BD4} for the states $\rho^{(n)}$ and $\sigma^{(n)}$.
Then, we expand the first term as follows:
\begin{align*}
\Tr &(I-T_n) \rho^{(n)}\\
&=
\sum_{\vec{x}_n,\vec{k}_{n-1}}
\Tr (I-T_{\vec{x}_n})
p_{X_1}(x_1)\cdots p_{{X}_n|\vec{X}_{n-1},\vec{K}_{n-1}}
({x}_n|\vec{x}_{n-1},\vec{k}_{n-1})
\Pi_{\vec{k}_{n-1}|\vec{x}_{n-1}}^{(n-1)}
\big(\rho_{x_1}\otimes \cdots \otimes \rho_{x_n}\big)
\Pi_{\vec{k}_{n-1}|\vec{x}_{n-1}}^{(n-1)}
\\
&=
\sum_{\vec{x}_n,\vec{k}_{n-1}}
\Tr (I-T_{\vec{x}_n})^2
p_{X_1}(x_1)\cdots p_{{X}_n|\vec{X}_{n-1},\vec{K}_{n-1}}
({x}_n|\vec{x}_{n-1},\vec{k}_{n-1})
\Pi_{\vec{k}_{n-1}|\vec{x}_{n-1}}^{(n-1)}
\big(\rho_{x_1}\otimes \cdots \otimes \rho_{x_n}\big)
\Pi_{\vec{k}_{n-1}|\vec{x}_{n-1}}^{(n-1)}
\\
&=
\sum_{\vec{x}_n,\vec{k}_{n-1}}
\Tr 
(I-T_{\vec{x}_n})
\sum_{\vec{j}_n}
\ketbra{v_{\vec{j}_n}^{\vec{x}_n}}
(I-T_{\vec{x}_n})
p_{X_1}(x_1)\cdots p_{{X}_n|\vec{X}_{n-1},\vec{K}_{n-1}}
({x}_n|\vec{x}_{n-1},\vec{k}_{n-1}) \\
&\hspace*{9.5cm}\cdot
\Pi_{\vec{k}_{n-1}|\vec{x}_{n-1}}^{(n-1)}
\big(\rho_{x_1}\otimes \cdots \otimes \rho_{x_n}\big)
\Pi_{\vec{k}_{n-1}|\vec{x}_{n-1}}^{(n-1)} \\
&=
 \sum_{\vec{x}_n,\vec{j}_n,\vec{i}_n,\vec{k}_{n-1}}
p_{X_1}(x_1)\cdots p_{{X}_n|\vec{X}_{n-1},\vec{K}_{n-1}}
({x}_n|\vec{x}_{n-1},\vec{k}_{n-1})
\lambda_{\vec{i}_n}^{\vec{x}_n}
\left|\bra{ u_{\vec{i}_n}^{\vec{x}_n}}
(I-T_{\vec{x}_n})
\Pi_{\vec{k}_{n-1}|\vec{x}_{n-1}}^{(n-1)}\ket{v_{\vec{j}_n}^{\vec{x}_n}}\right|^2,
\end{align*}
%\rangle \langle v_{\vec{j}_n}^{\vec{x}_n}|
where the first line follows from the definition of $T$, the second line is due to the fact that the final measurement can be chosen as a projective measurement, the third line follows because $\sum_{\vec{j}_n}\ketbra{v_{\vec{j}_n}^{\vec{x}_n}}=I^{\otimes n}$ and the last line is simple manipulation.

Similarly, we have
\begin{align*}
\Tr T \sigma^{(n)}
=
 \sum_{\vec{x}_n,\vec{j}_n,\vec{i}_n,\vec{k}_{n-1}}
p_{X_1}(x_1)\cdots p_{{X}_n|\vec{X}_{n-1},\vec{K}_{n-1}}
({x}_n|\vec{x}_{n-1},\vec{k}_{n-1})
\mu_{\vec{j}_n}^{\vec{x}_n}
\left|\bra{ u_{\vec{i}_n}^{\vec{x}_n}}
T_{\vec{x}_n}
\Pi_{\vec{k}_{n-1}|\vec{x}_{n-1}}^{(n-1)}\ket{v_{\vec{j}_n}^{\vec{x}_n}}\right|^2.
\end{align*}

For $m\in[1:n]$, define 
\begin{align*}
&q_{X_m K_{m-1}|\vec{K}_{m-2}\vec{X}_{m-1}\vec{I}_{m-1}\vec{J}_{m-1}}
(x_m,k_{m-1}|\vec{x}_{m-1},\vec{k}_{m-2},\vec{i}_{m-1},\vec{j}_{m-1}) \\
&%\hspace*{2.7cm}
\coloneqq
p_{{X}_m|\vec{X}_{m-1},\vec{K}_{m-1}}
({x}_m|\vec{x}_{m-1},\vec{k}_{m-1})
\frac{\left|\bra{u_{\vec{i}_{m-1}}^{\vec{x}_{m-1}}}\Pi_{\vec{k}_{m-1}|\vec{x}_{m-1}}^{(m-1)}
            \ket{v_{\vec{j}_{m-1}}^{\vec{x}_{m-1}}}\right|^2}
      {\left|\bra{u_{\vec{i}_{m-2}}^{\vec{x}_{m-2}}}\Pi_{\vec{k}_{m-2}|\vec{x}_{m-2}}^{(m-2)}
             \ket{v_{\vec{j}_{m-2}}^{\vec{x}_{m-2}}}\right|^2
            \cdot
       \left|\braket{ u_{{i}_{m-1}}^{{x}_{m-1}}}{v_{{j}_{m-1}}^{{x}_{m-1}}}\right|^2}.
\end{align*}
Hence,
\begin{align*}
\min_{T} &2^{an} \Tr (I-T) \rho^{(n)}+2^{bn} \Tr T \sigma^{(n)} \\
=& \sum_{\vec{x}_n,\vec{j}_n,\vec{i}_n,\vec{k}_{n-1}}
p_{X_1}(x_1)\cdots p_{{X}_n|\vec{X}_{n-1},\vec{K}_{n-1}}
({x}_n|\vec{x}_{n-1},\vec{k}_{n-1}) \\
&\cdot \Big(
2^{an}\lambda_{\vec{j}_n}^{\vec{x}_n}
\left|\bra{u_{\vec{i}_n}^{\vec{x}_n}}
(I-T_{\vec{x}_n})
\Pi_{\vec{k}_{n-1}|\vec{x}_{n-1}}^{(n-1)}
\ket{v_{\vec{j}_n}^{\vec{x}_n}}\right|^2
+
2^{bn}\mu_{\vec{i}_n}^{\vec{x}_n}
\left|\bra{u_{\vec{i}_n}^{\vec{x}_n}}
T_{\vec{x}_n}
\Pi_{\vec{k}_{n-1}|\vec{x}_{n-1}}^{(n-1)}
\ket{v_{\vec{j}_n}^{\vec{x}_n}}\right|^2
\Big) \\
\ge & \sum_{\vec{x}_n,\vec{j}_n,\vec{i}_n,\vec{k}_{n-1}}
p_{X_1}(x_1)\cdots p_{{X}_n|\vec{X}_{n-1},\vec{K}_{n-1}}
({x}_n|\vec{x}_{n-1},\vec{k}_{n-1})
\min\left\{
2^{an}\lambda_{\vec{i}_n}^{\vec{x}_n},
2^{bn}\mu_{\vec{j}_n}^{\vec{x}_n}\right\}
\\
&\cdot \Big(
\left|\bra{u_{\vec{i}_n}^{\vec{x}_n}}
T_{\vec{x}_n}
\Pi_{\vec{k}_{n-1}|\vec{x}_{n-1}}^{(n-1)}
\ket{v_{\vec{j}_n}^{\vec{x}_n}}\right|^2
+
\left|\bra{u_{\vec{i}_n}^{\vec{x}_n}}
(I-T_{\vec{x}_n})
\Pi_{\vec{k}_{n-1}|\vec{x}_{n-1}}^{(n-1)}
\ket{v_{\vec{j}_n}^{\vec{x}_n}}\right|^2
\Big) \\
\stackrel{\text{(a)}}{\ge}
 & \sum_{\vec{x}_n,\vec{j}_n,\vec{i}_n,\vec{k}_{n-1}}
p_{X_1}(x_1)\cdots p_{{X}_n|\vec{X}_{n-1},\vec{K}_{n-1}}
({x}_n|\vec{x}_{n-1},\vec{k}_{n-1})
\min\left\{2^{an}\lambda_{\vec{i}_n}^{\vec{x}_n},
2^{bn}\mu_{\vec{j}_n}^{\vec{x}_n}\right\}
\\
&\cdot \frac{1}{2}
\left|\bra{u_{\vec{i}_n}^{\vec{x}_n}}
\Pi_{\vec{k}_{n-1}|\vec{x}_{n}}^{(n-1)}
\ket{v_{\vec{j}_n}^{\vec{x}_n}}\right|^2
 \\
 =&
\frac{1}{2} \sum_{\vec{x}_n,\vec{j}_n,\vec{i}_n,\vec{k}_{n-1}}
q_{X_1}(x_1) 
\cdots
q_{X_n K_{n-1}|\vec{K}_{n-2}\vec{X}_{n-1}\vec{I}_{n-1}\vec{J}_{n-1}}
(x_n k_{n-1}|\vec{k}_{n-2}\vec{x}_{n-1}\vec{i}_{n-1}\vec{j}_{n-1}) \\
&\cdot \min\left\{2^{an}\Gamma_{\vec{x}_n}(\vec{i}_n,\vec{j}_n),2^{bn}\overline{\Gamma}_{\vec{x}_n}(\vec{i}_n,\vec{j}_n)\right\},
\end{align*}
where (a) follows from the relation
$|\alpha|^2+|\beta|^2\ge \frac{1}{2}|\alpha+\beta|^2$.
\end{proof}

\section{Proof of Theorem \ref{chernoff}}\label{AD}
We are now in a position to show Theorem \ref{chernoff},
which is shown by the combination of Lemma \ref{chernoffL} and Corollary \ref{hoeffding}, which are proven in this Appendix.
%present and prove our main result, the generalized Chernoff bound, as follows:
\begin{lemma}[Generalized Chernoff bound]
\label{chernoffL}
For two cq-channels $\mathcal{N}$ and $\OL{\mathcal{N}}$, and
for real numbers $a,b$ satisfying 
$-D(\mathcal{N}\|\overline{\mathcal{N}})\le a-b \le D(\overline{\mathcal{N}}\|\mathcal{N})$,
\begin{align*}
  C^{\underline{\mathbb{A}}^{c,0}}(a,b|\mathcal{N}\|\overline{\mathcal{N}})
    = C^{\underline{\mathbb{P}}^{0}}(a,b|\mathcal{N}\|\overline{\mathcal{N}})
    = C(a,b) = r_{a,b}-b = B(r_{a,b})-a.
\end{align*}
%\hfill $\square$ % Followed by proof!
\end{lemma}

\begin{proof}%[Proof of Theorem \ref{chernoff}]
For the direct part, i.e. that strategies in $\mathbb{P}^{0}$
achieve this exponent, the following non-adaptive strategy achieves $C(a,b)$.
Consider the transmission of a letter $x$ on every channel use.
Define the test $T_n$ as the projection to the eigenspace of the positive eigenvalues of 
$ 2^{na} \rho_x^{\otimes n} - 2^{nb} \sigma_x^{\otimes n}$.
Audenaert \emph{et al.} \cite{Audenaert_2007} showed that
\begin{align}
2^{na} \Tr [\rho_x^{\otimes n} (I-T_n)]+ 2^{nb} \Tr [\sigma_x^{\otimes n} T_n]
  &\le \inf_{0 \le \alpha \le 1} 
       \Tr (2^{na} \rho_x^{\otimes n})^{\alpha}(2^{nb} \sigma_x^{\otimes n})^{1-\alpha}
       \nonumber \\
  &=   2^{- n \sup_{0 \le \alpha \le 1} \big((1-\alpha)D_{\alpha}(\rho_{x}\|\sigma_{x})
           -\alpha a -(1-\alpha)b \big)}.
\end{align}
Considering the optimization for $x$, we obtain the direct part.

For the converse part, since
\begin{align*}
C^{\underline{\mathbb{A}}^{c,0}}(a,b|\mathcal{N}\|\overline{\mathcal{N}})
=C^{\underline{\mathbb{A}}^{c,0}}(B(r_{a,b}),r_{a,b}|\mathcal{N}\|\overline{\mathcal{N}})
+B(r_{a,b})- a
=C^{\underline{\mathbb{A}}^{c,0}}(B(r_{a,b}),r_{a,b}|\mathcal{N}\|\overline{\mathcal{N}})
+r_{a,b}- b,
\end{align*}
it is sufficient to show $C^{\underline{\mathbb{A}}^{c,0}}(B(r),r|\mathcal{N}\|\overline{\mathcal{N}})\ge 0$
for $r \in [0, D(\mathcal{N}\|\overline{\mathcal{N}})]$.
Observe that 
\begin{align*}
E_{a,b,n}^C=2^{an} \alpha_{n}(\Gamma\|\overline{\Gamma}|\mathcal{T}_{a,b,n})
+2^{bn} \beta_{n}(\Gamma\|\overline{\Gamma}|\mathcal{T}_{a,b,n}),
\end{align*}
where we let $\mathcal{T}_{a,b,n}$ be the optimal test to achieve 
$E_{a,b,n}^C$. We choose $a=B(r)$ and $b=r$ in Lemma \ref{L1}.
\if0
Note the following distributions
\begin{align}
q_{X_1}(x_{1}),\cdots, 
q_{X_n K_{n-1}|\vec{K}_{n-2}\vec{X}_{n-1}\vec{I}_{n-1}\vec{J}_{n-1}}(x_{n},k_{n-1}|\vec{k}_{n-2},\vec{x}_{n-1},\vec{i}_{n-1},\vec{j}_{n-1}).
\end{align}
\fi
The combination of \eqref{K1} and 
\cite[Eq.~(16)]{5165184} guarantees that
\begin{align}
  B_e^{\underline{\mathbb{A}}^{c,0}}(r|\Gamma\|\overline{\Gamma}) = B(r).
\label{EE7}
\end{align}
Notice that the analysis in \cite{5165184} 
does not assume any condition on the set $\mathcal{X}$.
When 
\begin{align*}
\liminf_{n\to \infty}\frac{1}{n}\log 
2^{r n} \beta_{n}(\Gamma\|\overline{\Gamma}|\mathcal{T}_{B(r),r,n})< 0,
\end{align*}
then Eq. \eqref{EE7} implies
\begin{align*}
\liminf_{n\to \infty}\frac{1}{n}\log 
2^{B(r) n} \alpha_{n}(\Gamma\|\overline{\Gamma}|\mathcal{T}_{B(r),r,n})\ge 0.
\end{align*}

Hence, we have
\begin{equation}\begin{split}
\label{E9}
 \liminf_{n\to \infty}\frac{1}{n}\log E_{B(r),r,n}^C 
         & = \max \left\{\liminf_{n\to\infty}
                         \frac{1}{n}\log 2^{rn}\beta_{n}(\Gamma\|\overline{\Gamma}|\mathcal{T}_{B(r),r,n}) \right. , \\
         &\phantom{=====}
                  \left. \liminf_{n\to \infty}
                         \frac{1}{n}\log 2^{B(r)n}\alpha_{n}(\Gamma\|\overline{\Gamma}|\mathcal{T}_{B(r),r,n}) \right\}
          \ge 0.
\end{split}\end{equation}
Therefore, the combination of Lemma \ref{L1} and \eqref{E9}
implies that
\begin{align}
C^{\underline{\mathbb{A}}^{c,0}}(B(r),r|\mathcal{N}\|\overline{\mathcal{N}})
=\liminf_{n\to \infty}\frac{1}{n}\log E_{B(r),r,n}^Q\ge 0.
\label{E10}
\end{align}
This completes the proof.
\end{proof}

\medskip
As corollary, we obtain the Hoeffding exponent.

\medskip
\begin{corollary}[Hoeffding bound]
\label{hoeffding}
For two cq-channels $\mathcal{N}$ and $\OL{\mathcal{N}}$, and
for any $0\leq r\leq D(\mathcal{N}\|\overline{\mathcal{N}})$,
\begin{align*}
  B_{e}^{\underline{\mathbb{A}}^{c,0}}(r|\mathcal{N}\|\overline{\mathcal{N}}) = 
  B_{e}^{\underline{\mathbb{P}}^{0}}(r|\mathcal{N}\|\overline{\mathcal{N}}) = 
  B(r).
\end{align*}
%\hfill $\square$ % Followed by proof!
\end{corollary}

\begin{proof}
%[Proof of Corollary \ref{hoeffding}]
For the direct part, note that a non-adaptive strategy following the Hoeffding bound 
for state discrimination developed in \cite{Hayashi_2007} suffices 
to show the achievability. More precisely, sending the letter $x$ optimizing 
the expression on the right-hand side to every channel use and invoking the result 
by \cite{Hayashi_2007} for state discrimination shows the direct part of the theorem.

For the converse part, note first that from
Theorem \ref{chernoff}, for any $r \in [0,D(\mathcal{N}\|\overline{\mathcal{N}})]$,
\begin{align*}
  C^{\underline{\mathbb{A}}^{c,0}}(B(r),r|\mathcal{N}\|\overline{\mathcal{N}})=0. 
\end{align*}
When a sequence of tests $T_n$ satisfies
$\liminf_{n\to \infty}\frac{1}{n}\log 
\beta_{n}[{\mathcal{N}}\|\overline{\mathcal{N}}|T_{n}]
\le -r_0 < -r_0+\epsilon,$
Eq. \eqref{E10} with $r=r_0-\epsilon $ implies that 
$\liminf_{n\to \infty}\frac{1}{n}\log 
 \alpha_{n}[{\mathcal{N}}\|\overline{\mathcal{N}}|T_{n}]
\ge -B(r_0-\epsilon) $.
Hence, we have 
\begin{align}
B_{e}^{\underline{\mathbb{A}}^{c,0}}(r_0|\mathcal{N}\|\overline{\mathcal{N}})
\le B(r_0-\epsilon) .
\end{align}
Due to Lemma \ref{cont}, 
taking the limit $\epsilon \to 0$ leads to the following inequality
\begin{align}
  B_{e}^{\underline{\mathbb{A}}^{c,0}}(r_0|\mathcal{N}\|\overline{\mathcal{N}}) \leq B(r_0).
\end{align}
This completes the proof.
\end{proof}

\bibliographystyle{unsrt}
\bibliography{discrimination.bib}

\begin{thebibliography}{10}

\bibitem{PhysRevLett.126.190505}
Hsin-Yuan Huang, Richard Kueng, and John Preskill.
\newblock Information-theoretic bounds on quantum advantage in machine learning.
\newblock {\em Phys. Rev. Lett.}, 126:190505, May 2021.

\bibitem{chernoff1952}
Herman Chernoff.
\newblock A measure of asymptotic efficiency for tests of a hypothesis based on the sum of observations.
\newblock {\em The Annals of Mathematical Statistics}, 23(4):493--507, Dec 1952.

\bibitem{hoeffding1965}
Wassily Hoeffding.
\newblock Asymptotically optimal tests for multinomial distributions.
\newblock {\em The Annals of Mathematical Statistics}, 36(2):369--401, Apr 1965.

\bibitem{42188}
{Te Sun} {Han} and Kingo {Kobayashi}.
\newblock The strong converse theorem for hypothesis testing.
\newblock {\em IEEE Transactions on Information Theory}, 35(1):178--180, Jan 1989.

\bibitem{Hiai-Petz}
Fumio Hiai and D{\'e}nes Petz.
\newblock The proper formula for relative entropy and its asymptotics in quantum probability.
\newblock {\em Communications in Mathematical Physics}, 143(1):99--114, 1991.

\bibitem{887855}
Tomohiro {Ogawa} and Hiroshi {Nagaoka}.
\newblock Strong converse and {S}tein's lemma in quantum hypothesis testing.
\newblock {\em IEEE Transactions on Information Theory}, 46(7):2428--2433, Nov 2000.

\bibitem{Audenaert_2007}
Koenraad M.~R. Audenaert, John Calsamiglia, Ramon Mu{\~{n}}oz-Tapia, Emili Bagan, Lluis Masanes, Antonio Acin, and Frank Verstraete.
\newblock {Discriminating States: The Quantum Chernoff Bound}.
\newblock {\em Physical Review Letters}, 98(16):160501, Apr 2007.

\bibitem{nussbaum2009}
Michael Nussbaum and Arleta Szko{\l}a.
\newblock The {C}hernoff lower bound for symmetric quantum hypothesis testing.
\newblock {\em The Annals of Statistics}, 37(2):1040--1057, Apr 2009.

\bibitem{Hayashi:book}
Masahito Hayashi.
\newblock {\em Quantum Information Theory: Mathematical Foundation}.
\newblock Springer Verlag, 2nd edition, 2016.

\bibitem{Ogawa-Hayashi}
T.~Ogawa and M.~Hayashi.
\newblock On error exponents in quantum hypothesis testing.
\newblock {\em IEEE Transactions on Information Theory}, 50(6):1368--1372, 2004.

\bibitem{Hayashi_2007}
Masahito Hayashi.
\newblock Error exponent in asymmetric quantum hypothesis testing and its application to classical-quantum channel coding.
\newblock {\em Physical Review A}, 76(6), Dec 2007.

\bibitem{nagaoka2006converse}
Hiroshi {Nagaoka}.
\newblock {The Converse Part of The Theorem for Quantum Hoeffding Bound}.
\newblock arXiv:quant-ph/0611289, Nov 2006.

\bibitem{Chefles_2007}
Anthony Chefles, Akira Kitagawa, Masahiro Takeoka, Masahide Sasaki, and Jason Twamley.
\newblock Unambiguous discrimination among oracle operators.
\newblock {\em Journal of Physics A: Mathematical and Theoretical}, 40(33):10183--10213, Aug 2007.

\bibitem{PhysRevA.81.032339}
Aram~W. Harrow, Avinatan Hassidim, Debbie~W. Leung, and John Watrous.
\newblock Adaptive versus nonadaptive strategies for quantum channel discrimination.
\newblock {\em Physical Review A}, 81:032339, Mar 2010.

\bibitem{KPP:POVM}
Aleksandra Krawiec, {\L}ukasz Pawela, and Zbigniew Pucha{\l}a.
\newblock Discrimination of {POVMs} with rank-one effects.
\newblock arXiv[quant-ph]:2002.05452, Feb 2020.

\bibitem{5165184}
Masahito {Hayashi}.
\newblock {Discrimination of Two Channels by Adaptive Methods and Its Application to Quantum System}.
\newblock {\em IEEE Transactions on Information Theory}, 55(8):3807--3820, Aug 2009.

\bibitem{berta2018amortized}
Mario {Berta}, Christoph {Hirche}, Eneet {Kaur}, and Mark~M. {Wilde}.
\newblock {Amortized channel divergence for asymptotic quantum channel discrimination}.
\newblock {\em Letters in Mathematical Physics}, 110(8):2277--2336, 2020.

\bibitem{Puchala:measurement}
Zbigniew Pucha{\l}a, {\L}ukasz Pawela, Aleksandra Krawiec, Ryszard Kukulski, and Micha{\l} Oszmaniec.
\newblock Multiple-shot and unambiguous discrimination of von {N}eumann measurements.
\newblock arXiv[quant-ph]:1810.05122v3, Apr 2020.

\bibitem{Lewandowska-et-al:measurement}
Paulina Lewandowska, Aleksandra Krawiec, Ryszard Kukulski, {\L}ukasz Pawela, and Zbigniew Pucha{\l}a.
\newblock Certification of quantum measurements.
\newblock arXiv[quant-ph]:2009.06776, Sep 2020.

\bibitem{PL:teleport}
Stefano Pirandola and Cosmo Lupo.
\newblock {Ultimate Precision of Adaptive Noise Estimation}.
\newblock {\em Physical Review Letters}, 118:100502, Mar 2017.
\newblock arXiv[quant-ph]:1609.02160.

\bibitem{PhysRevLett.103.210501}
Runyao Duan, Yuan Feng, and Mingsheng Ying.
\newblock {Perfect Distinguishability of Quantum Operations}.
\newblock {\em Physical Review Letters}, 103:210501, Nov 2009.

\bibitem{7541701}
Runyao {Duan}, Chen {Guo}, Chi-Kwong {Li}, and Yinan {Li}.
\newblock Parallel distinguishability of quantum operations.
\newblock In {\em 2016 IEEE International Symposium on Information Theory (ISIT)}, pages 2259--2263, July 2016.

\bibitem{cite-key}
Daniel Puzzuoli and John Watrous.
\newblock {Ancilla Dimension in Quantum Channel Discrimination}.
\newblock {\em Annales Henri Poincar{\'e}}, 18(4):1153--1184, 2017.

\bibitem{PhysRevResearch.1.033169}
Xin Wang and Mark~M. Wilde.
\newblock Resource theory of asymmetric distinguishability for quantum channels.
\newblock {\em Physical Review Research}, 1:033169, Dec 2019.

\bibitem{fang2019chain}
Kun {Fang}, Omar {Fawzi}, Renato {Renner}, and David {Sutter}.
\newblock {A chain rule for the quantum relative entropy}.
\newblock arXiv[quant-ph]:1909.05826, Sep 2019.

\bibitem{Cooney2016}
Tom Cooney, Mil{\'a}n Mosonyi, and Mark~M. Wilde.
\newblock {Strong Converse Exponents for a Quantum Channel Discrimination Problem and Quantum-Feedback-Assisted Communication}.
\newblock {\em Communications in Mathematical Physics}, 344(3):797--829, June 2016.

\bibitem{Hirche}
Christoph Hirche, Masahito Hayashi, Emilio Bagan, and John Calsamiglia.
\newblock {Discrimination Power of a Quantum Detector}.
\newblock {\em Physical Review Letters}, 118(16):160502, 2017.

\bibitem{Hayashi_2002}
Masahito Hayashi.
\newblock Optimal sequence of quantum measurements in the sense of stein's lemma in quantum hypothesis testing.
\newblock {\em Journal of Physics A: Mathematical and General}, 35(50):10759, dec 2002.

\bibitem{KretschmannWerner}
Dennis Kretschmann and Reinhard~F. Werner.
\newblock Quantum channels with memory.
\newblock {\em Physical Review A}, 72(6):062323, Dec 2005.

\bibitem{watrous:comb}
Gus Gutoski and John Watrous.
\newblock Toward a general theory of quantum games.
\newblock In {\em Proc. 39th annual ACM Symposium on Theory of Computing (STOC)}, pages 565--574, 2007.

\bibitem{chiribella:memory}
Giulio Chiribella, {G. Mauro} {D'Ariano}, and Paolo Perinotti.
\newblock Memory effects in quantum channel discrimination.
\newblock {\em Physical Review Letters}, 101:180501, Oct 2008.

\bibitem{chiribella:superchannel}
Giulio Chiribella, {G. Mauro} {D'Ariano}, and Paolo Perinotti.
\newblock Transforming quantum operations: Quantum supermaps.
\newblock {\em Europhysics Letters}, 83:30004, Aug 2008.

\bibitem{chiribella:comb}
Giulio Chiribella, {G. Mauro} {D'Ariano}, and Paolo Perinotti.
\newblock Theoretical framework for quantum networks.
\newblock {\em Physical Review A}, 80:022339, Aug 2009.

\bibitem{Watrous:SDP}
John Watrous.
\newblock {Semidefinite Programs for Completely Bounded Norms}.
\newblock {\em Theory of Computing}, 5(11):217--238, Nov 2009.

\bibitem{Gutoski:diamond}
Gus Gutoski.
\newblock On a measure of distance for quantum strategies.
\newblock {\em Journal of Mathematical Physics}, 53(3):032202, Mar 2012.

\bibitem{Pira:survey}
Stefano Pirandola, {Bhaskar Roy} Bardhan, Tobias Gehring, Christian Weedbrook, and Seth Lloyd.
\newblock Advances in photonic quantum sensing.
\newblock {\em Nature Photonics}, 12:724--733, Nov 2018.
\newblock arXiv[quant-ph]:1811.01969.

\bibitem{KW20}
Vishal Katariya and Mark~M. Wilde.
\newblock Evaluating the advantage of adaptive strategies for quantum channel distinguishability.
\newblock {\em Phys. Rev. A}, 104:052406, Nov 2021.

\bibitem{Kitaev}
Alexei Kitaev.
\newblock Quantum computations: algorithms and error correction.
\newblock {\em Russian Mathematical Surveys 52:6 1191-1249}, 52(6):1191--1249, 1997.

\bibitem{AKN}
Dorit Aharonov, Alexei Kitaev, and Noam Nisan.
\newblock Quantum circuits with mixed states.
\newblock In {\em Proc.13th Annual ACM Symposium on Theory of Computation (STOC)}, pages 20--30, 1997.

\bibitem{Paulsen:book}
Vern~I. Paulsen.
\newblock {\em Completely Bounded Maps and Operator Algebras}.
\newblock Cambridge Studies in Advanced Mathematics. Cambridge University Press, 2002.

\bibitem{2017arXiv170501642Y}
Nengkun {Yu} and Li~{Zhou}.
\newblock {Chernoff Bound for Quantum Operations is Faithful}.
\newblock arXiv[quant-ph]:1705.01642, May 2017.

\bibitem{PLLP:lowerbound}
Stefano Pirandola, Riccardo Laurenza, Cosmo Lupo, and Jason~L. Pereira.
\newblock Fundamental limits to quantum channel discrimination.
\newblock {\em npj Quantum Information}, 5:50, June 2019.
\newblock arXiv[quant-ph]:1803.02834.

\bibitem{761271}
Christopher~A. {Fuchs} and Jeroen {van de Graaf}.
\newblock Cryptographic distinguishability measures for quantum-mechanical states.
\newblock {\em IEEE Transactions on Information Theory}, 45(4):1216--1227, 1999.

\bibitem{904522}
Christopher {King} and Mary-Beth {Ruskai}.
\newblock Minimal entropy of states emerging from noisy quantum channels.
\newblock {\em IEEE Transactions on Information Theory}, 47(1):192--209, 2001.

\bibitem{convexbook}
Constantin~P. Niculescu and {Lars-Erik} Persson.
\newblock {\em Convex Functions and Their Applications}.
\newblock Springer Verlag, 2nd edition, 2018.

\bibitem{PhysRevA.59.3290}
Akio Fujiwara and Paul Algoet.
\newblock One-to-one parametrization of quantum channels.
\newblock {\em Physical Review A}, 59(5):3290--3294, May 1999.

\bibitem{LTT21}
Yonglong Li, Vincent Y.~F. Tan, and Marco Tomamichel.
\newblock Optimal adaptive strategies for sequential quantum hypothesis testing.
\newblock arXiv:2104.14706, Apr 2021.

\bibitem{Naimark:POVM}
Mark~A. {Naimark (Neumark)}.
\newblock On a representation of additive operator set functions.
\newblock {\em Doklady Akademii Nauk SSSR -- Comptes Rendus de l'Acad{\'e}mie des Sciences de l'URSS (N.S.)}, 41(9):359--361, 1943.

\end{thebibliography}

\end{document}